\pdfoutput=1
\documentclass[11pt]{article}

\usepackage[margin=1in]{geometry}
\usepackage{amsmath,amssymb,amsthm}
\usepackage{booktabs}
\usepackage{array}  
\usepackage{graphicx}
\usepackage{setspace}
\usepackage{natbib}
\setcitestyle{aysep={}}
\usepackage[colorlinks=true,linkcolor=black,citecolor=black,urlcolor=black]{hyperref}
\usepackage{caption}
\usepackage{subcaption}
\newcommand{\captitle}[1]{{\normalsize\bfseries #1\par}\vspace{1pt}\noindent\ignorespaces}
\usepackage{float}   

\graphicspath{{figures/}}

\newtheorem{proposition}{Proposition}
\newtheorem{lemma}{Lemma}
\newtheorem{corollary}{Corollary}
\theoremstyle{definition}
\newtheorem{definition}{Definition}
\theoremstyle{remark}
\newtheorem{remark}{Remark}

\newcommand{\E}{\mathbb{E}}
\newcommand{\Var}{\mathrm{Var}}
\newcommand{\DN}{\mathbb{D}}
\newcommand{\KRW}{KRW~}

\begin{document}

\title{Preying on Leveraged ETFs\thanks{%
Claude Opus 5 and Fable 5 provided extensive assistance with the coding and analysis underlying the paper's tables, figures, appendix, and the polishing of the manuscript under the author's direction. I thank Daniel Chen, Cam Harvey, Jakub Kastl, Baolian Wang, and Wei Xiong for extremely helpful comments and discussions. All errors are my own.}}
\author{Yinhong Zhao\thanks{%
Princeton University. Corresponding author: \texttt{williamzyh@princeton.edu}.}}
\date{This draft: August 28, 2026\\[4pt]
{\small Preliminary version, updated frequently.
\href{https://papers.ssrn.com/sol3/papers.cfm?abstract_id=7229918}{See the latest version here.}}}
\maketitle

\begin{abstract}
\noindent We argue that speculators preying on the closing rebalances of leveraged exchange-traded funds (LETFs) contributed to the Korean market's extreme volatility in 2026. An LETF's mandated daily rebalance is sized by the day's return, which generates an upward-sloping demand at market close. In response, rational speculators pre-position, enlarge the fund's order, and liquidate into the demand they have induced. Consistent with this mechanism, Korean stocks tracked by LETFs reverse about $75\%$ of their first-day response to pre-open U.S. news by the next close and oscillate for several days thereafter, a pattern absent in every control group. Our quantification implies that self-reinforcing rebalance raised SK~Hynix's annualized volatility from $100.0\%$ to $136.7\%$ over nine weeks and cost its products' predominantly retail holders $17.6\%$ of their initial investment. Dispersing the rebalance across the trading day may backfire, whereas a flexible leverage multiple could help.
\end{abstract}


\clearpage

\section{Introduction}
\label{sec:intro}

The Korean stock market experienced extraordinary volatility in 2026. In the three months after May~27, the KOSPI triggered its market-wide circuit breaker seven times, more than in the preceding twenty-six years.\footnote{The Korea Exchange suspends trading on an $8\%$ fall in the index. The halt was triggered on June~8, 23 and~26 and July~7, 13, 28 and~29, among them a $-10.8\%$ close on July~28. A rally instead triggers a sidecar, which suspends program orders for five minutes when index futures move $5\%$ in either direction. Sidecars fired 31 times over the same window, 16 sell-side and 15 buy-side, more than in any full year before 2026, when the record was 26. The count runs through August~24, 2026.} The index moved 4.1\% on an average day compared to 1.0\% in 2025. The regulators and the public are converging on a single culprit: the sixteen single-stock leveraged exchange-traded funds (LETFs) launched on May~27, 2026.\footnote{Koo Yun-cheol, Deputy Prime Minister and Minister of Economy and Finance of Korea, apologized to the National Assembly on July~29, 2026 for the products' introduction (``I am sorry that we failed to thoroughly anticipate the future repercussions when launching these products'') \citep{asiaeconomy2026}. Lee~Eog-weon, chairman of the Financial Services Commission, apologized separately the same day \citep{reuters2026}.} The funds reference Samsung Electronics and SK~Hynix, which together account for half the index. Within six weeks, all sixteen funds traded below their listing price, the worst down 43\%. Holding 92\% of the funds, retail investors bore these losses.

Standard asset-pricing theory, however, implies that these products could not have caused the volatility. An LETF repackages a levered position in the underlying. A mechanical function of an observed return, its passive daily rebalance conveys no information about fundamentals. The empirical evidence has generally supported this view. Studies of LETF rebalancing in U.S. index and single-stock markets find modest price impact that declines with predictability \citep{ivanovlenkey2018, brogger2021, barbon2021, lenkey2024}, as the theory of anticipated trading implies \citep{admatipfleiderer1991, bessembinder2016}.

This paper reconciles the two views by showing how rational speculators can prey on an LETF, amplifying volatility and extracting substantial wealth from its investors. The mechanism requires only that order flow move prices, as it must when risk-bearing capacity is finite \citep{grossmanmiller1988, duffie2010, gabaixkoijen2021, hartzmarksolomon2025}. A large mandated purchase at the close implies a predictable closing return, allowing speculators to profit by buying earlier in the day. Their trades bring forward the price impact and amplify the required rebalance, causing the closing price to overshoot and subsequently reverse. The LETF thus buys high and sells low in each cycle. This mechanism compounds the familiar ``volatility drag'' of LETFs\footnote{If the market rises 10\% and then falls 10\%, the stock price falls by 1\%. A portfolio levered two-to-one without daily rebalancing would lose 2\%, whereas a double LETF would lose 4\% because it rebalances daily. Volatility drag is the 2-percentage-point extra loss triggered by the LETF rebalance.} and transfers the incremental losses to the counterparties whose trading magnifies the price movements.

An LETF's rebalancing obligation is predictable, inelastic, and, critically, increasing in the price at which it is met. Consider a double-long fund, which means that it must deliver double of its underlying's return \emph{every day}. If the underlying asset rises by one percent, the fund's net asset value rises by two percent while its exposure rises by only one percent. The fund must therefore purchase an additional two percent of its assets before market close to maintain the same leverage multiple for the next trading day. If the underlying falls by one percent, it must sell the same amount. As a result, the fund must submit an upward-sloping rebalance order whose quantity increases with price, an inherently destabilizing feature.

We embed this schedule in a limits-of-arbitrage microstructure model of the trading day following~\citet{delongetal1990}. The market opens at fundamental value, impounding public overnight news such as earnings announcements or foreign-market movements. That news already fixes most of the fund's closing rebalance, so a continuum of rational speculators accumulate positions ahead of it. Their counterparties are risk-averse market makers, who absorb the residual imbalance and warehouse it until the next open. Price impact is what compensates them for bearing that overnight risk. The day ends in a call auction that sets the reference price for the fund, where the fund's mandated rebalance is met by pre-positioned speculators and market makers. Deep capital arrives overnight and corrects part of the mispricing, so the closing displacement itself contributes to the following day's return.

The central result concerns the behavior of rational speculators.
Although each speculator is an atomistic price-taker, speculators collectively build positions during the session in the direction of the fund's predictable rebalance and liquidate those positions at the close. The LETF therefore trades at a price distorted from fundamental value and in a quantity larger than the fundamental return would require. The resulting equilibrium is predatory in aggregate by the definition of \citet{brunnermeierpedersen2005}: speculators profit by inducing another agent's non-discretionary trading. It requires no manipulation or collusion --- the behavior of speculators may be indistinguishable from a benign liquidity provider or a intraday momentum factor trader.

The model quantifies predatory trading through a closed-form \emph{predatory share}: the fraction of the rebalance order induced by the price distortion. The magnitude of this response is governed by a measurable scalar, the \emph{loop gain} $\ell$. It is the marginal price impact of an LETF's rebalance in response to an exogenous stock return and equals the total rebalancing capital of every leveraged product tied to that stock---its \emph{complex}---multiplied by the price impact of the venue that clears the order. When $\ell = 0$, an infinitesimal fund cannot be preyed upon and the market attains the efficient outcome. At the small values prevalent in developed markets, anticipatory capital absorbs the predictable order and stabilizes the market. Preying on a leveraged fund becomes profitable only if $\ell$ is large enough. An LETF becomes prey because of its size relative to the venue that clears its rebalance order, not because of leverage alone.

No market exhibited a comparably large loop gain before Korea introduced single-stock LETFs in 2026. Their listing was the largest in Korean ETF history: combined assets began at \KRW 4.3 trillion and reached \KRW 14 trillion within three weeks. During the post-launch period, the worldwide mandated rebalance in SK~Hynix averaged $22.4\%$ of the stock's total daily traded value and reached $50.4\%$ on the peak day. The resulting loop gains have no U.S. counterpart. Samsung's post-launch average matches that of MicroStrategy, the most extreme U.S. single-stock complex, while SK~Hynix's more than doubles it and runs an order of magnitude above the index complexes behind the benign evidence.

A predatory overshoot at the close should reverse over the next cycle, appreciably so only when $\ell$ is large. Because news itself is permanent, transience in the closing price's response to public news is the mechanism's distinguishing prediction. Instrumenting public news with U.S. market movements that occur before Korea opens, we find no reversal in any comparison groups: U.S. index complexes, U.S. stocks with leveraged ETFs, SK~Hynix and Samsung before their LETF launch, or other Korean stocks after the launch. The reversal is only seen in the treated pair after the LETF launch: on average, $75\%$ of the first-day response reverses by the next close. The reversal itself also overcorrects, so the response oscillates over the next few days. A difference-in-differences design around the Korean launch establishes the causality and shows it also increases conditional volatility of the market.

We then quantify the effect of predatory trading. After calibrating the model to market data and replaying the recovered shock history, we find that SK~Hynix's annualized volatility would have been $100.0\%$, rather than the realized $136.7\%$, in a counterfactual with an infinitesimal LETF complex. The complex therefore accounts for over a quarter of the stock's realized volatility over the episode. Consequently, the mechanism transferred \KRW 1.38 trillion (US\$0.91 billion) away from the SK~Hynix products' own retail holders in nine weeks and cost 17.6\% of their initial investment. That said, the self-reinforcement amplified the path of the July repricing rather than its destination: the level of the pullback itself moves only modestly in the counterfactual.

The same accounting evaluates measures proposed or adopted by the authorities. Korean regulators acted on LETFs' execution schedule, encouraging the funds to disperse their rebalancing trades.\footnote{On July~28, 2026, the Financial Services Commission asked managers to advance and disperse the rebalancing moment; the industry association adopted a version of this guidance two days later.} The policy backfires in our simulation, further inflating the volatility of SK~Hynix and Samsung because dispersal changes only the timing of the mandated rebalance but not its size or predictability. The public news that determines most of the mandate is already known at the open, so speculators can simply front-load their trades there to capture any rebalance trade throughout the day regardless of their timing. Hong~Kong regulators instead acted on order size, granting LETF managers discretion to vary the leverage multiple daily beneath the announced ceiling.\footnote{Revised circular of July~24, 2026 \citep{sfc2026}. Accordingly, the largest cross-listed LETF for SK~Hynix, stock number 7703, changed its name from ``double leveraged'' to ``at most double leveraged'' on August~2, 2026.} We find this may mitigate the self-reference if managers use the discretion appropriately, especially if they could limit the size of rebalance during stress time, but this inevitably generates tracking error against the original fund mandate.

Based on our diagnosis, we propose two measures that could help stabilize the market. The first splits the single rebalance order into twice a day: the opening auction applies to the overnight gap and the closing auction to the session return. This design denies speculators the opportunity to pre-position against public overnight news, cutting holder losses by nearly one half without introducing tracking error. The second design changes the definition of the reference price.\footnote{A senior research fellow at the regulator's research institute independently proposed this approach on July~20, 2026; we provide the first quantitative evaluation of it.} Basing the mandate on a moving average of past prices lets any single displacement move the mark only fractionally, mitigating the predatory incentive and the excess volatility it generates. The fund then tracks a stale benchmark and carries large daily tracking errors, but those errors mean-revert rather than accumulate over longer horizons.

The paper contributes to three strands of literature. First, an extensive literature on institutional investors' mandated rebalancing generally reaches reassuring conclusions. Anticipatory traders supply liquidity to large, predictable index demand rather than exploit it \citep{bessembinder2016, pegorarosammonshim2026, sammonshim2026}; although such flows are inelastic and priced \citep{covalstafford2007, shiveyun2013, petajisto2011, pavlovasikorskaya2023, etula2020, harvey2026}, they have not been associated with financial instability. Evidence on leveraged funds is more mixed, as reviewed in \citet{madhavancheng2026}. Our model unifies the established evidence of late-day amplification \citep{tuzun2013, baibondhatch2015, shum2016, todorov2024} and the skeptical counterevidence \citep{ivanovlenkey2018, brogger2021, barbon2021, lenkey2024} as different regions of a response surface governed by the loop gain. We also contribute to recent work on single-stock ETFs \citep{bessembinder2025, huang2025, murraysammon2026} by proposing a new mechanism based on predatory trading.

Second, we contribute to the literature on predatory trading by documenting a new form of it. Existing models typically feature a prey with a fixed liquidation need that the predator must trigger \citep{brunnermeierpedersen2005, attari2005, carlin2007}. Predatory trading around LETFs instead enlarges the mandated rebalance itself. The mechanism also relates to reference-price manipulation \citep{allengale1992, kumarseppi1992, hubermanstanzl2004}, particularly at the market close \citep{hillionsuominen2004, comertonfordeputnins2014, bogousslavskymuravyev2023}, but the friction in our setting is contractual rather than informational. Related conduct led to a \KRW 2.44 trillion judgment in Korea in 2010 \citep{eom2021}. The conduct studied here is more difficult to police because no single large trader need be responsible for it.

Third, the upward-sloping demand we study in the LETF setting relates to a broad literature on positive feedback trading, first named by \citet{delongetal1990}. In general, predatory trading may arise whenever a trader follows a self-reinforcing investment strategy, as in fire sales \citep{shleifervishny2011}, portfolio insurance \citep{grossmanzhou1996}, and volatility targeting \citep{moreiramuir2017}. Our paper documents a new instance of this old insight in the LETF market and shows that it generated economically significant financial instability in Korea, much as constrained intermediary capital does elsewhere in the literature \citep{hekrishnamurthy2013, brunnermeiersannikov2014}.

The rest of the paper is organized as follows: 
Section~\ref{sec:background} introduces the rebalancing arithmetic and the Korean episode.
Section~\ref{sec:model} develops the model.
Section~\ref{sec:empirics} presents the reduced-form evidence.
Section~\ref{sec:quant} calibrates the model, prices the mechanism, and
evaluates the design counterfactuals. 
Section~\ref{sec:conclusion} concludes.

%

\section{LETF Rebalance Mechanism}
\label{sec:background}

This section introduces the rebalancing mechanism of leveraged exchange-traded funds. The fund's publicly predictable rebalancing trade is determined mechanically by the day's return, executed mostly in the closing auction that sets its reference price, and reached a market-wide scale within weeks of its 2026 introduction in Korea.

\subsection{Rebalance Arithmetic}
\label{sec:arithmetic}

An LETF promises to deliver a constant multiple $L$ of the \emph{daily} return of an underlying security: $+2\times$ for a double-long fund, $-2\times$ for a double-inverse fund. The promise disciplines the fund's portfolio at the start of every trading day, in that a fund with net asset value $A$ must hold exposure $LA$ to the underlying. However, the fund's net asset value moves with the underlying during the trading session, so it must rebalance before the market closes to restore its target multiple for the following day.

For example, consider a $2\times$ fund that begins the day with $200$ of exposure, financed by $100$ of investor capital and $100$ of borrowing. If the underlying declines by $10\%$, the value of the fund's position falls to $180$. After repaying the unchanged $100$ liability, the fund has a net asset value of $80$. Its effective multiple has therefore risen from $2.00$ to $2.25$, and the portfolio no longer satisfies its prospectus mandate. To restore a $2\times$ multiple for the next trading day, the fund must reduce its exposure to $160$ by selling $20$, an amount equal to $20\%$ of its initial net asset value. Following a $10\%$ increase in the underlying, the multiple instead falls to $1.83$, requiring the fund to purchase additional exposure. More generally, an $L\times$ fund with initial net asset value $A$ and an underlying asset return $r$ must trade
\begin{equation}
\underbrace{LA(1+Lr)}_{\text{required exposure}} \;-\;
\underbrace{LA(1+r)}_{\text{current exposure}}
\;=\; A\left(L^2 - L\right) r .
\label{eq:rebalance}
\end{equation}

LETF's rebalance has three important properties that foreshadow our analysis below. First, the mandated trade is always momentum-directional for both leveraged and inverse product. We note that the coefficient $L^2 - L > 0$ for any $L \notin [0,1]$: a $+2\times$ fund has $L^2-L = 2$, while a $-2\times$ fund has $L^2-L = 6$. The case of the inverse funds are counterintuitive: when the price rises, an inverse fund loses net asset value and gains short exposure in magnitude, so both margins oblige it to buy exposure back at a magnitude even larger than its long counterpart. Long and inverse funds therefore accumulate in rebalance order rather than cancel out.

Second, the mandated trade is publicly predictable, with each input of equation~\eqref{eq:rebalance} publicly known by the market close.
For an underlying asset referenced by many funds, we define the total \emph{rebalancing capital} by summing over every leveraged and inverse product referencing the same underlying asset, wherever domiciled and however replicated:
\begin{equation}
K \;\equiv\; \sum_i A_i\left(L_i^2 - L_i\right),
\label{eq:K}
\end{equation}
We auction that set the underlying's \emph{complex}: a grouping by the reference asset rather than by sponsor. The rebalancing capital $K$ can be interpreted as follows: a 1\% move in the underlying obliges the complex to aggregately trade $K \times 1\%$ in the direction of the move at the close. At any moment during the session, any trader can thus compute the size of the rebalance order the complex will submit at market close under the current price, as well as their entire demand schedule under any alternative closing price.

Third, the mandated trade cannot contractually be postponed. The fund is legally obligated to deliver the leverage multiple with minimal tracking error, leaving the managers essentially no discretion over the timing, size and pricing of the trade. In this respect, the fund is a passive and perfectly inelastic price taker.

LETF rebalancing is known to be costly when the underlying price oscillates. Because the fund buys after the underlying rises and sells after it falls, a round trip in the underlying imposes a loss on the fund. For example, if the underlying first rises by 10\% and then falls by 10\%, it ends 1\% below its initial value. A $2\times$ fund, by contrast, delivers $+20\%$ and then $-20\%$, exactly as promised, and therefore ends 4\% below its initial value---rather than 2\% below, as implied by twice the underlying's cumulative return.\footnote{A $-2\times$ fund, which loses 20\% and then gains 20\%, likewise ends 4\% below its initial value.} The additional 2\% loss is known as \emph{volatility drag}, the textbook cost of a daily-reset contract \citep{chengmadhavan2009, avellanedazhang2010}. \citet{murraysammon2026} measure this drag alongside fees and financing costs across the LETF universe, and show that single-stock products are launched on stocks near the top of the volatility distribution, where the drag is largest. This drag is not a deadweight loss, but it accrues to the counterparties that supply liquidity to the fund's rebalancing trades.

Volatility drag therefore creates the basis for predatory trading. We argue that speculators can deliberately amplify price oscillations by trading during the session in the direction of the fund's anticipated rebalance. Supplying liquidity to the fund's rebalance, speculators profit from the increase in volatility drag. The drag is thus no longer solely a byproduct of a daily-reset contract applied to an exogenous price path; it includes an endogenous, predatory component created by the traders who capture it.

\subsection{Rebalance Implementation}
\label{sec:close}

The exact volume of the rebalance is determined endogenously by the outcome of the closing auction, in which the LETF itself participates. Modern equity markets often end a trading day with a single-price batch clearing. All market participants submit limit orders during a pre-close window. The exchange then computes the price that maximizes executable volume as the official close and settles all matched shares at it. This price serves as the reference for a wide variety of purposes including fund net asset values and day-to-day returns.

Accordingly, LETF managers execute the rebalancing trade during or immediately after the closing auction. The standard practice is to submit limit orders and revise them in real time as the indicative price evolves.\footnote{During the closing-auction window, exchanges disseminate at high frequency an indicative price---the price at which the auction would clear if held immediately---to all qualified market participants.} The simultaneity of sizing and execution eliminates measured tracking error regardless of how the order displaces the closing price. LETF's demand schedule at market close is therefore upward-sloping. Section~\ref{sec:model} characterizes the fixed-point equilibrium under this.

That said, the rebalancing order is not always submitted by the fund itself or executed directly in the auction. Many funds, especially the cross-border products, maintain the exposure through over-the-counter total-return swaps with banks or other intermediaries. The intermediary carries the exposure and rebalances on the fund's behalf. Because its liability is determined by the same market close, it faces the same constraint as the fund and transmits the fund's rebalancing demand one-for-one to the closing market. Lemma~\ref{lem:swap} establishes this result formally.

Second, funds may rebalance in venues other than the closing auction, among them the futures market and the over-the-counter markets. LETFs often hold part of their exposure through single-stock futures and rebalance that portion in the futures market, which closes after the stock market by fifteen minutes in Korea. This portion is executed after the reference price has been settled in the closing auction. In addition, when the mandated order exceeds what the closing auction can absorb, the residual is executed in after-hours trading off exchange. Neither displacement alters the size of the order, which remains determined by the closing price. We discuss displacement as a potential policy decision for fund managers in Section~\ref{sec:quant} and formalize it in Appendix~\ref{app:phi}.

\subsection{The Korean Episode}
\label{sec:episode}

Korea entered 2026 as the world's best-performing major equity market. The KOSPI, Korea's principal stock-market index, rose 76\% in 2025, its third-largest annual gain on record. Two flagship semiconductor stocks, Samsung Electronics and SK~Hynix, together represented roughly half of index capitalization and benefited from shortages of memory chips used by artificial intelligence companies. The rally continued into 2026, and the index reached a record high on June~22.

On May~27, 2026, the Korea Exchange listed sixteen single-stock leveraged exchange-traded funds referencing Samsung Electronics and SK~Hynix. The launch followed the regulatory approval of single-stock leveraged funds in April 2026, subject to an eligibility screen that only these two stocks passed. With initial net assets of \KRW 4.3 trillion, the launch was the largest in the history of Korean ETFs; combined net assets reached \KRW 14 trillion within three weeks. Cross-listed products further increased the total: the worldwide SK~Hynix complex peaked at \KRW 37.5 trillion in leveraged assets, while the Samsung complex peaked at \KRW 6.6 trillion and was almost entirely domestic.\footnote{Offshore listings do not disperse the rebalancing pressure. By prospectus, the Hong Kong products rebalance ``at or around the close of trading of the Korea Exchange,'' concentrating their flows in the same ten-minute Seoul auction as the domestic products. Appendix Table~\ref{tab:complex} reports the composition of each complex.}

Unlike the relatively small LETFs observed elsewhere, the Korean funds were large enough to have market-wide effects. Consider June~5, 2026, when SK~Hynix closed 9.9\% below its previous close. Leveraged products referencing the stock held \KRW 21.5tn in net assets that day, implying rebalancing capital of \KRW 43.3tn under equation~\eqref{eq:K}. The $-9.9\%$ return therefore required the complex to sell \KRW 4.3tn of SK~Hynix stock, equivalent to 0.3\% of the firm's market capitalization. However, only 0.8\% of the firm's capitalization trades on an average day, so the required sale represented 36\% of daily trading volume. Moreover, the order had to be executed near the market close, during a brief window that normally accounts for only one-tenth of the day's traded value. 

June~5 was not an outlier. Over the post-launch weeks, the worldwide mandated rebalance averaged 22.4\% of SK~Hynix's entire daily traded value, reaching 50.4\% of it on the peak day. In comparison, the corresponding number for the largest U.S. index complexes is one to two orders of magnitude smaller. An order of this magnitude posed a substantial challenge to the market's capacity to absorb it while creating a significant profit opportunity for speculators with sufficient capital to participate.

Figure~\ref{fig:timeline} plots what followed. The market rose to a record high on June~22 and then collapsed, with July 2026 delivering the largest monthly KOSPI decline on record. Market-wide circuit breakers were repeatedly triggered following both upward and downward movements and a record level of market volatility. The funds' retail holders bore substantial losses: by July~8, all sixteen products launched on May~27 traded below their \KRW 20{,}000 listing price, and the worst performer had lost 43\% in six weeks.

While the coincidence between LETFs and market volatility remains suggestive, the Korean public and authorities have attributed the volatility to the LETFs. The financial authorities froze new single-stock LETF listings on July~16 together with a bundle of new restrictions on participation. The Deputy Prime Minister and Minister of Economy and Finance apologized to the National Assembly for the products' introduction on July~29, as did the chairman of the Financial Services Commission.

These expressions of regret embody a causal claim, but the correlation alone does not establish it. The remainder of the paper formalizes the mechanism, quantifies its strength, and evaluates several policy proposals from Korean and Hong Kong authorities.

\section{Model}
\label{sec:model}

This section develops a model to characterize the equilibrium asset-pricing implications of the self-reinforcing LETF rebalance described in Section~\ref{sec:background}. The model follows the seminal paper of~\citet{delongetal1990} to characterize the response of a continuum of rational, informed, yet atomistic speculators to the LETF trade. Although no individual trader can manipulate prices, speculators may collectively distort the LETF's reference price away from the stock's fundamental value. The model identifies the conditions under which this predatory trading occurs and quantifies its magnitude.

\subsection{Environment}
\label{sec:environment}

There is one risky asset and a trading day of length $T$, the rebalancing interval under the status quo. Returns are measured relative to the previous day's closing price, normalized to zero. We first present a static model of within-day trading and then link successive trading sessions in the dynamic extension of Section~\ref{sec:chaining}.

Over the day the fundamental value of the asset gains
\begin{equation}
v \;=\; \varepsilon + u,
\label{eq:fundamental}
\end{equation}
of which $\varepsilon \sim N(0,\sigma_0^2)$ is \emph{public overnight
news}, realized before the session opens and observed by all agents,
including the market makers; it can include earnings announcements,
macroeconomic news, or foreign market returns. The remainder
$u \sim N(0,\sigma_u^2)$ is \emph{private session news}, arriving
continuously during the session and observed in real time only by the
rational speculators defined below, and it can include investors' demand
for the asset or private information from speculators' own research.

The day itself unfolds in four stages, with their institutional
counterparts in the Korean market indicated in the table below:
\begin{center}
\begin{tabular}{lll}
\toprule
Stage & What happens & Counterpart (KST) \\
\midrule
Open $t=0$ & overnight news already public; speculators & \\
 & \quad pre-position at price $r_o$ & 08:30--09:00 \\
Session $t \in (0,T)$ & speculators and noise traders trade & \\
 & \quad continuously at price $r_s$ & 09:00--15:20 \\
Close $t=T$ & single-price call auction sets $r_c$; & \\
 & \quad the rebalance executes inside the auction & \\
 & \quad that determines its own size & 15:20--15:30 \\
Overnight & outside deep capital arrives & \\
 & \quad fundamental value revealed & \\
\bottomrule
\end{tabular}
\end{center}
We write $g$ for the \emph{fundamental correction}: the price change from the previous close that moves the price toward fundamental value. In a static model with a single trading day, $g = \varepsilon$. Once trading days are chained in Section~\ref{sec:chaining}, $g$ also carries the overnight correction of the error left in the previous close. The market does not open at $g$, however: speculators have a strategic motive to pre-position, and the opening print impounds their positions (equation~\eqref{eq:sessprice}).

Four agents trade across these stages. The first is the LETF itself. An $L\times$ LETF with assets $A$ must trade $A(L^2-L)r$ under a daily return of $r$, as introduced in equation~\eqref{eq:rebalance} of Section~\ref{sec:background}. Aggregating every product keyed to this reference price, the complex's demand at the auction can be written as
\begin{equation}
q_L \;=\; K\, r_c, \qquad K \equiv \sum_i A_i (L_i^2 - L_i),
\label{eq:automaton}
\end{equation}
where $r_c$ is the return that clears the closing auction in which the rebalance executes. This demand is upward-sloping in price and publicly known to other agents.

Products may differ in the venue they hold, with many using a swap contract rather than holding the underlying directly. Nevertheless, the following lemma establishes the swap structure changes neither the equilibrium price nor the behavior behind it, so the model applies to both.
\begin{lemma}[swap pass-through]
\label{lem:swap}
Consider a competitive, risk-averse swap counterparty whose notional
resets each day to $L\times$NAV struck at the reference close, and who
seeks to remain delta-flat. Then the counterparty (i) does not trade
during the session, (ii) must trade exactly $A(L^2-L)r$ at the reset, and
(iii) optimally executes in the underlying at the reference price.
\end{lemma}

The rebalance is large enough that executing it moves the price, which we formalize through a competitive market-making sector following
\citet{grossmanmiller1988}.
A mass $\mu$ of small dealers, each with CARA risk aversion $\gamma$, absorbs the residual order flow of the opening auction and the continuous session and warehouses it until the next open. Risk tolerance is additive across agents, so the sector behaves as a single agent with risk aversion $\gamma/\mu$ and prices
inventory according to the inventory it carries at that moment $H$:
\begin{equation}
r \;=\; g \;+\; \Lambda\,H,
\qquad
\Lambda \;=\; \frac{\gamma}{\mu}(T \sigma_u^2 + \sigma_0^2).
\label{eq:impact}
\end{equation}
Even though the rebalance conveys no new information to a Bayesian dealer, they must nevertheless hold the resulting inventory until the next session. The concession therefore compensates inventory risk rather than information acquisition. Equation~\eqref{eq:impact} also holds the attending mass fixed, so marginal price impact is linear. This assumption is innocuous through the session, where order flow arrives diffusely, but it binds at the opening and closing auction, where the entire rebalance clears in one event and an order of that size mobilizes capacity beyond the incumbent mass. The opening and closing curve is concave at scale rather than linear, and we take it as given for the rest of the section.\footnote{Appendix~\ref{app:rho} microfounds the attendance decision, Appendix~\ref{app:curveproof} states and proves the properties of the curve, and Appendix~\ref{app:curve} computes the equilibrium on it.}

In addition, the closing auction is generally deeper than the opening auction or the continuous session, as we also verify empirically in Korea.
\footnote{Measured by the price concession per unit of order flow at matched size and clock, the closing auction absorbs an order at 0.43 of the session's cost --- more than twice the depth --- while the opening auction is the thinnest of the three venues, clearing about a third of the closing auction's value (Appendix~\ref{app:impact} and Table~\ref{tab:kappa}).} To capture this difference, we allow the mass of participating market makers to vary across the two venues. Let $\mu$ denote the mass of market makers in the session and $\mu_c$ the mass in the closing auction. The session's makers do not attend the auction, so a position opened in the session and unwound at the auction clears against two different books, each pricing what it absorbs against the price it inherits.\footnote{Designated liquidity providers in Korea are exempt from quoting obligations during the 15:20--15:30 auction, so the two venues need not draw the same providers. The restriction is also testable: were the session's makers to flatten inside the auction, the session's displacement would be reversed there, whereas Section~\ref{sec:empirics} finds the closing-auction window statistically flat and the reversal overnight.}
Writing $\rho \equiv \mu/\mu_c$ for the depth ratio, the close is the \emph{deeper} venue whenever $\rho < 1$. The auction therefore absorbs flow at
\begin{equation}
\Lambda_c \;=\; \rho\,\Lambda.
\label{eq:rho}
\end{equation}
The split matters because speculators' pre-positioning and their liquidity provision do not cancel in prices. The former raises the price by $\Lambda$ per share, while the latter reverses only $\Lambda_c$, whose difference is retained in the closing price.

Following \citet{kyle1985}, we introduce exogenous noise flow to generate background variation. Noise flow $z$ with variance $\sigma_z^2$ trades during the session, and an exogenous closing imbalance $z_c$ with variance $\sigma_c^2$ enters the auction. These flows provide market makers with profits that offset losses to informed speculators, generate uninformative price movements for the mechanical strategy to follow, and prevent prices from fully revealing fundamental value.

Finally, the strategic sector is a continuum of mass $n$ of identical speculators as in \citet{delongetal1990}. Each has CARA risk aversion $\gamma_A$, observes the private information $u$ in real time, and is \emph{atomistic}, taking $r_o$, $r_s$ and $r_c$ as given, so no trader in the model can individually move the reference price.\footnote{Appendix~\ref{app:largetraders} extends the model to oligopolistic traders following \citet{kyle1989} and shows that the results become stronger with a small number of large traders.} Speculator $i$ trades $x_o^i$ at the opening auction, $x_s^i$ during the session and $x_c^i$ at the auction, cumulating into the position $x^i = x_o^i + x_s^i + x_c^i$ carried overnight and closed out at fundamental value against outside deep capital. $X_o, X_s, X_c$ denote the corresponding aggregates, and equation~\eqref{eq:impact} prices the running sum a dealer is carrying at each venue. Speculators do not arbitrage the mispricing away completely because of their risk-bearing capacity is limited, which leads to the following lemma.

\begin{lemma}[speculators' carry]
\label{lem:carry}
Let an atomistic trader be able to trade at two consecutive venues, at prices $p$ and then $p'$. A position held from one venue to the next exists in finite size only as a risk-priced carry, of size $x = \big(\E[p'] - p\big)\big/\big(\gamma_A \operatorname{Var}(p')\big)$.
\end{lemma}

\subsection{Equilibrium Characterization}
\label{sec:characterization}
We characterize the static equilibrium and then extend it to a dynamic setting by chaining successive trading days. The equilibrium depends on a single sufficient statistic, the loop gain, which separates a stable equilibrium from an unstable one.

\subsubsection{Prices and the loop gain}
\label{sec:loopgain}

The opening auction and the continuous session clear against the same book and differ only in the flow each has absorbed:
\begin{equation}
\begin{gathered}
r_o \;=\; g + \Lambda X_o ,
\qquad
r_s \;=\; r_o + \Lambda (X_s + z) , \\
X_o \equiv \int x_o^i \, di , \quad X_s \equiv \int x_s^i \, di.
\end{gathered}
\label{eq:session}
\end{equation}
$r_s$ differs from $r_o$ by the private news $u$ and the session noise $z$, which are realized after the opening auction. Before they arrive, the position a speculator wants depends on $g$ alone, which is available at the opening auction. The session trade $X_s$ is the update as $u$ and $z$ realize, neither of which involves $g$, so the two prints load on public news identically as we confirm in Section~\ref{sec:charac}.

The closing auction inherits the price from day session, and the attending market makers warehouses the residual of the other three participants. Among them, LETF's order depends on the prices:
\begin{equation}
r_c \;=\; r_s \;+\; \Lambda_c\,H_c ,
\qquad
H_c \;=\; \underbrace{K r_c}_{\text{rebalance}}
\;+\; \underbrace{X_c}_{\text{arbitrage}}
\;+\; \underbrace{z_c}_{\text{noise}} ,
\label{eq:fixedpoint}
\end{equation}
where $H_c$ is the inventory the auction's attendees carry overnight. Only the rebalance order responds to $r_c$, so collecting it on the left solves the auction in one step,
\begin{equation}
r_c \;=\; \frac{1}{1 - \ell}\,\big[r_s + \Lambda_c\,(X_c + z_c)\big],
\qquad
\ell \;\equiv\; \Lambda_c K .
\label{eq:loopgain}
\end{equation}
\begin{definition}[loop gain]
\label{def:loopgain}
The \emph{loop gain} $\ell \equiv \Lambda_c K$ is the closing-auction price
impact multiplied by the complex's rebalancing capital. It is the model's
sufficient statistic for the self-reinforcement created by the LETF's
mechanical rebalance.
\end{definition}

Consider a unit increase in the closing price from any endogenous or exogenous source. The LETF complex's mandate rises by $K$, and executing the additional volume against a book with price-impact coefficient $\Lambda_c$ raises the close by a further $\Lambda_c K = \ell$. The higher close again raises the mandate, and the total displacement converges to
$1 + \ell + \ell^2 + \cdots = 1/(1-\ell)$, the multiplier standing in
front of equation~\eqref{eq:loopgain}. A fraction $\ell$ of every closing
movement is therefore generated by the complex's response to its own
price.

\subsubsection{Equilibrium strategies}
\label{sec:charac}

We characterize the strategic sector's behavior backward, from the auction to the session. At the close, each speculator $i$ has complete information about the market: the fundamental value $v$, aggregate noise $z$, and all positions. Therefore, each speculator holds a terminal position $x_i$ for the day that equates the expected value to marginal risk as we established in Lemma~\ref{lem:carry}:
\begin{equation}
\underbrace{\big(g + u - r_c\big)}_{\text{value minus price}}
\;=\;
\underbrace{\gamma_A \sigma_0^2\, x^i}_{\text{marginal overnight risk}}
\Longrightarrow \qquad
X = \int x^i \, di = \frac{\psi_c}{\Lambda_c}(v- r_c) ,
\qquad
\psi_c \;\equiv\; \frac{n\,\Lambda_c}{\gamma_A\,\sigma_0^2} ,
\label{eq:auctionfoc}
\end{equation}
\emph{Auction capacity} $\psi_c$ prices the overnight risk of the terminal position and determines how aggressively speculators correct at the auction: clearing their schedule against the book gives $\Lambda_c X = \psi_c\,(v - r_c)$, so each unit of mispricing is met by a $\psi_c$ push-back toward fundamental. This definition thus carries the opposite sign to the loop gain in Definition~\ref{def:loopgain}. Note that equation~\eqref{eq:auctionfoc} does \emph{not} contain the speculator's session position. The speculator's terminal holding is determined by the risk--return trade-off, so speculators trade at the auction from the session holding to a fixed target.

Substituting the session and auction strategies into equation~\eqref{eq:fixedpoint} delivers the equilibrium of the entire trading day, with every position and price linear in the shocks. 
\begin{proposition}[equilibrium auction price]
\label{prop:passthrough}\label{prop:sessload}\label{prop:loadings}
Let $\psi$ be the session arbitrage capacity built at the closing-price variance a session
position still faces,\footnote{Only the closing imbalance is unresolved once the
session position is set, so $\operatorname{Var}_s(r_c) = a_{z_c}^2\sigma_c^2$.
That loading does not involve $\bar\rho$, so $\psi$ and $\bar\rho$ are explicit
in the primitives rather than a fixed point.} and let $\bar\rho$ be
the effective auction depth
\[
\psi \equiv \frac{n\Lambda}{\gamma_A \operatorname{Var}_s(r_c)}, \qquad 
\bar\rho \;\equiv\; \frac{\rho + \psi_c + \rho^2\psi}{\rho + \psi_c + \rho\,\psi}
\;\in\; [\rho,\,1] .
\]
If the loop gain satisfies $\ell < \bar\rho + \psi_c$, the closing price is linear in the day's shocks,
\begin{equation}
r_c \;=\; \Pi_g\, g \;+\; \Pi_u\, u \;+\; a_z\, z \;+\; a_{z_c}\, z_c ,
\label{eq:Pi0}
\end{equation}
with loadings
\[
\Pi_g \;=\; \frac{1}{1- \frac{\ell}{\bar\rho+\psi_c}} ,
\qquad
\Pi_u \;=\; \chi\,\Pi_g ,
\qquad
a_z \;=\; (1-\chi)\,\Lambda\,\Pi_g ,
\qquad
a_{z_c} \;=\; \frac{\Lambda_c}{1+\psi_c-\ell} ,
\]
where $\chi \equiv \psi_c/(\rho+\psi_c)$. The position the speculators carry overnight is
\begin{equation}
X \;=\; \frac{\psi_c}{\Lambda_c}\,(v - r_c)
\;=\; \frac{\psi_c}{\Lambda_c}\,
\big[ -(\Pi_g - 1)\, g \;+\; (1 - \Pi_u)\, u \;-\; a_z\, z \;-\; a_{z_c}\, z_c \big] .
\label{eq:terminal}
\end{equation}
\end{proposition}

Dividing through, $\Pi_g = 1/\big(1 - \ell/(\bar\rho+\psi_c)\big)$: the multiplier of Definition~\ref{def:loopgain}, evaluated not at $\ell$ but at $\ell$ deflated by the depth that absorbs it, with market makers contributing $\bar\rho$ and speculators $\psi_c$.\footnote{Equation~\eqref{eq:Pi0} is the linear-absorption reading. Section~\ref{sec:quant} prices each day on the concave impact curve of Lemma~\ref{lem:curve}, where the same clearing identity delivers a strictly smaller loading at every positive print and a finite one at every $\ell$; Appendix~\ref{app:concaveclose} states that closed form and shows this display to be its vanishing-print limit.} However large that depth, a positive loop gain leaves the dose positive and the multiplier above one, so the close overshoots public news. The overshoot in turn signs the speculators' position: the public gap enters equation~\eqref{eq:terminal} negatively, so they sell into the overshoot and finish a favorable-news day short the asset they were long. They close out at fundamental value against outside capital, profiting from the displacement they helped create while warehousing the reversal of Section~\ref{sec:chaining}.

\begin{corollary}[a universal overshoot]
\label{cor:overshoot}
For any $\ell > 0$ the close overshoots public news, and speculators sell into
it:
\[
\frac{\partial r_c}{\partial g} \;=\; \Pi_g \;>\; 1 ,
\qquad
\frac{\partial X_c}{\partial g} \;<\; 0 .
\]
\end{corollary}

Corollary~\ref{cor:overshoot} is the model's central prediction, and it rests on weak conditions. First, overshooting is not confined to distress episodes: it resembles the overshooting in the classic predatory-trading model of \citet{brunnermeierpedersen2005}, but arises here because the LETF must rebalance every day, whereas predation in that model requires a distressed liquidation. Second, it requires no manipulation by any individual trader, being the collective outcome of a continuum none of whom has price impact. Third, it requires no irrationality: every agent in the model optimizes.

Turning to the session, the position the speculators carry into the close loads on three factors.

\begin{proposition}[the position and prices before the auction]
\label{prop:sessionpos} 
Let $\mathcal{S}$ be the share of the position that is pre-positioned and
and $\chi$ be the arbitrage share of auction depth of
Proposition~\ref{prop:loadings}:
\[
\mathcal{S} \;\equiv\; \frac{\rho\,\psi}{\rho + \psi_c + \rho\,\psi} ,
\qquad
\Psi \;\equiv\; \chi\big[(\bar\rho+\psi_c) - (1-\mathcal{S})\,\ell\big].
\]
If $\ell < \bar\rho + \psi_c$, the trades before the auction are linear in the
shocks realized when each is made,
\begin{equation}
\begin{gathered}
X_o \;=\; \frac{\mathcal{S}\,\ell\, g}{\Lambda(\bar\rho+\psi_c-\ell)} , \\[4pt]
X_s \;=\; \frac{1}{\Lambda(\bar\rho+\psi_c-\ell)} \Big[\,
\underbrace{\Psi\, u}_{\text{private news}}
\;+\;
\underbrace{(\mathcal{S}\,\ell - \Psi)\,\Lambda\, z}_{\text{known noise}}
\,\Big] ,
\end{gathered}
\label{eq:sessload}
\end{equation}
and clearing them against equation~\eqref{eq:impact} gives the opening and session prices,
\begin{equation}
r_o \;=\; \Big(1 + \frac{\mathcal{S}\,\ell}{\bar\rho+\psi_c-\ell}\Big)\, g ,
\qquad
r_s \;=\; r_o \;+\; \Lambda\,(X_s + z) ,
\label{eq:sessprice}
\end{equation}

\end{proposition}

We note that the public-news position has been incorporated entirely at the opening auction: $X_o$ carries all of $g$ and the session trade $X_s$ carries none. Public news is known before the
first trade, so no speculator waits to act on it and the clearing of the opening auction fully prices in $g$ as in an efficient market. The opening and the session prices therefore have the same loading on public news, $1 + \mathcal{S}\ell/(\bar\rho+\psi_c-\ell)$, above one. The
overshoot is standing in the price from the open, hours before the mandated order is submitted, as plotted in Figure~\ref{fig:loading_path}. This prediction is confirmed later empirically in Section~\ref{sec:empirics}.

The loading on noise, by contrast, depends on the size of the complex, and its sign balances two motives: \emph{liquidity provision} and \emph{pre-positioning}. 
Without the complex, noise trading creates mispricing, and speculators supply liquidity against it at rate $\Psi$. With the complex present, the same noise also forecasts order flow because the rebalance is sized on the closing return: a unit of noise-induced price movement is worth $\ell$ at the auction, and pre-positioning earns the share $\mathcal{S}$ of it that the session leg captures. 
The session position is the net of the two, $(\mathcal{S}\ell-\Psi)$. Speculators fade noise while liquidity provision dominates and ride it once pre-positioning does, the two crossing
at the noise-chasing threshold of Proposition~\ref{prop:ladder}. The timing of the $u$- and $z$-components, which resolve during the session, is characterized in Appendix~\ref{app:strategies} (Proposition~\ref{prop:profile}).

Pre-positioning also raises the closing price. Speculators source in the session the shares they resell to the fund at the auction, so part of the depth the rebalance appears to face is the inventory acquired in the thinner venue. Arbitrage capital therefore cuts both ways: the same capital, with the same information and preferences, reduces the overshoot when it trades at the auction and increases it when it pre-positions. No individual moves the price, but together they shift the auction's marginal supply forward into the day, raising what the LETF pays.

\subsubsection{Chaining Trading Days: A Dynamic Extension}
\label{sec:chaining}

The static model is completed by an assumption about the fraction of the news pricing error corrected by overnight capital. The dynamic extension adds little further structure within each day: each trading day has the same form, and the complex's rebalancing contract links successive days. The complication lies in the cross-day cylces: when the current overshoot unwinds, LETF must rebalance again and may create a new overshoot. This sequence can generate oscillation and further amplify the close-to-close return volatility.

We index cycles by $t$ and let $e_t$ denote the \emph{news pricing error} of cycle $t$: the part of the closing displacement $r_{c,t} - v_t$ created by the amplification of the day's public news. A share $\theta \in (0,1]$ of it is corrected overnight by the outside capital: the news is public, so anyone can see how far the close has moved beyond it. Displacements left by private information and noise cannot be told apart from fundamental value, so nothing trades against them and they do not carry across cycles. The following cycle's initial gap is then $g_{t+1} = \varepsilon_{t+1} - \theta e_t$: fresh public news \emph{net of} the overnight correction, two components the complex cannot distinguish.%
\footnote{Under the opening-auction timing, speculators also pre-position at the cycle-$t{+}1$ open on the fresh news $\varepsilon_{t+1}$. They take no new position against the correction component $-\theta e_t$: the position that profits from the coming unwind is their cycle-$t$ terminal holding, equation~\eqref{eq:terminal}, acquired by selling into the overshoot and already impounded in cycle-$t$ prices, so it arrives at the next open as inventory rather than flow. The overnight leg is therefore the $\theta$-correction alone, and part~(iii) of Proposition~\ref{prop:chaining} is unchanged. The convention is testable: pre-positioning on the full gap would cap the ratio in part~(iii) at $1/\mathcal{S} \approx 1.08$ at every loop gain, whereas the measured ratio is $1.66$ (Section~\ref{sec:calibration}).}
The complex therefore sells into the decline generated by the unwinding of the overshoot it created yesterday. Beyond that, the single-day loadings we introduced in Section~\ref{sec:charac} are adequate to characterize the consequences.

\begin{proposition}[return chaining]
\label{prop:chaining}
The news pricing error follows a first-order autoregression with
negative coefficient,
\begin{equation}
e_{t+1} \;=\; (1 - \theta\Pi_g)\, e_t \;+\; (\Pi_g - 1)\,\varepsilon_{t+1} .
\label{eq:AR}
\end{equation}
It is stationary if and only if $0 < \theta\Pi_g < 2$. Following a one-unit public shock: (i) the close overshoots by $\Pi_g - 1$ on the day of the shock; (ii) at horizon $h$ the reversal delivered overnight is
$-\theta(\Pi_g-1)(1-\theta\Pi_g)^{h-1}$ and the total is
$-\theta\Pi_g(\Pi_g-1)(1-\theta\Pi_g)^{h-1}$; and hence (iii) the ratio of the total reversal to its overnight component equals $\Pi_g$ exactly, at every horizon and every $\theta$.
\end{proposition}

Proposition~\ref{prop:chaining} shows that the price response to news \emph{oscillates}: it alternates around the value of the news with geometric decay ratio $-(1-\theta\Pi_g)$, dampening when $\theta\Pi_g < 2$ and becoming self-sustaining beyond that threshold. This phenomenon is illustrated in Figure~\ref{fig:loading_path} and confirmed empirically in Figure~\ref{fig:loading_cycle}. This oscillation generates excess volatility, which the model derives in closed form below.

\begin{proposition}[excess volatility]
\label{prop:vol}
Let $\sigma_0^2$, $\sigma_u^2$, $\sigma_z^2$ and $\sigma_c^2$ be the
variances of the shocks in return units. If and only if $\theta\Pi_g < 2$,
the cycle return has finite variance
\begin{equation}
V(\ell) \;=\;
\big[\Pi_g^2 + h\,(\Pi_g-1)^2\big]\,\sigma_0^2
\;+\; \Pi_u^2\,\sigma_u^2
\;+\; a_z^2\,\sigma_z^2 \;+\; a_{z_c}^2\,\sigma_c^2 ,
\qquad
h \;\equiv\; \frac{\theta\Pi_g}{2-\theta\Pi_g} ,
\label{eq:vol}
\end{equation}
$V$ is strictly increasing in $\ell$ on the whole
stationary region, channel by channel, for every $\psi_c$ and every $\theta$, and
for $\ell > 0$, $V(\ell) \to \infty$ as $\theta\Pi_g \to 2$.
\end{proposition}

The self-reinforcement raises the volatility through two distinct channels. Within each cycle it inflates the loadings on the day's own shocks, in every channel; across cycles it carries the day's news pricing error into the next return, which is collected in the term $h$ and diverges as the process approaches the stability boundary. The second channel can be tested empirically. News reaches the next return only through the end-of-day pricing error due to overshoot, so the response factors into two steps,
\[
\frac{\partial r_{c,t+1}}{\partial \varepsilon_t}
\;=\;
\underbrace{\frac{\partial r_{c,t+1}}{\partial e_t}}_{-\theta\Pi_g}
\;\cdot\;
\underbrace{\frac{\partial e_t}{\partial \varepsilon_t}}_{\Pi_g-1}
\;=\; -\,\theta\Pi_g(\Pi_g-1).
\]
Regressing today's absolute return on yesterday's absolute news identifies $\theta\Pi_g(\Pi_g-1)$ exactly, which is our conditional volatility test in Section~\ref{sec:volvol}. By contrast, unconditional volatility increases with the variance of every shock and therefore does not isolate the loop gain.

\subsubsection{Equilibrium Regimes}
\label{sec:ladder}

The loop gain is a central determinant of equilibrium behavior, which changes at a series of thresholds summarized in Proposition~\ref{prop:ladder}. Read in order, the four thresholds describe what a market looks like as the complex grows relative to the clearing venue.

\begin{proposition}[loop gain thresholds]
\label{prop:ladder}
Write $\chi = \psi_c/(\rho+\psi_c)$. If $\theta\,\psi_c < (2-\theta)\,\rho\,\mathcal{S}$ and $2\psi_c < \theta\,(\rho+\psi_c)$, the following four thresholds rank serially in the loop gain $\ell$. Each is the effective depth $\bar\rho+\psi_c$ times a factor in $(0,1]$, so the ladder is a statement about the effective dose:
\begin{center}
\begin{tabular}{lll}
\toprule
Threshold & Location & What changes there \\
\midrule
Noise-chasing & $\ell = (\bar\rho+\psi_c)\dfrac{\psi_c}{\psi_c+\rho \mathcal{S}}$ &
\begin{tabular}[t]{@{}l@{}}speculators begin amplifying\\
uninformative movements into the close\end{tabular} \\
\addlinespace
Oscillation & $\ell = (\bar\rho+\psi_c)\big(1-\tfrac{\theta}{2}\big)$ &
\begin{tabular}[t]{@{}l@{}} The echo of the overshoot ceases to damp; \\
an oscillation sustains itself absent news
\end{tabular} \\
\addlinespace
Revelation & $\ell = (\bar\rho+\psi_c)(1-\chi)$ &
\begin{tabular}[t]{@{}l@{}} the close begins to overshoot private news \\
  in addition to public news
\end{tabular} \\
\addlinespace
Tipping & $\ell = \bar\rho + \psi_c$ &
\begin{tabular}[t]{@{}l@{}}a pole of the equilibrium map:\\
no interior equilibrium exists beyond\end{tabular} \\
\bottomrule
\end{tabular}
\end{center}
\end{proposition}

Below the noise-chasing threshold, the complex is present but dominated by arbitrage capacity. Although the close already overshoots public news by Corollary~\ref{cor:overshoot}, arbitrage capital continues to trade against noise, and the pricing error reverses during the next cycle. The market differs quantitatively, but not qualitatively, from one without leveraged products. Between the noise-chasing and oscillation thresholds, speculators begin to ride noise traders rather than absorb their flow. Noise is thus converted into price movements that enters the closing price, and volatility rises through both channels in equation~\eqref{eq:vol}. The market nevertheless corrects across days, with errors diminishing over time. Between the oscillation and revelation thresholds, the response ceases to dampen, so a single shock initiates an alternation around fundamental value that persists without further news, and the variance of close-to-close return diverges to infinity. Between the revelation and tipping thresholds, the close over-prices every kind of information, even those privately known to the speculators.

When $\ell \geq \bar\rho+\psi_c$, the model's interior equilibrium ceases to exist. This nonexistence is a modeling artifact: the loop gain is sufficiently large that every movement in the closing price is amplified, allowing speculators to move the price without bound while remaining profitable.

Under the calibrated primitives of Section~\ref{sec:calibration}, which differ across the two treated names only through the auction capacity $\hat\psi_c^{\,i}$, the four thresholds sit at $\ell = 0.29$, $0.41$, $0.46$ and $0.74$ for SK~Hynix and at $0.51$, $0.53$, $0.45$ and $0.94$ for Samsung. The measured post-launch loop gains place SK~Hynix ($0.58$) between the revelation and tipping thresholds and Samsung ($0.22$) below all four. Samsung's larger arbitrage capacity also violates the ordering conditions of Proposition~\ref{prop:ladder}, putting its revelation threshold below its noise-chasing one; the conditions are tested at the calibrated values, never imposed.

\subsection{Results and Implications}
\label{sec:results}

The characterization yields three groups of implications. 
The first concerns prices. The loading path on public news of Corollary~\ref{cor:overshoot} are illustrated in Figure~\ref{fig:loading_path} and then empirically validated in Figure~\ref{fig:loading_cycle} of Section~\ref{sec:empirics}. 
The second concerns conduct. Building on the thresholds of Section~\ref{sec:ladder}, the model separates anticipation from the self-referential component of the overshoot and prices the latter in closed form. 
The third concerns welfare. We discuss the loser and winners of these trades and why the standard performance metric like volatility drag fails to quantify them.

\subsubsection{The overshoot}
\label{sec:toll}

Positive feedback trading generally produces overshooting, as in the classic papers of \citet{delongetal1990,brunnermeierpedersen2005}. The overshoot in our model is summarized by the coefficient $\Pi_g$ and generates three testable implications. First, Corollary~\ref{cor:overshoot} shows that at every positive loop gain, the close overprices public news by $\Pi_g - 1$ per unit. This loading is fixed since the market opens due to speculators' prepositioning. Second, Proposition~\ref{prop:chaining} shows that the closing pricing error reverses the next day. The reversal begins overnight and decays geometrically, and the ratio of the total reversal to its overnight component equals $\Pi_g$ at every horizon and correction speed. Finally, Proposition~\ref{prop:vol} gives the second-moment implication: the same $\Pi_g$ determines how the magnitude of yesterday's news affects the magnitude of today's return.

\subsubsection{Pre-positioning and predation}
\label{sec:conduct}

Speculators pre-position for the rebalance, but this behavior alone does not establish predation. We construct a benign benchmark by restricting speculators to pure value trading. A value trader's order is always directed against the mispricing, so each trade moves the price toward $v$. The trader takes the offsetting position at the auction, holds it overnight, bears the risk of incomplete correction, and is compensated for doing so.

The equilibrium speculator fails this standard in one channel at every positive loop gain. By Proposition~\ref{prop:sessionpos}, the speculator carries a position in the direction of the public gap, and the session price loads on that gap above one because the speculator carries it. The position cannot be attributed to private information: $g$ is public before the first trade, and by equation~\eqref{eq:sessprice} the opening print it produces already stands above fundamental value. The position is profitable only because a mechanical buyer is known to arrive at the close. Above the noise-chasing threshold, the noise channel operates for the same reason, and predatory trading becomes more pronounced as $\ell$ increases. In comparison, a benchmark value trader would hold no position on a day driven solely by public news and then corrects at the auction exactly as the equilibrium speculator does.

The pre-positioning also move the price impact at the closing auction forward into the day. The benchmark trader holds no position on a day driven solely by public news, so the benchmark opens at fundamental value and stays there, $r_o = r_s = g$, and the entire benchmark overshoot, $\ell/(1+\psi_c-\ell)$ per unit of news, is delivered as a single step inside the closing auction. The predatory equilibrium instead places the step at the opening print (Figure~\ref{fig:loading_path}) and leaves only a small portion of the price impact for the auction window. An overshoot standing in the price before the fund has traded can be put there only by capital acting in anticipation of the mandate. The proposition below quantifies the predation by comparing the closing price under speculators and benigh value traders.

\begin{proposition}[the predatory share]
\label{prop:predshare}\label{prop:relocation}
Let $\Pi_g^{\,0} = (1+\psi_c)/(1+\psi_c-\ell)$ be the closing overshoot in the
benchmark economy of value traders. Pre-positioning raises the overshoot above
$\Pi_g^{\,0}$ if and only if $\rho < 1$, and the \emph{predatory share} of the
overshoot is
\begin{equation}
\mathcal{P} \;\equiv\;
\frac{\Pi_g \;-\; \Pi_g^{\,0}}{\Pi_g - 1}
\;=\; \frac{1-\bar\rho}{1+\psi_c-\ell}
\;=\; \frac{\mathcal{S}\,(1-\rho)}{1+\psi_c-\ell} ,
\qquad
\mathcal{S} \;=\; \frac{\partial r_s/\partial g - 1}{\Pi_g - 1} ,
\label{eq:predshare}
\end{equation}
with $\mathcal{S}$ the session share of the public-news overshoot. On the interior
region, and for $\rho \le 1$, $\mathcal{P} \in [0,1)$: it vanishes exactly
at $\rho = 1$ or $\psi = 0$, increases with $\mathcal{S}$, the depth wedge $1-\rho$,
and the dose $\ell$, and decreases with $\psi_c$ at fixed $\mathcal{S}$.
\end{proposition}

The share has a causal interpretation. An economy of value traders would deliver a closing price that is less costly for the fund by exactly $\mathcal{P}$ of the overshoot. The measure therefore defines conduct by its consequence for the prey rather than by the form of any individual position. It also depends materially on the size of the complex. Holding conduct $(\mathcal{S}, \rho, \psi_c)$ fixed, the same pre-positioning generates a larger fraction of a larger overshoot as the complex grows, with $\mathcal{P} \to 1$ as the economy approaches the pole of the equilibrium map.

The same margin can be expressed in terms of the LETF's rebalancing quantity. Because $q_L = K r_c$, a share carried into the auction enlarges the complex's mandated order by
\[
\frac{\partial q_L}{\partial X_s}
\;=\; \frac{K\,(\Lambda - \Lambda_c)}{1-\ell+\psi_c} \;>\; 0 ,
\]
the depth wedge of Section~\ref{sec:environment} scaled by the size of the
complex. Predatory trading increases both the price the fund pays and, at that price,
the quantity it is obligated to buy. When the session is driven only by public news, the
fraction of the order induced by the price distortion is also $\mathcal{P}$.

No individual trader manipulates the fund's order. Each speculator is of measure zero, so $\partial q_L / \partial x_s^i = 0$ for every $i$. The invariance of the
prey's trade to others' trades --- the boundary
\citet{brunnermeierpedersen2005} draw between predatory trading and
manipulation --- fails in the aggregate while holding for every individual speculator. 
A conduct test evaluated on individual holdings accordingly returns zero for this economy while the aggregate harm is undiminished.

\subsubsection{Welfare}
\label{sec:welfare}

Predatory trading has direct welfare consequences for LETF holders. The overshoot creates a mechanical market distortion that transfers wealth from holders to their counterparties. The following proposition quantifies holder welfare per cycle and per unit of rebalancing capital.

\begin{proposition}[holder welfare]
\label{prop:loss}
Marking all positions at fundamentals, the holders' expected welfare per
cycle per unit of rebalancing capital, $W_h \equiv \E[r_c(v - r_c)]$,
decomposes across the shocks:
\begin{equation}
W_h \;=\;
\underbrace{-\,\Pi_g (\Pi_g - 1)\, \sigma_0^2}_{\text{public channel}}
\;\underbrace{-\,\Pi_u (\Pi_u - 1)\, \sigma_u^2}_{\text{private channel}}
\;\underbrace{-\,\big(a_z^2\, \sigma_z^2 + a_{z_c}^2\, \sigma_c^2\big)}_{\text{noise channels}} .
\label{eq:loss}
\end{equation}
The public and noise channels are negative at every $\ell > 0$, while the private
channel is positive if and only if $\ell$ lies below the revelation threshold
$(\bar\rho+\psi_c)(1-\chi)$. 

Meanwhile, there exists a threshold $\ell^{\ast}$ decreasing in $\sigma_0/\sigma_u$
\begin{equation}
\ell^{\ast} \;=\;
\frac{(\bar\rho+\psi_c)\,(1-\chi)\,
      \big[\chi\,\sigma_u^2 \;-\; (1-\chi)\,\sigma_z^2\big]}
     {\sigma_0^2 \;+\; \chi\,\sigma_u^2} ,
\qquad \chi = \frac{\psi_c}{\rho+\psi_c} ,
\label{eq:losthreshold}
\end{equation}
such that $W_h < 0$ whenever $\ell > \ell^{\ast}$, with equivalence at $\sigma_c = 0$.
\end{proposition}

Equation~\eqref{eq:losthreshold} bears directly on the policy debate. The industry argues that a fund's momentum rebalancing may earn a drift subsidy. The model supports this claim only for the private-information channel, in which the close underreacts. It does not extend to public information, which accounts for most daily information arrival. 

The model also demonstrates who receives the wealth transferred from LETF holders. Every transaction occurs at a common price (Appendix Proposition~\ref{prop:ledger}), so the holders' loss must accrue to speculators and market makers. The market makers' share rises with the loop gain: the inelastic complex pays an inventory concession on gross flow that grows with $\Pi_g^2$, while speculator positions are bounded by risk tolerance. However, neither party earns a markup because no individual has market power in the model. Speculators' gains compensate them for the overnight risk they bear in correcting the closing price.
Their certainty-equivalent surplus is proportional to the mean squared pricing error they fail to eliminate, so their profit is a by-product of the distortion.

\section{Empirical Evidence}
\label{sec:empirics}

The model generates testable implications for three families of outcomes. First, the closing price should overshoot the level implied by public news, with the excess partially reversing over the next cycle (Proposition~\ref{prop:chaining}). Without predation, public news is incorporated at the open and persists throughout the session, so the continuation coefficient is nonnegative. Second, overshooting public news generates conditional, news-driven excess volatility in stocks with leveraged products. Third, the mandated order and the trading it attracts should migrate toward the mechanisms that set the reference price. All three predictions are supported by our tests in the recent Korean episode.

The empirical analysis utilizes a stock-day panel covering both the Korean and U.S. markets from January 2024 through July~31, 2026. The Korean panel includes the two treated stocks, Samsung Electronics and SK~Hynix; ten large-capitalization KOSPI controls, the largest firms that failed to meet the LETF eligibility screen; and ten KOSPI semiconductor peers. The U.S. panel includes eight stocks with listed leveraged products and four major indices, providing comparison groups with relatively small LETF products. We also obtain intraday data directly from the Korea Exchange (KRX): one-minute quotes and signed volume for every KOSPI stock and futures contract from May~4 through July~31, 2026.

\subsection{Measuring the Loop Gain}
\label{sec:dose}

The loop gain $\ell$ measures the self-reinforcing effect of an LETF complex's mandated rebalance and equals the product of the complex's price impact and rebalancing capital. It directly quantifies both the pressure the complex exerts on the underlying asset and the resulting predatory response by speculators. We measure each input and compare the resulting loop gains across the Korean and U.S. markets.

We measure the two inputs separately. Rebalancing capital $K$ is the worldwide sum of equation~\eqref{eq:K} over every LETF product referencing the stock: the Korean domestic funds, the Hong Kong cross-listings, the U.S. products, and the index products at their pass-through weights, each converted to Korean won at the day's foreign exchange rate. Same-day creations and redemptions could offset the mandated trade, so we measure them as well from the outstanding shares implied by net assets. We note that LETF fund flow is not significantly correlated with its same-day return.

The impact coefficient $\Lambda_c$ is the capital needed to displace the
\emph{closing price} by one percent. Unlike a generic price-impact measure, it is estimated over the pre-close window in which the complex trades rather than over an average trading day. Using the native intraday exchange data, we estimate it from binned conditional means of the price change to
the close on imbalanced pre-close order flow, which additionally delivers the curvature
of the impact curve. More details of the implementation are reported in Appendix~\ref{app:impact}. Because that curve is concave, the loop gain depends on the size of the order at which it is evaluated.

Figure~\ref{fig:dose} plots loop gain over the episode, computed identically for the Korean stocks and the U.S. index and single-stock complexes.
Most U.S. stocks and indices reproduce the null documented in the literature. The loop gains of all but one U.S. single-stock complex lie between $0.03$ and $0.10$, so the mandated order contributes only a small fraction of the session's movement in the closing price. MicroStrategy is the exception, with a loop gain of $0.22$. It is also among the most volatile U.S. stocks, with a daily return standard deviation of 5.3\%, compared with 2\% for the average large-capitalization U.S. stock.

After the domestic LETF launch, Samsung Electronics occupied a position similar to MicroStrategy, with an average loop gain of $0.25$. This comparison is informative because MicroStrategy is a U.S. outlier: Samsung is comparable to the most extreme U.S. case and thereby locates the upper bound of the range over which the U.S. evidence remains benign.

SK~Hynix reached a distinctly higher level. At the beginning of 2026, a Hong Kong-listed product already generated a loop gain of $0.25$ in the underlying stock; the domestic launch increased it further.
Its loop gain averaged $0.51$ during the post-launch period, twice Samsung's value and above any observed U.S. value. At its measured peak, SK~Hynix crossed the oscillation threshold beyond which the overshoot and reversal persist across days without further news, given the parameters calibrated in Section~\ref{sec:quant}.\footnote{Proposition~\ref{prop:ladder} characterizes the oscillation
threshold at $\ell = (\bar\rho+\psi_c)(1-\theta/2)$, or equivalently
$\theta\Pi_g \ge 2$. It computes as $\ell^{\ast} = 0.63$ at the calibrated $(\bar\rho, \hat\psi_c, \hat\theta) = (0.48,\, 0.65,\, 0.88)$. Because the threshold is a property of the calibration, the crossing is read on the calibrated loop-gain path of Section~\ref{sec:calibration}, which peaks at $0.79$ for SK~Hynix and at $0.28$ for Samsung against its own threshold of $0.94$. The path plotted here is the same measurement placed on the cross-market convention that puts every market on one panel, and peaks at $0.62$ and $0.32$. Because $\hat\tau$ is inverted to
match the pooled $\hat\Pi_g$, which pins $\bar\rho+\psi_c$, $\ell^{\ast}$ moves only between $0.629$ and $0.633$ across the four readings of $\rho$, spanning $0.25$ to $0.48$.} If the
mechanism has a threshold anywhere in the plotted range, only Korea crossed
it. Section~\ref{sec:echo} scales the mandated flow directly by
realized closing-auction values as an independent check on this ordering. Consistent with this ranking, SK~Hynix exhibited record volatility in the
post-launch window, with a daily return standard deviation of $8.8\%$, the highest of any equal-length window in our data since 2015.

\subsection{The Overnight Reversal}
\label{sec:echo}

Proposition~\ref{prop:chaining} establishes that the next cycle's close-to-close return loads \emph{negatively} on the current public shock. The closing price overweights the news by $\Pi_g - 1$, which unwinds overnight. For each cell $j$, we therefore estimate the following level regression
\begin{equation}
r_{i,t+1} \;=\; \alpha_j \;+\; \beta_j\, \varepsilon_t
\;+\; \gamma_j\, \varepsilon_{t+1} \;+\; e_{i,t+1},
\label{eq:echoreg}
\end{equation}
where $r_{i,t+1}$ is the next-day close-to-close return. The intervening shock $\varepsilon_{t+1}$ is controlled so that $\beta_j$ is protected from its serial correlation. 
We first separately estimate $\beta_j$ in equation~\eqref{eq:echoreg} for each cell and test which cells exhibit a significantly negative coefficient. This allows us to inspect every cell that could have produced the result but did not.

We use returns on three U.S. indices---SMH (semiconductors), QQQ (broad technology), and SPY (broad market)---as instruments for the public-news shock $\varepsilon_t$. These returns provide valid instruments for three reasons. First, the U.S. market is roughly 25 times larger than the Korean market, making these large, diversified indices unlikely to respond materially to prior Korean returns. Second, the U.S. session closes around 05:00 KST, four hours before the Korean market opens, so the shocks are realized before Korean trading begins and are observable to all Korean traders. Third, the instruments are relevant because Korea's economy is highly export-oriented, and its two leading producers of AI-memory semiconductors serve substantial U.S. demand. U.S. index returns therefore capture a salient component of the model's public-news shock. We use and report results for all three instruments throughout to demonstrate robustness.

Before examining next-day return loadings, we note that the same-day loadings are strongly positive in every cell. A one-unit overnight QQQ movement is incorporated into the same day's close with a loading of $1.59$ for the treated Korean pair after the LETF launch, $0.77$ for the same stocks before the launch, $0.59$ for the broad Korean controls, and $1.49$ and $0.62$ for the U.S. single-stock and index complexes, respectively. These coefficients reflect the usual cross-market transmission scaled by economic exposure.

Table~\ref{tab:echo} reports the next-day loadings. In the treated-post cell---Samsung and SK~Hynix after May~27, 2026---the next-day close-to-close return exhibits a statistically significant reversal, with loadings on the first-day shocks of $-0.83$, $-1.76$, and $-3.92$ across the three instruments. Netting out the three control cells, the triple differences are $-0.79$, $-1.55$, and $-2.29$; scaled by the first-day response measured at the day-$t$ open, they imply that approximately $80\%$, $75\%$, and $68\%$ of a trading day's news-driven movement, respectively, reverses by the next close. The three instruments therefore deliver consistent and robust estimates of the overshoot factor.

One might argue that the reversal instead reflects the mechanical price impact of the mandated order at the close. Panel~(a) of Figure~\ref{fig:loading_cycle} rules this out by tracing the complete impulse response to the overnight news shock. The underlying's loading on that shock is priced from the open and stays flat through the first day, so the overshoot is already in place hours before the mandated order is submitted, and no further impact appears in the closing-auction window. The order's price impact is anticipated and absorbed during the session rather than delivered at the close.

The reversal begins only on the following day. The right panel of Table~\ref{tab:echo} shows that between $32$ and $45$ percent of the first-day loading unwinds overnight, before the next open, with the remainder arriving during the following session. That reversal is itself an overreaction, through the chaining described in Section~\ref{sec:chaining}. Panel~(b) of Figure~\ref{fig:loading_cycle} follows the response over a longer horizon using local projections \citep{jorda2005}. The reversal is reversed again over the next two days, and by the third day the loading has returned close to its first-day level. Proposition~\ref{prop:ladder} predicts exactly such oscillation when the loop gain is larger than a threshold.

No comparison cell exhibits a comparable reversal: neither the U.S. index complexes, the eight U.S. stocks with their own $2\times$ products, the treated stocks before the domestic launch, nor the Korean controls before or after the launch (Table~\ref{tab:echo}; Figure~\ref{fig:loading_cycle}; Appendix Table~\ref{tab:echo_full}). Nvidia and AMD provide the closest comparison. Both are semiconductor firms with leveraged products, yet neither exhibits a statistically significant loading. The Korean reversal therefore cannot be attributed to the semiconductor cycle, AI-related news flow, or the mere existence of a leveraged product. Instead, the evidence points to the scale of the Korean LETFs discussed in Section~\ref{sec:loopgain}, which distinguishes the post-launch treatment cell.

Collapsing the treatment and control groups into a single estimate, we stack
the Korean panel and estimate
\begin{equation}
r_{i,t+1} \;=\; b_3\,\varepsilon_t T_i P_t \;+\; \gamma' x_{i,t}
\;+\; e_{i,t+1},
\label{eq:ddd}
\end{equation}
where $T_i$ marks the treated pair, $P_t$ the post-launch period, and
$x_{i,t}$ collects every lower-order term of the triple interaction. The coefficient $b_3$ is therefore the causal effect of the launch on the treated pair's next-day loading, net of the
controls' loading over the same days, the pair's own loading before the
launch, and any shift common to the market. The estimate is $-0.79$ using large-capitalization Korean stocks as controls and $-0.48$ using the ten Korean semiconductor peers, providing further evidence of a causal effect of the LETF launch.

We further rule out two behavioral interpretations of the result. The first is investor psychology, under which the treated names simply became more volatile or more sentiment-driven after the launch. 
The mandated rebalance trade scales with the size of the day's move, however, and so does the reversal: splitting shock days at the median absolute shock, the reversal phenomenon loads almost entirely on the large-shock half, while the small-shock half demonstrates a statistical null (Appendix Table~\ref{tab:mechanical}). This proportionality can not be explained by sentiment. 
The second is return-chasing retail flow, as in \citet{frazzinilamont2008}. Instead, we see the reverse in the Korean episode, where retail buy the dip. Net purchases load $-0.30$ on the leveraged product's own same-day return and $-0.78$ and $-0.75$ on SK~Hynix and Samsung shares,
strengthening after the launch (Appendix Table~\ref{tab:flowchasing}). Retail investors were trading against the move while the mandated order trade towards it.

These findings contrast with a literature that has generally found LETF rebalancing benign in the U.S. market. Existing studies document small price effects that decline over time and are readily absorbed for U.S. index, VIX-linked, and single-stock products \citep{ivanovlenkey2018, brogger2021, barbon2021, lenkey2024, bessembinder2025, huang2025}, and \citet{murraysammon2026} find that holders of U.S. broad-index products gained on net, both in absolute terms and against the counterfactual of holding the underlying. Our estimates extend, rather than contradict, those results by examining a scale not previously observed in the United States.
As Section~\ref{sec:dose} shows, the loop gains of the Korean single-stock complexes are orders of magnitude larger than those of their U.S. counterparts. Because equation~\eqref{eq:echoreg} is
estimated identically in every market, it produces near-zero estimates where rebalancing pressure is small and detects reversal where pressure is large. The two bodies of evidence therefore represent different regions of the same response surface (Section~\ref{sec:toll}).

\subsection{Volatility and Volume}
\label{sec:volvol}

This subsection examines two additional classes of outcomes: volatility and the location of trading. The overshoot and reversal of the closing price predict both greater volatility and a shift in trading volume.

The magnitudes are substantial. After the launch, annualized volatility reached $136.7\%$ for SK~Hynix and $115.3\%$ for Samsung, three to four times their levels in the two years before any leveraged product referenced them ($46.6\%$ and $29.0\%$). Figure~\ref{fig:timeline} shows consecutive days with absolute returns above 5\%; only two such days occurred in all of 2025. Because the treated pair represented approximately 50\% of the Korean market, aggregate market volatility rose concurrently: annualized KOSPI volatility increased from $19.5\%$ to $86.1\%$ over the same comparison. We therefore measure the treated pair's increase relative to this common movement using a difference-in-differences design.

We use a conditional design because unconditional volatility increases with the variance of any shock, whereas conditional volatility isolates the relation between the magnitude of daily returns and the magnitude of overnight news. We estimate
\begin{equation}
\left|r_{i,t+1}\right| \;=\;
\sum_{b} \left(\theta_b \left|\varepsilon_t\right| + \phi_b\right)
   T_i \mathbf{1}\{t \in b\}
\;+\; \beta \left|\varepsilon_t\right| T_i \;+\; \gamma' z_{i,t}
\;+\; \mu_i \;+\; \lambda_t \;+\; u_{i,t+1},
\label{eq:aftershock}
\end{equation}
where $b$ indexes calendar periods, $\theta_b$ is the treated-versus-control difference in the news-volatility slope in period $b$ and $\phi_b$ the corresponding level shift, $\mu_i$ and $\lambda_t$ are stock and date fixed effects, $z_{i,t}$ carries the intervening shock
$|\varepsilon_{t+1}|$ interacted the same way. The date effects absorb the common tendency of large-news days to move every stock, so each $\theta_b$ compares the treated pair's
news-conditional volatility with that of the controls on the same days. Parallel
trends therefore impose the testable restriction that $\theta_b = 0$ in every
pre-launch period.

Figure~\ref{fig:aftershock} shows the pre-launch $\theta_b$ individually and
jointly indistinguishable from zero in both universes, then rising at the
launch and strengthening into July as the loop gain itself grew. Pooled after
the launch, each unit of overnight semiconductor-index movement generated
$0.66$ units of additional treated-stock movement the next day against the
large-capitalization group and $0.81$ against the technology peers. The slope maps into
a 55\% increase in the treated pair's idiosyncratic daily volatility relative to the broad union of control sets, and an imprecisely estimated 18\% increase relative to the technology peers alone. The latter estimate is robust to removing time-varying factor exposures. Appendix Table~\ref{tab:battery} reports the complete set of control-group specifications, including both statistically significant and insignificant results.

The design provides a conservative estimate because volatility may spill over to the control group. The treated pair represents roughly half of the index, and several well-established channels propagate volatility across assets \citep{kylexiong2001, kodrespritsker2002, antonpolk2014}. Any effect transmitted by the launch to the
controls is differenced away, rendering our estimate a lower bound on the treated pair's excess volatility. A complementary within-stock design, differencing close-anchored against
open-anchored absolute returns, removes cross-stock contagion by construction
and confirms the effect at the magnitude the model predicts.\footnote{Fundamental volatility moves both marks alike, whereas the closing
self-reinforcement creates variance specifically at the reference price, so
$z_{it} = \log|r^{cc}_{it}| - \log|r^{oo}_{it}|$ nets out every common
component \citep{frenchroll1986}. The launch-window rise is $+0.52$
close-anchored against $+0.45$ open-anchored, leaving a close-specific
component of $+7$ to $+15$ per cent (Appendix Table~\ref{tab:closespec}).
The model predicts a small effect: the volatility contribution consists primarily of actual
price movements from executed trades, and only the displacement of the mark
itself is leg-specific. A mechanism that doubled volatility through the close
would be rejected here.}

We apply the same design to trading volume in Appendix Table~\ref{tab:quantity}. The most precisely estimated effect occurs in the futures market, the hedging venue for the futures-replicating half of
the domestic complex. We find the treated contracts' closing-auction share of
the day's traded value quadrupled at the launch, from 0.3\% to 1.4\%, and
their final-thirty-minute share rose by 1.9 percentage points. On the
cash side the treated stocks' closing-auction share rose from 7.3\% to
9.8\%, alongside a market-wide movement toward the close that is itself
consistent with the arrival of anticipatory volume.

\section{Policy Implications}
\label{sec:quant}

The reduced-form analysis provides direct evidence of the predatory trading targeting LETFs, supporting the model predictions. This section quantifies the distortion and evaluates alternative designs that could mitigate it.

\subsection{Calibration}
\label{sec:calibration}

We measure and calibrate the model primitives using Korean intraday data, with values presented in Panel~A of Table~\ref{tab:calibration} and the methodology detailed in Appendix~\ref{app:impact}. The average loop gain is $\bar\ell = 0.58$ for SK~Hynix and $0.22$ for Samsung.\footnote{One convention is used throughout: the unweighted mean of the uncapped post-launch daily path $\ell_t=\Lambda_c K_t$ ($0.583$ and $0.219$), the same number in the calibration, the counterfactual tables, and the text. Figure~\ref{fig:dose} alone keeps a cross-market re-anchored construction ($0.51$ and $0.25$) so every market sits on one panel.} In terms of price impact, it takes \KRW 1{,}252 billion of imbalanced flow at SK~Hynix and \KRW 1{,}302 billion at Samsung to displace the closing price by one percent, which is concave in size with $\hat\delta$ measured at 0.73. The closing auction absorbs order imbalances at $\hat\rho = 0.43$ of the continuous session's impact, making it the deeper venue, as Proposition~\ref{prop:attendance} predicts.
Correction of the news pricing error is almost complete across cycles, with $\hat\theta$ measured at 0.88 in Appendix Table~\ref{tab:theta}; the displacements private information and noise leave behind show no detectable within-window correction, as the model assumes (Appendix~\ref{app:calibration}).
As for information, we find public overnight news carries about as much variance as private in-session news ($\hat\sigma_\varepsilon^2/\hat\sigma_u^2 = 1.07$), while noise trading contributes about as much variance as private news ($\hat\lambda_{\mathrm{post}} = 1.06$, from the session-residual identity of Appendix~\ref{app:calibration}).

We cannot directly measure aggregate dealer capacity in either the continuous session or the closing auction, so we recover them from observed moments. Session capacity is identified from the diffusion of public news into session prices. The session share of the overshoot, $\mathcal{S}$, is what the closing price responds to, and it enters through the effective depth ratio $\bar\rho=1-\mathcal{S}(1-\rho)$ alone; our estimate is $\hat{\mathcal{S}}=0.93$, which places the effective depth ratio at $\bar\rho=0.48$. The capacity behind it follows from the session-share map and is common to both names, $\hat\psi=22.8$ at the pooled $\hat\psi_c$. Nothing below depends on that inversion: the calibration and every counterfactual run on $\hat{\mathcal{S}}$ and $\bar\rho$ directly. Auction capacity $\psi_c$ is identified from the overshoot factor. The model implies that the ratio of the total one-cycle reversal to its overnight component equals $\Pi_g$. Using the broad-technology instrument in Table~\ref{tab:echo}, we adopt $\hat\Pi_g=1.66$.\footnote{The three instruments yield $1.85$ for SMH, $1.66$ for QQQ, and $1.35$ for SPY. We adopt the middle estimate, whose implied correction speed, $\hat\theta=0.88$, lies strictly within the unit interval. The estimate is conservative because the gradual diffusion of foreign news documented in every benign cell adds positive drift to both components.} Inverting the pooled overshoot of the concave close at the measured loop gains (Appendix~\ref{app:calibration}) yields a common capacity per unit of impact, $\hat\tau$, and stock-specific auction capacities $\hat\psi_c^{\,i}=\hat\tau\Lambda_c^i$: $0.26$ for SK~Hynix and $0.47$ for Samsung. The data locate $\hat\tau$ only weakly; the appendix reports the diagnostic, and every result below is invariant across its admissible range.

Panel~B of Table~\ref{tab:calibration} evaluates the model's fit on the untargeted moments.
The model implies that the overshoot parameter $\Pi_g$ is $1.18$ for Samsung and $1.69$ for SK~Hynix, both lying inside the confidence intervals around the per-name estimates of $1.75$ for Samsung and $1.60$ for SK~Hynix.
Meanwhile, the recursion requires the next session to carry $\Pi_g - 1$ times the overnight component. That ratio is $0.68$ in the data against $0.69$ in the model, so a little under three fifths of the reversal is accomplished before the next open; because the ratio is mechanically the targeted moment less one, we report it as a consistency check rather than a validation.
The closing-return autocorrelation, the moment most sensitive to the loop's day-to-day echo, is matched without being targeted: the model delivers $-0.06$ against an empirically measured $-0.08$ at Samsung and $-0.14$ against $-0.16$ at SK~Hynix. On the pre-launch window the model further requires a zero autocorrelation and a unit variance ratio exactly, with no free parameter, and both restrictions pass (Appendix~\ref{app:calibration}).

\subsection{Price and Welfare Effects}
\label{sec:effect}
\label{sec:mechanical}

Annualized volatility reached $136.7\%$ for SK~Hynix and $115.3\%$ for Samsung during the nine weeks following the domestic launch, as we introduced in Section~\ref{sec:episode}. Both stocks also fell sharply with the repricing of the AI-memory cycle in July 2026, with peak-to-trough drawdowns of $50\%$ and $38\%$ respectively. Losses were concentrated among retail investors, who held approximately $92\%$ of the single-stock products. Public commentary attributes both the increase in volatility and the subsequent pullback to LETFs. We evaluate these claims by quantifying the effect of the LETF launch on the asset prices of the underlying with our calibrated model.

To do so, we replay the episode's own shocks through a counterfactual market without the self-reinforcement, i.e., with zero loop gain. We apply a Kalman filter and smoother to the observed daily returns to estimate the unobservable structural shocks. The same shock history is then replayed with the loop gain set to zero from the
launch date. The details are provided in Appendix~\ref{app:quant}. The replay is deterministic and runs uncapped: at the calibration both names remain strictly inside the model's stability region on every post-launch day, so no stability device operates on the reported magnitudes.

We find that SK~Hynix's realized annualized volatility of $136.7\%$ would have been $100.0\%$ in a counterfactual with zero loop gain. The complex thus accounts for over one quarter of the stock's volatility during the episode. Samsung also documents $+13.2$ percentage points of excess volatility from the counterfactual $102.1\%$ to the realized $115.3\%$.\footnote{Appendix~\ref{app:replay} reports the bands: the SK~Hynix figure varies between $33.6$ and $37.8$ points over the admissible capacity range and the opaque-correction set, and between $18.7$ and $36.3$ under the out-of-sample rescaling of the dose.}

However, the rebalance did not materially alter the destination of the underlying assets. SK~Hynix's July drawdown of $50.1\%$ compares with $43.5$--$49.2\%$ without self-reinforcement, and Samsung's $38.0\%$ with $34.4$--$37.5\%$, where each range spans the opaque-correction rates the data do not pin down (Appendix~\ref{app:replay}): the counterfactual attributes at most a few points of the drawdown to the loop, and almost none at the upper end of the set. The decline in levels is driven by either sentiment or fundamental news. The self-reinforcement mainly added repeated day-to-day overshoot and reversals that transferred wealth from LETF holders to their counterparties. The paper thus supports the public narrative on volatility but, at most weakly, on the magnitude of the crash.

LETF holders suffer the exaggerated volatility drag generated by the market swings. An $L$-times fund incurs drag equal to $(L^2-L)/2$ times the variance of the tracked asset, which accrues through every reversal and cumulates to a substantial amount over a short period. Marking the complex's realized daily flows against the identified public component of the recovered pricing error (Appendix~\ref{app:replay}), we find the complex transferred \KRW 1.46 trillion (US\$0.97~billion) away from its own holders in nine weeks, \KRW 1.38 trillion of it at SK~Hynix. On the cash-price channel alone, where our structural estimate is \KRW 1.22 trillion, the reduced-form engine---which uses no impact estimate and therefore cannot price the futures leg---puts the transfer at \KRW 1.5 trillion.\footnote{Korean won amounts are converted at
\KRW 1{,}509 per U.S. dollar, the average of the daily rate over the
post-launch window (May~27--July~31, 2026). The same rate is applied to
every holder-loss figure reported in this section.} Per unit, the same
distortion cost SK~Hynix holders $17.6$ per cent of their initial
investment: \KRW 10{,}000 placed in a $2\times$ SK~Hynix product at the launch was worth \KRW 4{,}790 at the end of the window against \KRW 6{,}554 on the counterfactual path; the corresponding figures at Samsung are \KRW 5{,}971 against \KRW 6{,}740, a shortfall of $7.7$ per cent of the initial investment.\footnote{Both marks are taken at the recovered fundamental rather than at the displaced close. The two replays share a fundamental path, so the only wedge between their closing marks is the pricing error standing on the last day of the window.} The magnitude greatly exceeds estimates for the rebalancing of other institutional portfolios, such as eight basis points per year in \citet{harvey2026} and tens of basis points in \citet{petajisto2011}.

A modest but distinct share of this loss is induced by pre-positioning rather than by the mechanical rebalance alone.
If we prohibit speculators from pre-positioning while holding all other parameters fixed, the model's overshoot factor falls from $1.69$ to $1.50$ at SK~Hynix and from $1.18$ to $1.14$ at Samsung, so roughly one quarter of each name's overshoot is manufactured. Replaying the same shut-off on the recovered shocks, pre-positioning accounts for $10\%$ of the transfer (\KRW 117 billion at SK~Hynix and \KRW 7 billion at Samsung) and for $+5.7$ and $+2.6$ percentage points of the two names' excess volatility.
Define $\bar A$, the \emph{manufacturing leverage}, as the additional LETF rebalance triggered per unit of a trader's own pre-close purchases. At the median complex, $\bar A = 1.77$ at SK Hynix and $0.34$ at Samsung. The manufacturing leverage depends only on the price the mandate refers to, not on when the mandate executes. Mitigating this predatory trading therefore requires changing what price the rebalance reacts to.

Having documented the costs of LETF self-reinforcement, we next evaluate several alternative designs that may mitigate it.
Three features of LETF design are subject to consideration: the \emph{schedule} on which the order executes, the \emph{size} of
the order, or the \emph{reference price} against which it is struck.  Tables~\ref{tab:schedule}--\ref{tab:reference} evaluate each design by replaying the shock history of Section~\ref{sec:effect} under the measured calibration.\footnote{The tables hold each venue's depth fixed while changing the flow that arrives there. Attendance is a decision
(Proposition~\ref{prop:attendance}), so a lever that shrinks the mandated order drives away the capacity that came to meet it, and one that routes the order to a second venue makes that venue deeper.}

\subsection{The Execution Schedule}
\label{sec:schedule}

On July~28, 2026, the chairman of the Korean Financial Services Commission asked managers to move rebalancing activity earlier in the session rather than concentrate it immediately before the close. Financial Services Commission release \#87417 argues that concentrating rebalancing near the close ``further amplifies market volatility'' and expresses the aim of reducing ``other investors' anticipatory trading.'' The Commission calls for dispersing the rebalancing \emph{moment}, which moves the reference price with the trade.\footnote{The industry's implementing guidance instead describes dispersing the rebalancing \emph{trades} while retaining the close as the reference price. The latter is the execution split studied in Section~\ref{sec:reference_benchmark}, not a schedule change. The public record does not establish which interpretation managers adopted.} 

Such proposal rests on a natural intuition: distributing demand across time should reduce the price pressure at any single moment. However, our reduced-form evidence cast doubt on this narrative because the rebalance's price impact has emerged since market open due to speculator's pre-positioning. In this section, we evaluate three implementations of this idea that alter the frequency of rebalances, the number of orders submitted, or the participating auctions, while keeping the leverage multiple unchanged. Under our estimated parameters, we find the dispersing the rebalancing moment is indeed counterproductive.

\subsubsection{Dispersing the rebalancing moment}
\label{sec:dispersal}

To disperse the rebalance into the day session, a fund must size its order according to the contemproraneous price rather than a future closing price.\footnote{Executing an order throughout the session while continuing to settle against the closing price is a distinct design, analyzed as the execution split in Section~\ref{sec:reference_benchmark}.} Dispersal therefore creates a distinct reference price at each rebalancing time. The fixed point clearing from Section~\ref{sec:model} repeats at every such rebalancing moment. For simplicity, we model $m$ funds that equally split the complex's capital and rebalance at separate times against their respective contemporaneous references. We evaluate this design at its most favorable form, with every rebalance auction granted the full capacity that the closing auction attracts, and the order splits on the measured concavity $\hat\delta$.\footnote{Granting each auction the close's full attending capacity replicates that capacity $m$ times, an upper bound on what Proposition~\ref{prop:attendance} could deliver; conserving it instead---each auction carrying $\mu+(\mu_c-\mu)/m$ while absorbing $K/m$---leaves no staggered configuration for SK~Hynix with a damped equilibrium. The bound reported in the table therefore fails on concavity and persistence alone.}

Panel~A of Table~\ref{tab:schedule} shows that even this best case does not help: every staggered configuration is worse than the single closing print, monotonically in $m$. Two prints raise SK~Hynix's volatility from $136.7\%$ to $146.9\%$ and its holder loss from \KRW 1.15 to \KRW 1.56 trillion; three and five reach $152.0\%$ and $155.7\%$, lower bounds because past two prints the compounded loading passes the ringing bound on most days. The result can be attributed to two factors. First, staggering divides the order, but not the \emph{shock}. Each fund's reference return still incorporates the entirety of overnight public news, which accounts for more than half of the information each day, so staggering could not fragment the majority of the rebalance.\footnote{Proposition~\ref{prop:dispersal} characterizes the limit. Let $\pi(x)$ denote the loading from one auction with per-auction gain $x$. Then $m$ staggered auctions produce $[\pi(\ell/m)]^m\to e^\ell$ under any microstructure. Because this limit remains strictly above one, fragmentation cannot eliminate daily amplification even with linear impact.} Second, price impacts accumulate rather than cancel out across different rebalance auctions in the session. Due to the concave impact assumption, the accumulation of their price impacts therefore exceeds that brought about from a single closing auction rebalance.~\footnote{With $\hat\delta=0.73$, halving an order reduces its impact by only 40\%. The per-auction loop gain is $0.583$ at the close, compared with $0.351$, $0.261$, and $0.180$ under two, three, and five auctions, but their sum increases with $m$.}

The Commission's second rationale---reducing anticipatory trading---also does not remove the distortion in the model. Even if rebalance orders arrive at a random moment, speculators can still pre-position for it, knowing that the rebalance must react to the overnight public news. Consistent with the evidence provided in Figure~\ref{fig:loading_cycle}, pre-positioning foreloads the price impact since market opens. Concealing the order's arrival may also reduce the capacity attracted by its predictability (Proposition~\ref{prop:attendance}), weakening depth without providing an offsetting benefit from fragmentation.

A related schedule policy alters the frequency of rebalance. A fund that rebalances every $T$ days must trade against the cumulative return over those $T$ days in a single auction. Less frequent rebalancing is among the most common proposals for reforming LETFs. We find the proposal backfires in the counterfactual presented in Panel~B of Table~\ref{tab:schedule}. At a two-day interval, SK~Hynix's loop gain rises to $0.68$, its volatility to $143.2\%$, and its holder loss by one third; Samsung moves in the same direction.

The observation is straightforward from the lens of the loop gain measure. Lengthening the interval further increases it because the fund must respond to a larger cumulative return. Shortening the interval attenuates the loop gain, but it cannot drive the gain to zero due to the unresolved pre-session public news, no matter how finely the trading day is partitioned.\footnote{As $T\to0$, the floor on $\ell(T)/\ell(1)$ is $0.65$ for Samsung and $0.84$ for SK~Hynix.}

\subsubsection{A two-auction rebalance}
\label{sec:twoprintsplit}

The importance of pre-session public news fails a family of design related to rescheduling, as we discussed above. A natural extension is to explicitly address the overnight public news through a dedicated rebalance at market open, so speculators do not have the opportunity of pre-positioning for them. This leads to our first recommended schedule design --- executing the rebalance twice a day at market open and market close. The opening auction absorbs the overnight return, and the closing auction absorbs the in-session return. This arrangement is not equivalent to rebalancing twice as often. A midday rebalance leaves intact the incentive to distort the opening price, whereas participating directly in the opening auction substantially reduces it.

The effectiveness of this policy still relies on one measured quantity: the depth of the opening auction relative to that of the closing auction. Panel~C of Table~\ref{tab:schedule} reports three cases, ordered from optimistic to pessimistic: the opening auction is as deep as the closing auction; its depth responds endogenously to dealer participation; and the fund's own order leaves it unchanged. All three cases deliver gains, the net of two opposing forces: concavity penalizes splitting the order across venues, whereas the dedicated opening trade curbs predation. With an endogenously determined depth, SK~Hynix's volatility drops from $136.7\%$ to $120.5\%$ and its holder loss from \KRW 1.38 trillion to \KRW 592 billion, a reduction of nearly three fifths. For Samsung, volatility falls from $115.3\%$ to $110.3\%$, and the loss falls from \KRW 79 billion to \KRW 51 billion.

However, such design still has two limitations. First, when the overnight and in-session returns have opposite signs, the two orders would run in opposite directions, raising SK~Hynix's total traded value by roughly one half. Second, as with any rescheduling policies, the design reduces holder's loss without addressing its strategic source. Manufacturing leverage remains the same at $\bar A=1.77$.

\subsection{The Size of Rebalance}
\label{sec:flexlev}
\label{sec:threshold}

A second class of policies admits a flexible leverage multiple to make the size of the mandated rebalance more elastic. This was proposed by Hong~Kong's revised circular on July~24, 2026 \citep{sfc2026}, requires products ``with highly dynamic capacity dependent on evolving market conditions'' to adopt a flexible leverage structure under which ``the leverage factor may vary on a daily basis,'' subject to an unchanged maximum of $2\times$. The target for the following day must be published after each close.\footnote{SFC circular 26EC43 does not impose a minimum but ask each product to choose its own permissible range. Effective August~3, 2026, CSOP, the largest provider of leveraged products in Hong Kong, adopted a range of $1.1\times$ to $2\times$ for its twelve single-stock products, including those tracking Samsung and SK~Hynix.} Such flexibility gives managers discretion over daily rebalancing quantities. We therefore evaluate four rules that could govern this discretion.

Because each rule reduces the delivered multiple to some extent, we first isolate the direct effect of lower leverage. The mandated order is quadratic, rather than linear, in the leverage multiple: reducing the leverage multiple by $10\%$ from $2$ to $1.8$ already lowers the loop gain by $28\%$. Panel~A of Table~\ref{tab:flexrules} provides a benchmark fund holding a constant $1.8\times$ leverage multiple. Volatility falls from $136.7\%$ to $132.4\%$, while the holder loss falls from \KRW 1.38 trillion to \KRW 1.17 trillion, and manufacturing leverage falls from $1.77$ to $1.30$. These benefits come with a tracking error of $12.8\%$ each year. We treat the fixed multiple as a benchmark, rather than a policy proposal, because it does not deliver the same leverage exposure as the flexible designs. For the rest of the evaluations, we calibrate the policy parameters so they all deliver the same leverage multiple of $1.8\times$, so any policy impact arises from the implementation itself.

Under the first rule, the fund remains passive during episodes of turbulence or high volatility, but, given the leverage cap, only after positive returns. Indeed, the rule triggers far less rebalancing volume than the $1.8\times$ benchmark and, with it, much smaller manufacturing leverage. The holder loss falls to \KRW 881 billion and volatility to $127.1\%$, removing about a quarter of the excess volatility, at a mandate deviation of $7.8\%$ a year, smaller than under the benchmark. The rule helps because the fund's own rebalances no longer add to the turbulence in a crisis, damping the feedback loop; managers can avoid part of the year's most damaging losses by staying passive through turbulent episodes if granted more flexibility.

The second rule generalizes the first by imposing an inaction band between $1.8\times$ and $2\times$. The fund remains passive while the leverage multiple lies inside the band and rebalances to the nearer boundary once the multiple leaves it. The average leverage multiple again stays at $1.8\times$. This rule reduces volatility by about seven per cent and cuts the holder loss by nearly a third, performing on par with the first; the band is, however, triggered repeatedly during consecutive upturns or downturns of the market, as in the Korean episode, which limits how much of the crisis losses it can avert.

The third rule instead imposes a volatility cap, adjusting the multiple to target constant volatility relative to a rolling anchor. This rule delivers no improvement over the fixed $1.8\times$ benchmark, with both volatility and the holder loss marginally worse. Because changing the multiple itself requires a trade, the rule increases the total order the closing auction must absorb, to more than $40\%$ above the fixed-multiple benchmark, even as it reduces the capital underlying the order. Such an unconditional rule does not mitigate the inelasticity of the order, so it is important for managers to adjust their rebalances according to market conditions rather than to a pre-imposed rule.

The fourth rule is a capacity cap, which constrains the fund's rebalancing order relative to the capacity of the venue that must absorb it. This rule is the strongest of the four: volatility falls to $122.3\%$ and the holder loss to \KRW 743 billion, the largest reduction on both margins, because the cap conditions the order directly on the venue's measured capacity and therefore binds precisely on the stocks and days on which the loop gain is extreme.

This comparison yields three broader conclusions. First, flexibility benefits fund holders overall: three of the four rules we evaluate deliver gains in both volatility and holder welfare. Flexibility makes the rebalancing order more elastic and can avert the most damaging rebalances in a crisis, saving a substantial amount of volatility drag even by simply staying passive. While flexibility comes at the cost of tracking error, that cost can be disclosed in advance in the fund prospectus. If the fund markets itself as a $1.8\times$ fund, the tracking error would in fact be minimal. Second, the majority of the holder loss occurs in episodes of extreme volatility, that is, in the crisis regime. Limiting the size of the rebalance or staying passive on these occasions appears to be first-order for protecting holders. Finally, the evaluation we provide here offers only a few prototypes, and fund managers may well design more effective policies that balance holder benefits against tracking error, whether by combining these rules or by applying their own discretion. We therefore expect the Hong~Kong regulators' policies to have positive effects on the market.

\subsection{The Reference Price}
\label{sec:twoprint}
\label{sec:reference_benchmark}

The LETF rebalances at the market close according to one single reference price. This convention minimizes measured tracking error, but it also invites predators to distort the price against the LETF. The reference price is thus the margin through which policy can directly affect the strategic channel without reducing the nature of the product, yet regulators have paid comparatively little attention to this dimension. Table~\ref{tab:reference} evaluates two alternatives: executing a smaller share of the order in the auction that determines the reference price, and basing the mandate on an average of multiple prices.

\subsubsection{Splitting execution}

LETFs typically implement part of the rebalance in the closing auction and the other part in the futures market that closes a few minutes after the auction, as we introduced in Section~\ref{sec:close}.\footnote{The practical execution share reflects funding constraints, the securities transaction tax on spot but not futures transactions, and short-sale restrictions affecting inverse funds (Appendix~\ref{app:phi}). SK~Hynix's rebalancing capital is dominated by a swap product whose counterparty hedges at the close.} The first design thus changes the share of the mandated order that executes in the price-setting auction. Let $\varphi$ denote this share. To fix ideas, we assume that the remaining orders execute in a venue identical to the closing auction except that they do not participate in the determination of the closing price, so the results are not driven by differences in depth across venues.

Panel~A of Table~\ref{tab:reference} considers alternative values of $\varphi$. The effect reflects a trade-off between less self-referential pricing and concave price impact and is therefore non-monotonic in $\varphi$, as we established theoretically in Proposition~\ref{prop:phi}. The trade-off mostly cancels out in at the calibrated parameters, with different $\phi$ parameters barely changing holder's loss. Less self-referential pricing can mitigate the excess volatility, but the magnitude of this policy is limited. Mandate tracking error remains zero in every case because $\varphi$ changes only where the order is executed; the fund continues to deliver twice the same daily return.

However, the choice of $\varphi$ is no free lunch for fund managers. The mixture of stocks and futures reflects other considerations, such as collateral management. Moreover, our counterfactuals still price the off-auction impact at the auction's favorable depth, so pragmatic considerations would de-prioritize $\varphi$ as an approach to address the loop gain problem given. We therefore do not make strong claims about the effect of this policy.

\subsubsection{Averaging the reference}

The second design changes the reference price while leaving execution at the market close. The mandate is based on an average of several prices rather than on a single closing print, so it alters the target rather than the flow. Because part of that reference is already fixed, the target a trader faces is smaller and less responsive to the current price: displacing one print moves the reference by only a fraction of the original amount, which weakens the self-reinforcing rebalance that predatory trading exploits. This attenuation is arithmetic in the number of prices averaged and, unlike the designs considered above, does not depend on the depth of any alternative venue. As a comparison, reducing the size of the product lowers manufacturing leverage by limiting the rebalance size, while rescheduling the order leave the target just as predictable as before.

Panel~B of Table~\ref{tab:reference} evaluates several averaging rules. Averaging prices within the trading day fails for the same reason rescheduling does. The overnight news is already public at the open and is impounded in every price thereafter, as Figure~\ref{fig:loading_cycle} establishes, so each of those prices already carries the speculators' pre-positioning. Intraday averaging therefore offers no refuge; at the calibration the same-day averaged mark admits no damped equilibrium at all, the open's own anticipation loop passing the ringing bound (Table~\ref{tab:reference}).

A reference averaged across days, by contrast, leaves the self-reinforcing rebalance far less responsive to any last-minute price distortion. Averaging therefore reduces both the loop gain and the manufacturing leverage substantially, almost in proportion to the number of prices averaged. Under the most favorable rule, using a five-day moving average closing price as the reference price, SK~Hynix's volatility falls from $136.7\%$ to $110.7\%$, eliminating $71\%$ of the excess volatility, and the holder loss falls from \KRW 1.38 trillion to \KRW 161 billion; a three-day average already removes $56\%$ of the excess volatility and $83\%$ of the loss.

One might expect such a policy to produce large tracking errors, because it continually tracks a stale benchmark. This is true at the daily frequency: the averaging rules deliver daily deviations of about $8.5\%$ against the original mandate, the highest of any design we evaluate. These deviations are mean-reverting rather than cumulative, however, so they largely cancel out across days. Over the entire sample period, the mandate tracking error is $3.0$ annualized percentage points for the three-day moving average and $17.3$ for the five-day. We therefore recommend the three-day average, which delivers better performance over longer horizons.

\subsection{Policy Recommendation}
\label{sec:policy}
\label{sec:recommendation}

The counterfactuals identify three effective designs, one from each policy family. Each operates on a distinct margin, so the policies are complements rather than substitutes. Used together, these three policies address the location, scale, and reference-price margins of predatory trading around LETFs and substantially reduce the losses borne by retail investors.

The two-auction rebalance determines where the order arrives. It reduces holder losses by nearly one half while generating exactly zero mandate tracking error, because the fund continues to deliver twice the same daily return. Among the three, it is the only reform that fully preserves investors' contracted exposure and is therefore the natural first step.

A flexible multiple adjusts \emph{how large} the order can become. The best size rule, a cap tied to the absorbing venue's measured capacity, removes about $40\%$ of the excess volatility and $45\%$ of the welfare transfer, because it binds precisely on the stocks and trading days on which the loop gain is extreme.

A multi-day average alters \emph{which price} sizes the order. It reduces manufacturing leverage in proportion to the number of prices averaged, independently of venue depth: a three-day average cuts the loop gain to one-third of its original value and the holder loss to \KRW 237 billion, the largest effect of the three designs. It is also the only policy that directly weakens the incentive to pre-position, because part of the reference price a trader would need to move is already fixed.


\section{Conclusion}
\label{sec:conclusion}

A leveraged exchange-traded fund submits a demand schedule into the closing auction that is upward-sloping in price. The
consequences are governed by one measurable scalar, the loop gain $\ell = \Lambda_c K$: the price impact of absorbing order imbalance at the reference price, multiplied by the total rebalancing capital of all related funds. For any $\ell > 0$, the closing price overweights public news, and the excess loading reverses during the following cycle. The pattern was only documented among Korean stocks with LETF offerings after May~27, 2026, where the loop gain runs an order of magnitude above the index complexes behind the benign evidence and more than double the most extreme United States single-stock complex. In a replay without the self-reinforcement, the arrangement added approximately 36.7 percentage points to SK~Hynix's annualized volatility and cost its products' holders \KRW 1.38 trillion (US\$0.91~billion) over nine weeks, or 17.6\% of a unit's initial investment. The mechanism nevertheless destabilized the path of the July repricing rather than its destination.

Two policy implications follow. The first concerns the market into which such an instrument is introduced. The mandated rebalance is fixed by prospectus and identical across markets, but the venue obligated to absorb it may have different capacities. Measured against that mechanism, the loop gain of the largest United States single-stock complexes lies between $0.06$ and $0.22$, while SK~Hynix's post-launch loop gain averaged $0.51$ on the same cross-market construction. The launch of Korean single-stock LETFs accordingly generated far greater damage to financial stability than any single-stock LETF has caused in the United States. The governing quantity is thus a combination of the product's size and the clearing venue's capacity, so regulators in emerging markets should investigate that capacity directly rather than infer it from an instrument's success in a developed market.

The second concerns the scale of the instrument, which governs both the sign and magnitude of its effect. At loop gains near zero, the traders who meet the fund's predictable order supply liquidity to it, damping the
closing displacement and stabilizing the market. As the loop gain rises, the same behavior by the same agents
widens the displacement it formerly damped, because a share purchased at the close now enlarges the order into which it will be sold. The transition requires only growth in $K$, with no change in any agents' conduct, which the
product's own success generates: Korea's rose from \KRW 4.3 to approximately \KRW 14 trillion within three weeks of listing. The response is in addition nonlinear in $\ell$, so evidence accumulated while an instrument is small
does not extrapolate to the same instrument once it is large.

More broadly, the mechanism we proposed is not specific to leverage itself as in the intermediary asset pricing literature \citep{hekrishnamurthy2013, brunnermeiersannikov2014}. Any mandate that could lead to upward-sloping demand, as previously documented in firesales \citep{shleifervishny2011}, portfolio insurance \citep{grossmanzhou1996}, and volatility targeting strategies \citep{moreiramuir2017}, carries the same opportunity for predatory trading in the spirit of \cite{brunnermeierpedersen2005}. The same measurements of loop gain quantify those opportunities: the size of the rebalance capital and the capacity of the venue that clears it. This paper documents the example of leveraged products, where the loop gain is large enough to generate instability in the Korean market.

We cannot conclude whether any of the conduct documented here was manipulative or unlawful. Manipulation is a claim about intent; ours is a claim about prices and aggregate quantities, and the second does not establish the first. At the measured loop gain, the observable strategy---pre-positioning during the session and liquidating at the close---is qualitatively identical to that of a benign liquidity provider or an intraday momentum or reversal factor trader. Moreover, the same aggregate response can arise from a continuum of atomistic traders, each too small to move the price and each individually incapable of manipulation, so the mechanism requires no single trader to hold market power. What separates these motives is intent, which only account-level records of the closing auction---orders, cancellations, and identified accounts---can reveal, and those records are a supervisory instrument rather than a research one. We accordingly document a mechanism and quantify the transfer it generates, taking no position on the legality of any participant's conduct.

\clearpage
\bibliographystyle{econ-aea}  
\bibliography{references}     

\clearpage

\section*{Tables and Figures}
\addcontentsline{toc}{section}{Tables and Figures}

\begin{table}[H]
\centering
\caption{\captitle{Overnight reversal of pre-open foreign news} The table presents the loading of Korean stock returns on a pre-open U.S. return, the public-news shock $\varepsilon$.
The columns present the shock instrument, either SMH (semiconductors ETF), QQQ (broad technology ETF), or SPY (broad market ETF), for the \emph{next-day}  close-to-close return $\ln(C_{t+1}/C_t)$ in the left block and its \emph{overnight} leg $\ln(O_{t+1}/C_t)$ in the right, with the cell's observation count.
The rows present the four cells of the difference-in-differences design, the treated
pair (Samsung Electronics, SK~Hynix) against matched large-cap controls, before versus after the May~27, 2026, domestic launch, followed by the triple difference that pools the four in one regression. ``Pre'' is the
offshore-products-only window. 
The last three rows present the same triple differences expressed as shares of the first-day response to the news, measured as the loading of the day-$t$ opening leg $\ln(O_t/C_{t-1})$ in the treated post-launch cell. The opening leg shares no price mark with either numerator, and its own triple difference is zero, so it scales the reversal without being moved by the launch.
The third leg of the cycle, the next session, which
Proposition~\ref{prop:chaining} requires to carry $\Pi_g-1$ times the overnight one, is in Appendix Table~\ref{tab:window_split}; the earlier
no-product window and the same-sector triple difference are in Appendix Tables~\ref{tab:echo_full} and~\ref{tab:echo_tech}.}
\label{tab:echo}
\medskip
\small
\setlength{\tabcolsep}{4.5pt}
\begin{tabular}{lccccccr}
\toprule
 & \multicolumn{3}{c}{Next-day (total), $\ln(C_{t+1}/C_t)$} &
   \multicolumn{3}{c}{Overnight leg, $\ln(O_{t+1}/C_t)$} & \\
\cmidrule(lr){2-4}\cmidrule(lr){5-7}
 & SMH & QQQ & SPY & SMH & QQQ & SPY & Obs. \\
\midrule
Treated $\times$ post &
  $-0.83^{***}$ & $-1.76^{***}$ & $-3.92^{***}$ &
  $-0.44^{**}$  & $-1.05^{**}$  & $-2.88^{***}$ & 82 \\
 & (0.29) & (0.66) & (1.34) & (0.21) & (0.43) & (0.93) & \\
Treated $\times$ pre &
  $-0.07$ & $-0.09$ & $-0.07$ &
  $-0.07$ & $-0.10$ & $-0.10$ & 365 \\
 & (0.15) & (0.26) & (0.37) & (0.09) & (0.15) & (0.24) & \\
Control $\times$ post &
  $0.17$ & $0.23$ & $-0.27$ &
  $0.07$ & $0.11$ & $-0.18$ & 410 \\
 & (0.13) & (0.23) & (0.42) & (0.06) & (0.13) & (0.23) & \\
Control $\times$ pre &
  $0.03$ & $0.02$ & $0.02$ &
  $-0.01$ & $-0.03$ & $-0.04$ & 5{,}430 \\
 & (0.04) & (0.07) & (0.09) & (0.02) & (0.04) & (0.06) & \\
\midrule
Triple difference &
  $-0.79^{***}$ & $-1.55^{***}$ & $-2.29^{***}$ &
  $-0.32^{***}$ & $-0.79^{***}$ & $-1.50^{***}$ & 7{,}008 \\
\quad (treated $\times$ post)
 & (0.19) & (0.46) & (0.76) & (0.11) & (0.23) & (0.40) & \\
\midrule
First-day response &
  $+0.99$ & $+2.06$ & $+3.35$ &
  $+0.99$ & $+2.06$ & $+3.35$ & 88 \\
\quad at market open & (0.24) & (0.52) & (1.04) & (0.24) & (0.52) & (1.04) & \\
Share reverted &
  $80\%$ & $75\%$ & $68\%$ &
  $32\%$ & $38\%$ & $45\%$ & \\
\bottomrule
 \multicolumn{8}{c}{Standard errors clustered by date in parentheses, $^{*}$, $^{**}$, $^{***}$ for $10\%$, $5\%$, $1\%$.}
\end{tabular}
\end{table}

\begin{table}[H]
\centering
  \caption{\captitle{Calibrated parameters}
The table presents the calibration of the model of Section~\ref{sec:model} for
Samsung and SK~Hynix. Panel~A presents the parameters: the six top rows are
measured outside the model, and the arbitrage capacity is identified from
the overshoot factor $\Pi_g$. Panel~B compares the targeted and untargeted
moments in the data with those the model implies; its brackets are $95\%$
intervals, date-block bootstrap except the reversal share, which is
delta-method. See Table~\ref{tab:g4calibration} for the curve's own primitives,
and Appendix Table~\ref{tab:calibration_full} for the robustness of the measured
rows and the share of bootstrap draws at the zero-capacity boundary.}
\label{tab:calibration}
\medskip
\footnotesize
\begin{tabular}{>{\raggedright\arraybackslash}p{5.0cm}lcc>{\raggedright\arraybackslash}p{5.1cm}}
\toprule
 & & Samsung & SK Hynix & Identified from \\
\midrule
\multicolumn{5}{l}{\textit{Panel A. Parameters}} \\
loop gain, post-launch mean & $\bar\ell = \Lambda_c K$ & 0.219 & 0.583 & impact curve $\times$ worldwide $K$, unweighted \\
impact concavity & $\hat\delta$ & \multicolumn{2}{c}{0.73} & exponent of the impact curve \\
correction speed & $\hat\theta$ & \multicolumn{2}{c}{0.88} & decay of the echo's legs \\
news / liquidity variance ratio & $\sigma^2_\varepsilon/\sigma^2_u$ & \multicolumn{2}{c}{1.07} & pre-launch window \\
noise / liquidity ratio, post-launch & $\lambda_{\mathrm{post}}$ & \multicolumn{2}{c}{1.06} & session-residual variance identity \\
closing-book depth ratio & $\hat\rho = \Lambda_c/\Lambda$ & \multicolumn{2}{c}{0.43} & auction vs.\ continuous impact at matched size \\
\addlinespace
arbitrage capacity & $\hat\psi_c^{\,i} = \hat\tau\,\Lambda_c^i$ & 0.466 & 0.260 & one pooled $\hat\tau$ inverted from $\hat\Pi_g$ \\
\bottomrule
\end{tabular}

\medskip

\begin{tabular}{p{7.6cm}ccc}
\toprule
\multicolumn{4}{l}{\textit{Panel B. Moments: data versus model}} \\
\addlinespace
Moment & Data & Model & Targeted \\
\midrule
pooled overshoot factor $\hat\Pi_g$ (reference print) & 1.66 [0.84, 3.27] & 1.66 & yes \\
\quad at the actual daily gaps & 1.66 [0.84, 3.27] & 1.37 & no \\
\quad per name, SK~Hynix & 1.60 [0.72, 2.89] & 1.69 & no \\
\quad per name, Samsung (attenuation prediction) & 1.75 [0.85, 3.96] & 1.18 & no \\
\addlinespace
closing-return autocorrelation $\rho_1$, SK~Hynix & -0.16 & -0.14 & no \\
\quad Samsung & -0.08 & -0.06 & no \\
pre-launch $\rho_1$ & 0.02 [-0.08, 0.12] & 0 & no \\
pre-launch variance ratio VR(5) & 1.00 & 1 & no \\
overnight share of daily variance, post & 0.62 & 0.51 & no \\
total\,/\,overnight variance ratio, post & 1.73 & 2.48 & no \\
Table~\ref{tab:echo} reversal share & 0.75 [0.17, 1.32] & 0.31 & no \\
next-session\,:\,overnight leg, $\hat\Pi_g - 1$ & 0.68 & 0.69 & no \\
\bottomrule
\end{tabular}
\end{table}

\begin{table}[H]
\centering
\caption{\captitle{Policy counterfactuals on order scheduling, Korea's proposal} 
Focusing on SK~Hynix, this table evaluates designs that move the order temporally by varying its rebalance interval, order splitting and auction participation.
Panel~A considers $m$ equally sized funds, each rebalancing at its own intraday moment against its own reference price, a proposal suggested by Korean regulators. The panel prices the design at its most favorable---every auction granted the closing auction's full attending capacity, with the orders split on the measured concavity $\hat\delta$.
Panel~B rebalances every $T$ days on the cumulative $T$-day return.
Panel~C splits the order across two reference prices, the opening auction trading the overnight gap and the closing auction the session return. 
The columns present the flow-weighted loop gain $\bar\ell$; the manufacturing leverage $\bar A$ of Proposition~\ref{prop:manufactured}; annualized volatility; the \emph{mandate tracking error}, the total departure over the $46$ post-launch days between the exposure the fund delivers and the $2\times$ its holders bought, measured by compounding both and comparing terminal values so that offsetting daily misses cancel; and the holder loss over the same window, the realized Proposition~\ref{prop:loss} transfer $\sum_t K_t r_{cc,t} e_t$ together with the execution cost of the futures leg, each measured net of the same design's $\ell\to0$ benchmark, so that the column reports what the loop causes rather than what running a leveraged fund costs (positive $=$ holders lost). 
$^{\S}$ The compounded loading passes Proposition~\ref{prop:chaining}'s ringing bound on some days, and those days are held at the largest dose the model can price, so the entry is a lower bound.
See Appendices~\ref{app:designs} and~\ref{app:replay} for detailed methodology.
Evaluations on Samsung are shown in Appendix Table~\ref{tab:schedule_both}.}
\label{tab:schedule}
\medskip
\footnotesize
\setlength{\tabcolsep}{5pt}
\begin{tabular}{lrrrrr}
\toprule
 & $\bar\ell$ & manufacturing & vol & mandate & holder loss \\
 &  & leverage & (\%) & TE (\%) & (\KRW bn) \\
\midrule
\multicolumn{6}{l}{\textit{Benchmarks}} \\
status quo$^\dagger$ & 0.583 & 1.77 & 136.7 & 0.0 & 1{,}380 \\
no loop ($\ell\to0$) & 0.000 & 0.00 & 100.0 & 0.0 & 0 \\
\addlinespace
\multicolumn{6}{l}{\textit{Panel A. Staggered reference prints: $m$ equal funds, each at its own intraday moment and print}} \\
$m = 2$ & 0.351 & 1.77 & 146.9 & 0.0 & 1{,}858 \\
$m = 3$ & 0.261 & 1.77 & 152.0$^{\S}$ & 0.0 & 2{,}098$^{\S}$ \\
$m = 5$ & 0.180 & 1.77 & 155.7$^{\S}$ & 0.0 & 2{,}235$^{\S}$ \\
\addlinespace
\multicolumn{6}{l}{\textit{Panel B. The rebalancing interval: every $T$ days, on the cumulative $T$-day return}} \\
$T=2$ & 0.675 & 1.77 & 143.2 & -- & -- \\
\addlinespace
\multicolumn{6}{l}{\textit{Panel C. Two reference prints: the order split between the opening and closing auctions}} \\
equal depth & 0.665 & 1.77 & 118.2 & 0.0 & 497 \\
endogenous attendance & 0.738 & 1.77 & 120.5 & 0.0 & 592 \\
pre-launch depth & 0.736 & 1.77 & 121.0 & 0.0 & 635 \\
\bottomrule
\end{tabular}
\end{table}

\begin{table}[H]
\centering
\caption{\captitle{Policy counterfactuals on rebalance flexibility, Hong~Kong's proposal} This table evaluates Hong Kong regulators' proposal of allowing flexible leverage multiples for LETFs. 
\emph{Panel~A} presents a fixed leverage multiple in the status quo and one alternative multiple of 1.8 as benchmark for evaluations below. 
\emph{Panel~B} evaluate a few flexible proposals, all holding the \emph{band} at $1.1\times$--$2\times$, the structure Hong~Kong's SFC circular 26EC43 \citep{sfc2026}, and with an average leverage multiples of 1.8. The rules managers could potentially adopt are listed as follows: restore $2\times$ at every close \emph{except} after an up move in stress; hold a no-trade band, reflected to the nearer edge on exit; target constant volatility against a rolling $250$-day anchor; or target capacity, keeping the predicted order inside a fixed share of the closing auction.
Every rule carries one free threshold, which is calibrated so that the rule delivers the same leverage multiple of $1.8\times$.
In terms of columns, total close flow is the whole order the closing auction receives, as a ratio to the status quo at each design's own equilibrium; the manufacturing leverage is the Table~\ref{tab:schedule} object, which these rules reach only through the feedback capital, since none of them touches the reference price; the \emph{daily miss} is $\mathrm{mean}\,|(L_t-2)\,r_{cc,t}|$, the average size of one day's departure from the mandate; the mandate tracking error and the holder loss are defined in the same way as Table~\ref{tab:schedule}.
See Appendix~\ref{app:replay} for more details on methodology.}
\label{tab:flexrules}
\medskip
\footnotesize
\setlength{\tabcolsep}{3pt}
\begin{tabular}{lrrrrrrr}
\toprule
 & mean & total close & manufacturing & vol & daily & mandate & holder loss \\
rule & $L$ & flow & leverage & (\%) & miss (\%) & TE (\%) & (\KRW bn) \\
\midrule
\multicolumn{8}{l}{\textit{Panel A. A fixed multiple $L$: the leverage CAP}} \\
\quad $L = 2\times$$^\dagger$ & 2.00 & 1.00 & 1.77 & 136.7 & 0.00 & 0.0 & 1{,}380 \\
\quad $L = 1.8\times$ & 1.80 & 0.72 & 1.30 & 132.4 & 1.28 & 12.8 & 1{,}173 \\
\addlinespace
\multicolumn{8}{l}{\textit{Panel B. A flexible band $1.1\times$--$2.0\times$: what the manager's own rule delivers inside it}} \\
\quad Passive on up moves in stress & 1.95 & 0.54 & 0.90 & 127.1 & 0.19 & 7.8 & 1{,}044 \\
\quad Inaction band $1.8\times$--$2\times$ & 1.91 & 0.30 & 0.68 & 127.4 & 0.50 & 15.1 & 974 \\
\quad Volatility cap ($1.40\times$ long-run vol) & 1.80 & 1.04 & 1.26 & 132.8 & 1.08 & 11.8 & 1{,}187 \\
\quad Capacity cap ($S \le 1.61$) & 1.80 & 0.57 & 0.44 & 122.3 & 0.85 & 15.1 & 743 \\
\addlinespace
\bottomrule
\end{tabular}
\end{table}

\begin{table}[H]
\centering
\caption{\captitle{Policy counterfactuals on reference price} 
The table evaluates the designs that affect the
strategic channel without shrinking the product, also focusing on SK~Hynix.
\emph{Panel~A} varies the venue split $\varphi$, the share of the worldwide mandate settling in the cash closing auction rather than in single-stock futures, holding the mandate itself fixed; the measured $\varphi = 0.49$ is the status quo and reproduces the realized replay.
\emph{Panel~B} changes the reference price itself: the average of the two same-day auction prices, or
a $k$-day moving average of closes. The three-day average design was proposed by Korea Capital Market Institute July~20, 2026.
Columns and method are the same as in Table~\ref{tab:schedule}, with one exception: the daily miss and the mandate tracking error are the Table~\ref{tab:flexrules} columns applied to a different departure, since the multiple is held at $2\times$ here and the settlement mark leaves the day's close, so the design delivers $2\,dR_t$ for reference return $dR_t$. 
Samsung is priced on every design in Appendix Table~\ref{tab:reference_both}.}
\label{tab:reference}
\medskip
\footnotesize
\setlength{\tabcolsep}{5pt}
\begin{tabular}{lrrrrrr}
\toprule
 & $\bar\ell$ & manufacturing & vol & daily & mandate & holder loss \\
 & & leverage & (\%) & miss (\%) & TE (\%) & (\KRW bn) \\
\midrule
\multicolumn{7}{l}{\textit{Panel A. The venue split: share $s$ of the mandate settling in the closing auction, the rest in futures}} \\
$s = 0.25$ & 0.356 & 1.77 & 126.3 & 0.0 & 0.0 & 1{,}332 \\
$s = 0.49$$^\dagger$ & 0.583 & 1.77 & 136.7 & 0.0 & 0.0 & 1{,}380 \\
$s = 0.75$ & 0.795 & 1.77 & 145.1 & 0.0 & 0.0 & 1{,}356 \\
\addlinespace
\multicolumn{7}{l}{\textit{Panel B. Reference averaging: the mandate settles on an average of several prints}} \\
same-day average (open, close)$^{\ddagger}$ & 0.291 & 1.49 & \multicolumn{4}{c}{\textit{no damped eq.}} \\
MA(3) settlement mark & 0.194 & 0.59 & 116.2 & 8.8 & 3.0 & 237 \\
MA(5) settlement mark & 0.117 & 0.35 & 110.7 & 8.4 & 17.3 & 161 \\
\bottomrule
\end{tabular}
\end{table}

\clearpage

\begin{figure}[H]
\centering
\includegraphics[width=\textwidth]{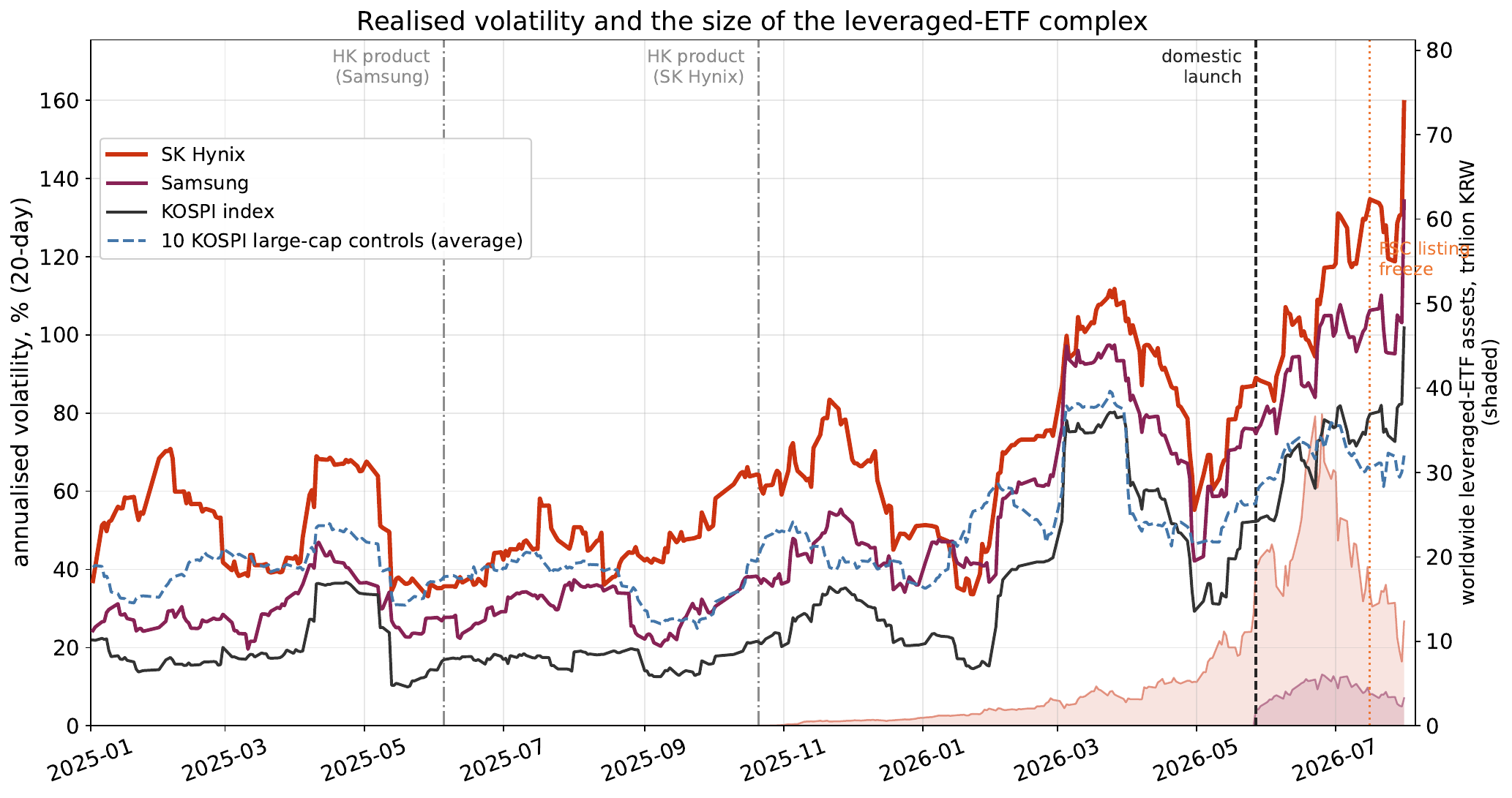}
\caption{\captitle{The Korean episode, May--July 2026} The figure plots the rise of annualized volatility of Korean stocks (Samsung, SK~Hynix, KOSPI index, a few other control stocks) in 2026, together with the worldwide rebalancing capital $K$ of LETFs with the two Korean underlyings. Vertical lines mark the domestic launch of the sixteen single-stock leveraged and inverse ETFs on May 27, 2026 as well as earlier launches of Hong Kong products.}
\label{fig:timeline}
\end{figure}

\begin{figure}[H]
\centering
\includegraphics[width=\textwidth]{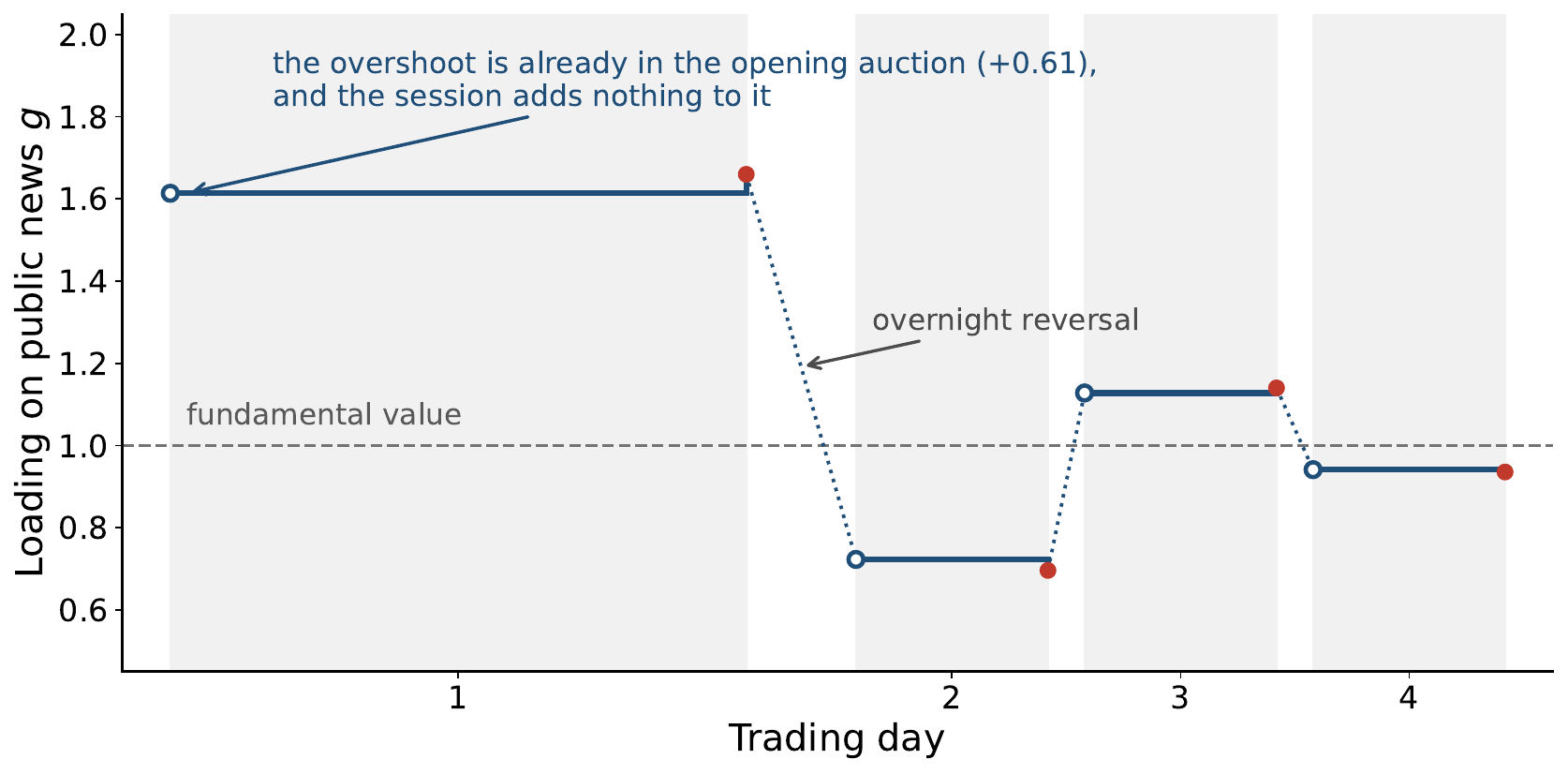}
\caption{\captitle{The price path after one unit of public news}
The figure plots the price's loading on public news after one unit of news
arrives before the opening auction of day~1, with no news thereafter, so the path
is deterministic. The loading steps up at the opening auction and is flat across
the session, and the closing auction adds almost nothing; overnight capital then
corrects a share $\theta$ of the standing error, and the price oscillates around
fundamental value with geometric decay. The shaded bands are trading sessions and
the dotted segments are overnight. The open circles represent opening prices and
the filled markers closing prices. Day~1 occupies half the horizontal axis so that
the within-day shape is legible, so the axis is not uniform in time. The
parameters are the calibration of Section~\ref{sec:calibration},
$\hat\Pi_g = 1.66$, $\hat{\mathcal{S}} = 0.93$ and $\hat\theta = 0.88$.}
\label{fig:loading_path}
\end{figure}

\begin{figure}[H]
\centering
\includegraphics[width=\textwidth]{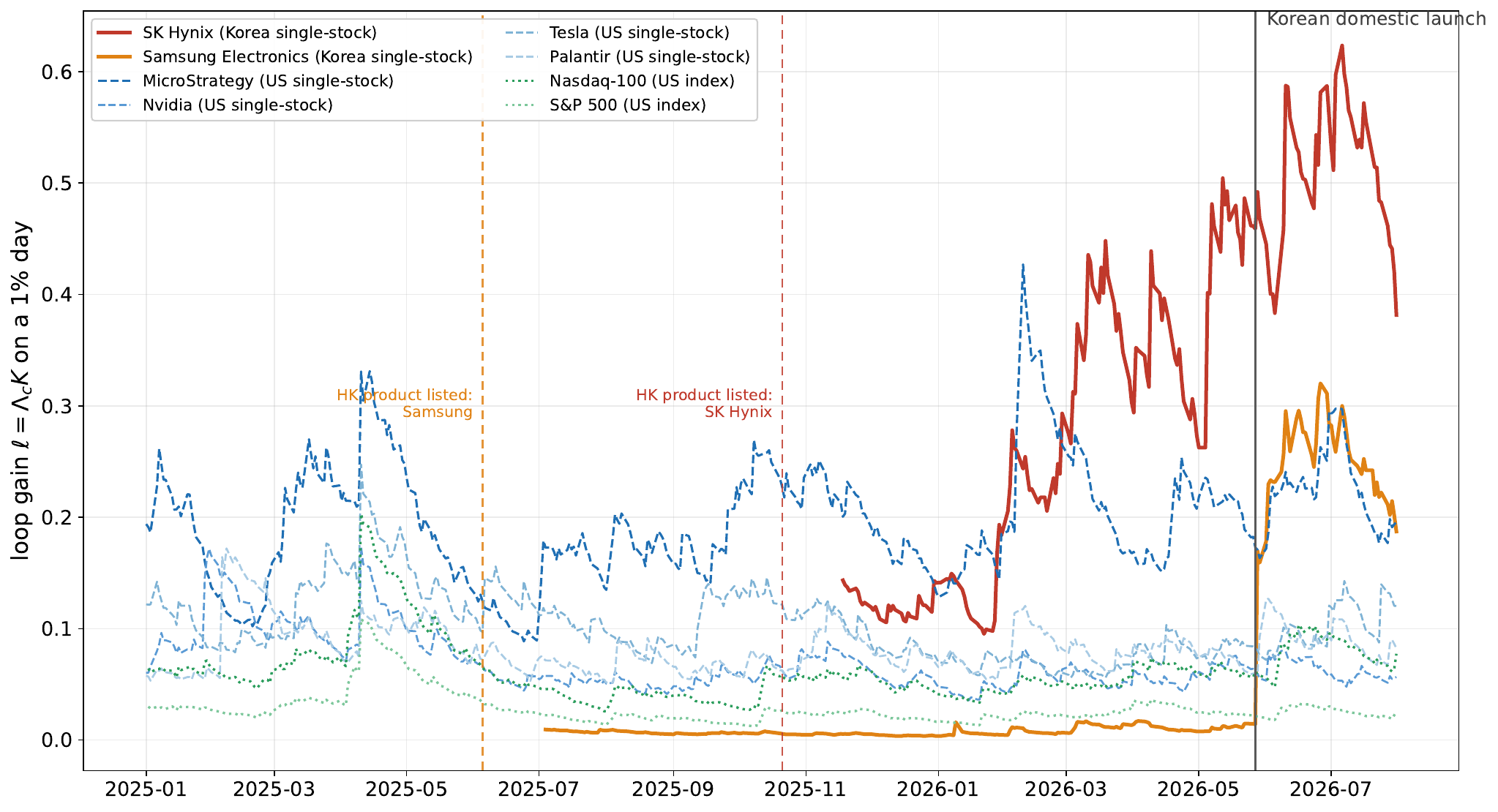}
\caption{\captitle{The loop gain over time, Korea versus the United States}
This figure plots the time series of loop gain for the two Korean treated stocks against the U.S. single-stock and index complexes computed identically. Loop gain is defined as$\ell = \Lambda_c K$, closing-auction price impact times worldwide rebalancing capital, introduced detailedly as a sufficient statistic of Section~\ref{sec:model}. 
$\Lambda_c$ is the tape-estimated impact curve of Section~\ref{sec:calibration}, measured from signed flow arriving before the close through to the close, its shape fitted on the Korean intraday data.
$K$ is worldwide rebalancing capital at the Korean closing auction.
Vertical lines mark the domestic launch of the sixteen single-stock leveraged and inverse ETFs on May 27, 2026 as well as earlier launches of Hong Kong products.
Appendix Table~\ref{tab:impact_menu} compares these estimates with the Amihud proxy.}
\label{fig:dose}
\end{figure}

\begin{figure}[H]
\centering
\begin{minipage}[t]{.48\textwidth}
\centering
\includegraphics[width=\linewidth]{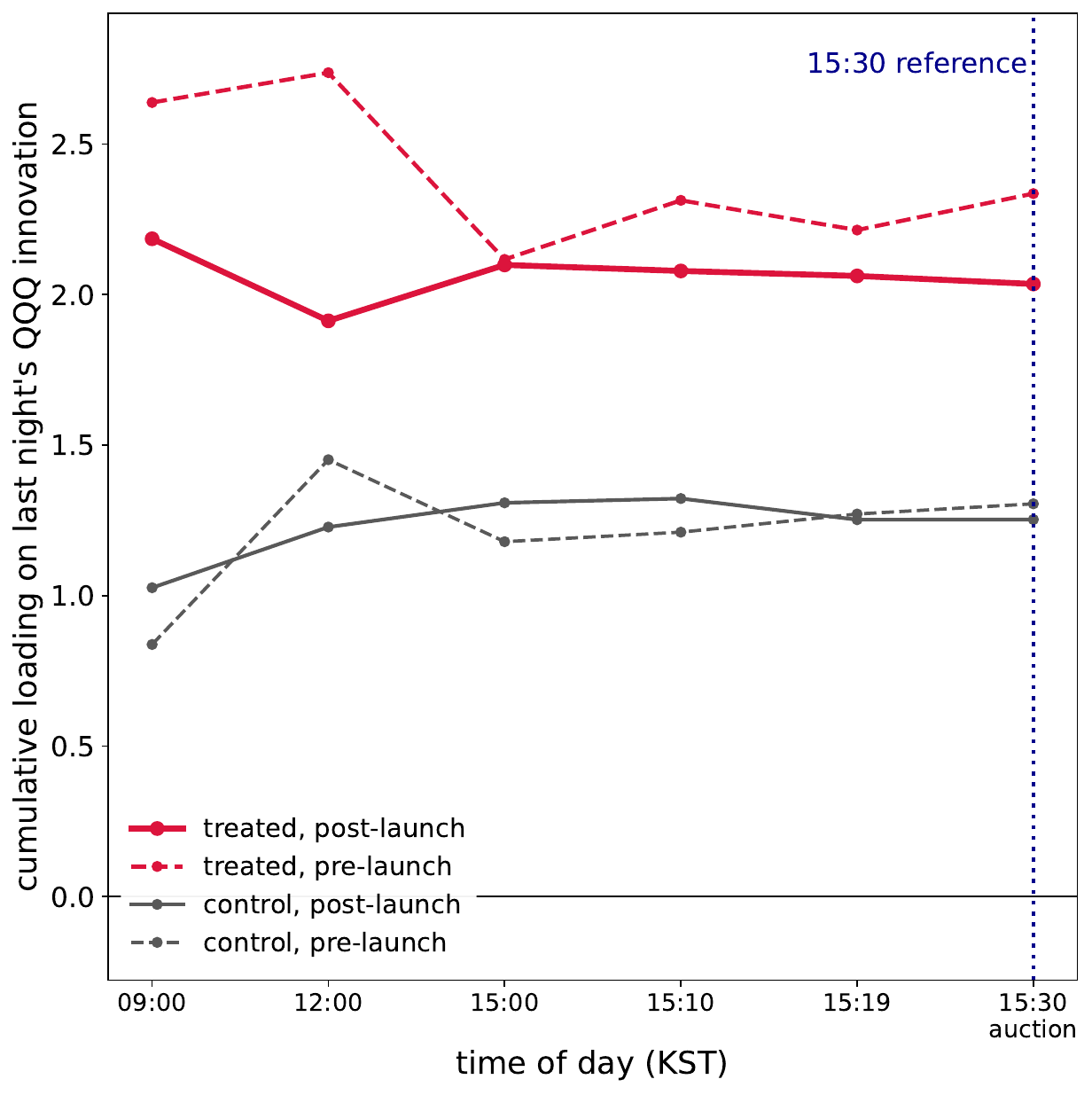}
\par\smallskip {\small (a) Within day $t$}
\end{minipage}\hfill
\begin{minipage}[t]{.48\textwidth}
\centering
\includegraphics[width=\linewidth]{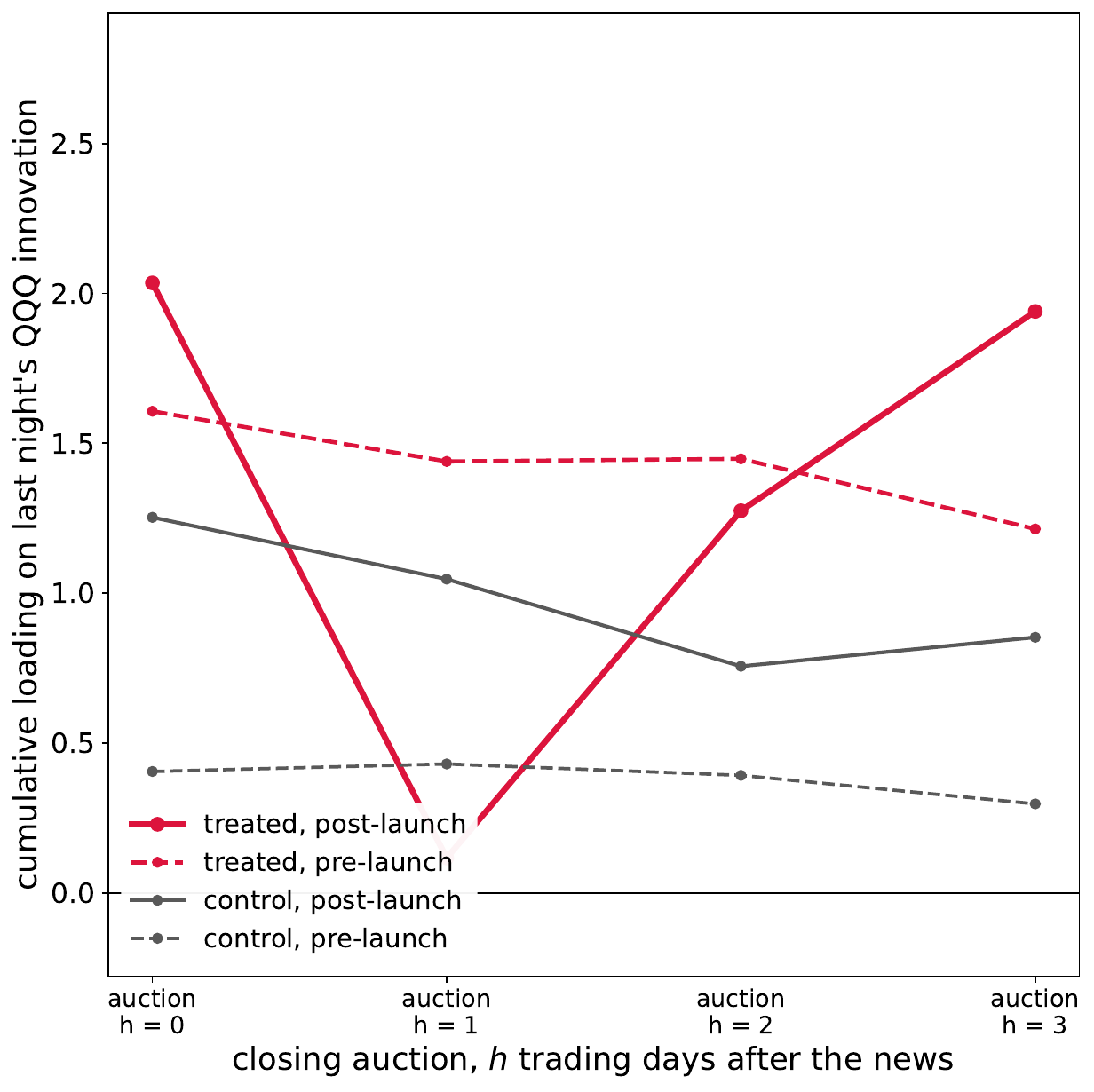}
\par\smallskip {\small (b) Across auction closes}
\end{minipage}
\caption{\captitle{The news-loading cycle}
The figure plots the loading of the treated and control stocks' cumulative return, measured from the pre-news close, on the overnight U.S.\ technology (QQQ) return, before and after the launch.
Panel~(a) presents day $t$, from the open to the 15:30 closing auction that sets the LETF reference price; all four cells are estimated on the May--July 2026 minute panel, whose pre-launch window holds fourteen trading days.
Panel~(b) presents the closing auctions of days $t$ through $t{+}3$; the post-launch cells stay on the 2026 panel, while the pre-launch cells are estimated on the full pre-launch daily history, the offshore-products window for the treated pair and 2024--26 for the controls.
The coefficients are local projections \citep{jorda2005} on the innovation in the overnight return, conditioning on two lags of the shock and of the stock's own return, with stock fixed effects and standard errors clustered by date.
After the launch the treated pair carries roughly twice the controls' loading at the reference price and gives the whole of it back by the next close, while every other cell holds its loading.
The overnight and next-session legs of the unwinding are in Appendix Table~\ref{tab:window_split}.}
\label{fig:loading_cycle}
\end{figure}

\begin{figure}[H]
\centering
\includegraphics[width=.95\textwidth]{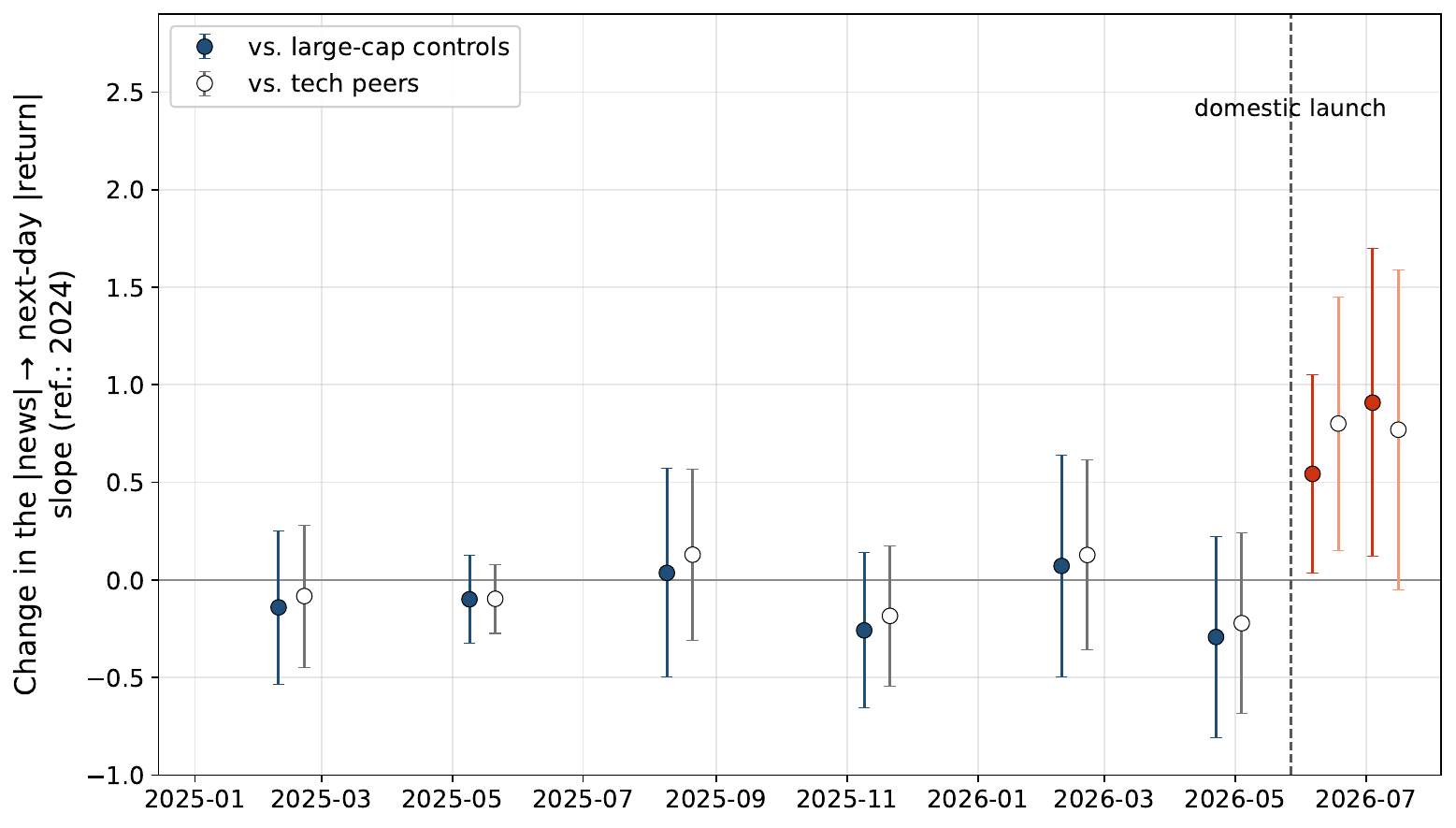}
\caption{\captitle{Parallel trends for the aftershock design}
This figure plots the dynamic treatment effects of LETF launch on conditional volatility in Korea.
The vertical axis illustrates the treated-versus-control difference in the $|$news$|$-to-next-day-movement slope of Appendix Table~\ref{tab:aftershock}, the magnitude counterpart of the signed echo of Table~\ref{tab:echo}. The filled marker (blue) represents the treatment effect against ten never-eligible large-cap stocks in Korea, while the open marker (white) represents that against the ten technology peers in Korea. Whiskers are
$95\%$ intervals. The coefficients are estimated with stock and date fixed effects and the lead-shock controls, standard errors clustered by date.
The pre-launch path is jointly null in both universes (Wald $p = 0.78$ and $0.71$), so the identifying assumption of Section~\ref{sec:empirics} is not rejected.}
\label{fig:aftershock}
\end{figure}



\clearpage

\appendix
\setcounter{table}{0}\renewcommand{\thetable}{A\arabic{table}}
\setcounter{figure}{0}\renewcommand{\thefigure}{A\arabic{figure}}

\begin{center}
{\Large\bfseries Online Appendix for}\\[6pt]
{\Large\bfseries ``Preying on Leveraged ETFs''}\\[14pt]
{\large Yinhong Zhao\footnote{This online appendix was generated by
Claude Opus 5, cross-checking the original code of the paper.}}\\[3pt]
{\normalsize Princeton University}\\[10pt]
{\normalsize This draft: August 2026}
\end{center}
\bigskip

\noindent This online appendix collects the data documentation and
validation (Appendix~\ref{app:institutions}), the proofs of every result
stated in the paper (Appendix~\ref{app:proofs}), theoretical extensions
(Appendix~\ref{app:theory}), and the details of the quantification
(Appendix~\ref{app:quant}). All appendix tables and figures are collected
in Appendix~\ref{app:exhibits}.

\section{Data and Validation}
\label{app:institutions}

\subsection{The product families and the execution share}
\label{app:families}

Four product families reference the treated stocks; they share the
arithmetic of equation~\eqref{eq:rebalance} and differ only in who executes,
in which instrument, and against which reference price. \emph{U.S.\ index
products}: physical basket ($\approx$ half of exposure), OTC total-return
swaps across $\approx$ ten counterparties, and index futures;
market-on-close orders by 3:50~p.m., executed at the 4:00 close.
\emph{U.S.\ single-stock products}: swap-based with exchange-traded and FLEX
options as overflow; the SK~Hynix ADR products reference the 4:00~p.m.\ ET
ADR close, so their hedges reach Seoul only at the next Korean open and
contribute nothing to $K$ at the Seoul close. \emph{Korean domestic
products}: in-house stock plus single-stock futures, no swap layer;
spot-type funds key to the 15:20--15:30 stock auction, futures-type and inverse
products to single-stock futures trading to 15:45. \emph{Hong Kong
cross-border products}: swap-based synthetics contractually keyed to the
Korean close, so their hedges execute into the same Seoul auction.

Sizing and execution are simultaneous by design (the order fills at the
price at which the NAV is struck, with the KRX disseminating the indicative
auction price in real time through the closing window); the sequential
components---the futures leg of spot-type funds (15:30 reference, hedging to
15:45), stressed-day overflow, next-day true-ups, and Hong~Kong post-close
option blocks---are the empirical counterpart of $\varphi < 1$.

The execution share is measured from portfolio deposit files for all
sixteen domestic products (July~24, 2026) combined with venue weights:
spot-type funds hold 44--51\% of exposure in spot (one issuer at 37\%),
futures-type funds 10--21\%, inverse funds none. The capital-weighted
domestic spot slice entering equation~\eqref{eq:saturation} is $0.44$ for
Samsung and $0.41$ for SK~Hynix ($0.49$ under the narrower published
convention). Worldwide, SK~Hynix's $K$ is 59\% swap-intermediated against
41\% in-house; Samsung's is 99\% domestic. By Lemma~\ref{lem:swap} the
implied at-the-close shares are $\varphi \approx 0.76$ (SK~Hynix) and
$\varphi \approx 0.44$ (Samsung), the cross-stock inversion used in
Appendix~\ref{app:phi}.

\subsection{The KRX intraday deliveries and their validation}
\label{app:krx}

The within-day evidence uses two native Korea Exchange deliveries: 1-minute
single-stock futures bars, all KOSPI roots, April~1--July~31, 2026 (83
trading days), and 1-minute cash-equity bars with quotes and order-flow
fields, all KOSPI issues, May~4--July~31, 2026 (61 days; 15{,}170{,}614 rows,
927 issues, 55{,}740 issue-days, 98.3\% price-forming). All pre-registered
validation gates pass. The features that affect the estimates:

\emph{Sparse grid.} The key (date, minute, issue) is unique; per-file
counts sum to the concatenated length. An absent minute is a missing
observation, never a zero return. Rows with a blank opening price (1.7\%)
are block or negotiated trades; every price mark uses price-forming rows
only.

\emph{The closing auction is a single batch.} The 15:20--15:29 window holds
eight rows in the entire cash delivery (one price-forming). On all
55{,}245 auction bars, high $=$ low and value/volume $=$ close exactly, so
the auction bar's traded value is $W^{auc}_t$: the saturation denominator
is observed, not imputed.

\emph{Volatility-interruption extensions.} A closing-auction interruption
extends the auction to 15:32; 401 issue-days on 52 of 61 days are affected,
including SK~Hynix June~24 and July~2 and Samsung June~29. The implemented
rule takes the first price-forming row in $[15{:}30,\,15{:}40)$ and carries
an extension flag; extended issue-days leave the venue comparison of
Appendix~\ref{app:impact} and are counted.

\emph{Interval fields, never cumulative.} Cumulative fields increment
through the after-hours auctions and include non-price-forming trades.
Reconciled against the daily exchange file, the auction close equals the
official close on 100\% of ticker-days and interval volume matches daily
volume within 0.5\% on 100\%.

\emph{Field semantics.} (i) The buy--sell imbalance field is a buy
\emph{share} in $[0,1]$ (signed value $=$ value $\times (2b-1)$; null $=$
one-sided interval); it is invalid at a single-price auction, where there is no
aggressor---86\% of auction-window impact slopes built from it are
negative---so the venue comparison is run gross on both sides. (ii) The
disseminated midpoint field is an interval average; every window return
uses end-of-interval quote midpoints. (iii) The close-bar spread fields are
degenerate by construction and are not used. (iv) The 09:00 bar mixes the
opening auction with the first continuous minute; the open is used
descriptively and the next-open markout mark is the 09:01 midpoint.

\emph{Futures gates.} Value $=$ price $\times$ volume $\times$ 10 on
99.10\% of single-price bars (residual: two roots with adjusted contract
size, constant, returns unaffected). The 15:45 auction is clean: OHLC coincide
on all 9{,}168 populated auction bars. The 15:35--15:44 window holds 0.9\% of
the 15:30--15:34 row count but 49.2\% of its volume (blocks, no price
formation) and is excluded, so 15:34 is the last continuous mark. The
primary contract is the nearest settlement month with a 15:45 close, ties
by traded value; returns are within contract, and every table is reported
with and without roll days. Roots map to underlyings numerically by price
level: eight of twelve cash names match inside a 0.10\% median basis
(both treated names); the four that do not are rejected.

\emph{Coverage.} The treated pair marks 100\% on all four futures marks on
83 of 83 days. Tier~A futures: 24 roots (treated pair $+$ 22 controls);
185 of 209 roots are too thin to mark at the close. Tier~A cash: 91
controls, selected on auction-bar coverage, staleness and top-decile median
traded value, with the ten named daily controls forced in. Not contained:
April on the cash side, August onward on both, the funds' own tape, KOSDAQ,
and investor-type attribution.

\subsection{Rebalancing capital and saturation}
\label{app:kappa}

Worldwide rebalancing capital per underlying, equation~\eqref{eq:K}, is
assembled from exchange filings (sixteen Korean products), issuer and
exchange files (Hong Kong), issuer files (U.S.\ ADR), and index
pass-through at KOSPI200 weights, FX-converted at daily rates; leverage
coefficients $L^2-L = 2$ ($+2\times$) and $6$ ($-2\times$). Index products
enter the capital path of Figure~\ref{fig:timeline} but are excluded from
the auction-relevant $K$ of Table~\ref{tab:complex}, since they rebalance
against the 15:45 index-futures close.

Same-day creations and redemptions are measured from the unit count implied
by net assets, not from the listed-unit series, which lags net assets by one
trading day. On the corrected series the primary market shows no same-day
response to the return (Table~\ref{tab:flowchasing}), so $K$ is reported
gross.

Restricted to the domestic spot slice and multiplied by the previous-close-to-15:20
return, lagged $K$ gives the predetermined demand $\hat d^{\,pre}_t$;
against the value the auction clears,
\begin{equation}
S_t \;=\; \frac{\left|\hat d^{\,pre}_t\right|}{W^{auc}_t}.
\label{eq:saturation}
\end{equation}
For SK~Hynix the post-launch median is $1.06$, the maximum $7.50$, above one
on 26 of 46 days (Figure~\ref{fig:saturation}).

Two reconciliations. Applying the published domestic convention (domestic
products, spot slice) to our data returns 1.58\% (Samsung) and 2.04\%
(SK~Hynix) of daily traded value against published 1.6\% and 2.1\%. The
compounding identity $K_+ = K(1+Lr)$ is recovered from the Korean panel
with slope $1.100$ (s.e.\ $0.033$, $R^2 = 0.65$); the offshore series do
not support the regression at any alignment (FX and timing conventions), so
offshore capital enters only in multi-week levels. Table~\ref{tab:complex}
assembles the complex.

\subsection{Panel, comparison sets, and inference conventions}
\label{app:inference}

\emph{Panel.} One stock-by-day panel of the three daily windows plus leads
to $h=3$: 15{,}300 stock-days, January~2, 2024--July~31, 2026, all product
regimes and the Korean controls. U.S.\ prices are split-adjusted; the
Korean offshore-window split uses each stock's actual offshore listing
date. Comparison groups: the \emph{large-cap} group (ten largest KOSPI
names never meeting the single-stock eligibility screen) and the
\emph{technology-peer} group (Samsung Electro-Mechanics, LG Electronics,
LG Display, LG Innotek, Hanmi Semiconductor, DB HiTek, LX Semicon, Samsung
SDS, NAVER, Kakao); SK~Square is excluded everywhere as
treated-contaminated.

\emph{Lead controls.} The shock series are mildly mean-reverting, so
intervening shocks are purged; the choice moves magnitude, not sign
(treated-post overnight leg, semiconductor instrument: $-0.533$
uncontrolled, $-0.386$ per-horizon baseline, $-0.302$ with three leads).
The headline triple difference uses the single lead control of the frozen
specification.

\emph{Inference.} OLS with date-clustered standard errors cell by cell.
Because clustered errors are unreliable with two treated units, every
differenced estimate is also subjected to exact randomization inference
over the treated pair (66 pairs daily, 4{,}278 intraday) and over 47
placebo launch dates in the pre-period. Nonlinear statistics use
date-block bootstraps (400--800 draws, whole dates resampled, every
statistic recomputed per draw). Randomization inference is the arbiter: a
result that fails permutation is reported as a null whatever its
$t$-statistic.

\emph{Coverage gates.} Windows with signed coverage below 0.9, a missing
boundary mark, or fewer than half their minutes traded are dropped; an
absent minute is never a zero return; winsorization is within issue;
volatility and volume scalers are lagged; VI-extended issue-days leave the
venue comparison and are counted (Appendix~\ref{app:krx}). Commercial
vendor intraday files truncate at 15:00 and are not used.

\subsection{The shared-mark rule}
\label{app:sharedmark}

Any outcome signed by $s_t = \mathrm{sign}(\hat d^{\,pre}_t)$ must lie
strictly after the window used to form the sign and share no price mark
with it: if sign and outcome share a mark $P$, $\mathrm{sign}(\hat r)\cdot
y$ inherits a one-signed covariance in the measurement error of $P$ that
scales with volatility, survives differencing, date fixed effects, and
leave-one-out diagnostics, and is removed only by construction. All five
registered futures outcomes begin at 15:20, exactly where the predetermined
return ends; the cash outcomes are built the same way.

\section{Proofs}
\label{app:proofs}

This appendix proves every lemma, proposition, corollary and remark of
Section~\ref{sec:model}, in order of appearance, together with
Proposition~\ref{prop:ledger} (with the welfare proofs) and
Proposition~\ref{prop:curve} (with Lemma~\ref{lem:curve}, in
Appendix~\ref{app:curveproof}). The strategic sector is a continuum, so
every first-order condition is an individual optimum at fixed prices and
aggregation is a single integral; closed forms are exact. All displays are
verified symbolically in the replication package.

\subsection{Notation and the equilibrium system}
\label{app:notation}

Throughout, $S_T \equiv T\sigma_u^2 + \sigma_0^2$, $\Lambda = \gamma
S_T/\mu$, $\Lambda_c = \rho\Lambda$, $\ell = \Lambda_c K$,
$m = 1/(1-\ell)$, $\psi_c = \tau\Lambda_c$, $\psi = \tau_s\Lambda$, and
$\mathcal{S}$ is the measured session share of the public-news overshoot. Write
\begin{equation}
\bar\rho \;\equiv\; \frac{1+\psi\rho}{1+\psi} \;=\; 1 - \mathcal{S}(1-\rho) ,
\qquad
\DN \;\equiv\; (1+\psi)\big(\bar\rho + \psi_c - \ell\big) ,
\qquad
M \;\equiv\; \frac{1}{1-\ell+\psi_c} ,
\label{eq:aDN}
\end{equation}
so that $\DN = (1-\ell+\psi_c) + \psi(\rho + \psi_c - \ell)$; the session
and close blocks factor only at $\rho = 1$, where $\bar\rho = 1$ and $\DN =
(1+\psi)M^{-1}$. The \emph{stable region} is $\DN > 0$, equivalently
$\ell < \ell_{tip} \equiv \bar\rho + \psi_c$; all statements below are on
that region, and $\ell_{tip}$ always denotes the expression $\bar\rho +
\psi_c$, never a solved root. Two further abbreviations are used below and
are collected here: the \emph{arbitrage share of effective depth} and the
\emph{effective dose},
\begin{equation}
\chi \;\equiv\; \frac{\psi_c}{\rho + \psi_c} ,
\qquad
\ell_{\text{eff}} \;\equiv\; \frac{\ell}{\bar\rho + \psi_c} ,
\qquad\text{so}\qquad
\Pi_g \;=\; \frac{1}{1-\ell_{\text{eff}}},
\quad
\Pi_u \;=\; \chi\,\Pi_g ,
\quad
a_z \;=\; (1-\chi)\Lambda\Pi_g .
\label{eq:achi}
\end{equation}
$\chi$ is the arbitrage sector's share of the depth the \emph{auction}
draws on, and is built on $\rho$ rather than $\bar\rho$ because the session
position reaches the close through two legs whose composition it measures;
see Remark~\ref{rem:invariance}.
Risk aversion $\gamma_A > 0$ is maintained throughout: at $\gamma_A = 0$
both capacities diverge, $\chi \to 1$, $\ell_{\text{eff}} \to 0$, and every
loading collapses to its frictionless value. A second maintained assumption
is $\sigma_c > 0$: the closing imbalance is the only shock unresolved when
the session position is set, so at $\sigma_c = 0$ the capital gain to the
close is riskless, Lemma~\ref{lem:carry} applies with $\Var(p') = 0$, and
the session carry has no interior optimum. The model's $\sigma_c \to 0$
limit is nonetheless well defined in prices: $\psi \to \infty$,
$\bar\rho \to \rho$, and $\E_s[r_c] - r_s \to 0$, exactly as
Lemma~\ref{lem:carry} requires.

\emph{Units.} In the equilibrium loadings of
equation~\eqref{eq:aloadings} the flows $z$ and $z_c$ are measured in
\emph{shares}, so $a_z$ and $a_{z_c}$ carry a depth coefficient. In every
second-moment display --- Proposition~\ref{prop:vol}, its threshold
$\ell^{\ast}$, and Proposition~\ref{prop:loss} --- the four shocks are
measured in \emph{return} units, that is, $\sigma_z^2$ and $\sigma_c^2$
denote the variances of the price displacements $\Lambda z$ and
$\Lambda_c z_c$ rather than of the share counts. Under that convention the
two noise loadings appearing in those displays are
\begin{equation}
a_z \;=\; \frac{\bar\rho}{\bar\rho+\psi_c-\ell} ,
\qquad
a_{z_c} \;=\; \frac{1}{1+\psi_c-\ell} ,
\label{eq:aunits}
\end{equation}
the same objects as equation~\eqref{eq:aloadings} divided by $\Lambda$
and $\Lambda_c$ respectively. The two conventions must not be mixed inside
one display.

Private news carries a parallel convention. In $S_T = T\sigma_u^2 +
\sigma_0^2$ and in the accumulation profile of
equation~\eqref{eq:profile}, $\sigma_u^2$ is the arrival \emph{rate} of
private-news variance, so that the session delivers $T\sigma_u^2$ in total
and $(T-t)\sigma_u^2$ remains unresolved at $t$; private news arrives at
constant intensity over $[0,T)$ and the session noise $z$ accumulates
uniformly, which is what makes the unresolved fraction $1-t/T$. In
Propositions~\ref{prop:vol} and~\ref{prop:loss}, where no time index
appears, $\sigma_u^2$ is instead the total session variance, i.e.\
$T\sigma_u^2$ of the first convention. Only the terminal value
$V_T = a_{z_c}^2\sigma_c^2$ of equation~\eqref{eq:profile} enters any closed
form below, so no result depends on the reconciliation, but the two
$\sigma_u^2$ should not be equated across the two groups of displays.

\subsection{Lemmas}
\label{app:lemmas}

\begin{proof}[Proof of Lemma~\ref{lem:swap}]
(i) A total-return swap pays the counterparty $-N_{\text{out}}$ times the
return; between resets $N_{\text{out}}$ is fixed, so the counterparty's
delta is a constant share count and a flat hedge requires no intraday
trade.

(ii) At the reset the notional moves from $LA$ to $LA(1+Lr)$. The change in
delta, in currency at the reference price, is
\[
LA(1+Lr) - LA(1+r) \;=\; A(L^2-L)r ,
\]
identical to equation~\eqref{eq:rebalance}, keyed to the same reference
return and required at the same instant.

(iii) A hedge in instrument $j$ leaves residual variance
$\Var(r_j - r_c) > 0$ against a liability struck at $r_c$, zero only for
the underlying priced at the reference price; a risk-averse counterparty
therefore selects the underlying executed at the close.
\end{proof}

\begin{proof}[Proof of Lemma~\ref{lem:carry}]
Let the trader hold $x$ from the first venue to the second. Atomistic, he
takes both prices as given; his payoff is $x(p'-p)$ and his objective
$x(\E[p'] - p) - \tfrac{\gamma_A}{2}\Var(p')\,x^2$. If $p'-p$ is known at
the time of trading the objective is linear with certain slope and has no
interior maximum unless the slope vanishes; aggregate demand at the first
venue is $\pm\infty$ whenever $p' \neq p$, so the two venues clear
together. With $\Var(p') > 0$ the objective is strictly concave and
$\E[p'] - p = \gamma_A\Var(p')\,x$ gives the stated demand.
\end{proof}

\begin{lemma}[filterability]
\label{lem:filter}
The automaton's demand $q_L = K r_c$ is conditionally deterministic given
the closing price and public arithmetic. A Bayesian dealer assigns it zero
information content but must still hold it: the inventory concession on the
rebalancing flow, and with it $\ell$ and the loss mechanism, survive
Bayesian updating unchanged. Under risk-neutral Bayesian dealers the
concession is the information content alone, so perfectly anticipated
uninformed flow trades at zero impact.
\end{lemma}

\begin{proof}
$q_L = K r_c$ is measurable with respect to the closing price and public
data, so for any information set containing $r_c$,
$\E[v \mid \text{flow},\, r_c] = \E[v \mid \text{flow} - q_L,\, r_c]$: the
dealer's posterior is unchanged. The dealer's optimal inventory
$h = (\E v - r)/(\gamma S)$ is unchanged in its information argument; its
risk argument is not, since $q_L$ still enters the warehoused inventory and
the certainty equivalent through $\tfrac{\gamma}{2}h^2 S_T$. The concession
on the rebalancing flow is therefore invariant to updating; $z_c$ is
uninformative by construction and enters only as noise.

The risk-neutral case requires one further step, because the dealer must
net $q_L$ out of a flow it observes only in total. Let the dealer set
$r_c = \E[v \mid \omega]$ from the aggregate $\omega = q_L + X_c + z_c$,
with $X_c = \beta(v - r_c)$ the arbitrage schedule. Since $q_L = K r_c$ is
a deterministic function of \emph{the very price the dealer is setting},
the residual $\omega - K r_c = \beta(v-r_c) + z_c$ is computable by the
dealer at the moment it quotes, and pricing on it is the fixed point
$r_c = g + \lambda\big[\beta(v - r_c) + z_c\big]$ for the Bayesian
$\lambda > 0$ that $z_c$ leaves strictly positive. Solving,
$r_c = g + \lambda(\beta u + z_c)/(1 + \lambda\beta)$, so
$\partial r_c/\partial g = 1$: the public gap neither overshoots nor
reverts. The mechanical flow is priced at zero not because $\lambda = 0$ ---
adverse selection from $X_c$ keeps $\lambda > 0$ --- but because a flow
measurable with respect to the quoted price is separable from it. The
contrast with equation~\eqref{eq:fixedpoint} is exactly the risk channel:
the risk-averse dealer also knows $q_L$ carries no information, yet must
still warehouse it, and it is the inventory concession on that warehoused
flow, not any informational component, that produces $\ell$.
\end{proof}

Lemma~\ref{lem:curve} is stated and proved in Appendix~\ref{app:curveproof}.

\subsection{Equilibrium strategies}
\label{app:strategies}

\begin{proof}[Proof of Proposition~\ref{prop:loadings}: the auction stage]
At the auction, speculator $i$ knows $v = g+u$ and submits a
price-contingent schedule. His terminal position $x^i$ is marked at
fundamental value, which the deep capital of Section~\ref{sec:chaining}
restores, and bears one cycle of overnight news, $\sigma_A^2 = \sigma_0^2$;
his objective is
$x^i(v - r_c) - \tfrac{\gamma_A}{2}\sigma_A^2 (x^i)^2$ at given $r_c$. The
first-order condition is equation~\eqref{eq:auctionfoc},
$v - r_c = \gamma_A\sigma_A^2 x^i$, whence $x^i = \tau(v-r_c)$ with
$\tau = n/(\gamma_A\sigma_A^2)$ after aggregation. The session position
does not appear, so, writing $\bar X \equiv X_o + X_s$ for the position
carried into the auction, the net purchase there is
$X_c = \tau(v-r_c) - \bar X$.

Substituting into equation~\eqref{eq:fixedpoint} with $\Lambda_c\tau =
\psi_c$ and $\Lambda_c K = \ell$,
\[
r_c \;=\; \tilde r + \ell\,r_c + \psi_c(v - r_c),
\qquad
\tilde r \;=\; g + \Lambda(\bar X+z) - \Lambda_c \bar X + \Lambda_c z_c ,
\]
so $r_c(1 - \ell + \psi_c) = \tilde r + \psi_c v$, i.e.
\[
r_c \;=\; \frac{\psi_c\,v + \tilde r}{\psi_c + (1-\ell)} ,
\]
which is the auction stage taken alone, with $\bar X$ still unsolved inside
$\tilde r$; combined with the session stage below it yields
equation~\eqref{eq:Pi0}. Its denominator vanishes at $\ell = 1+\psi_c$,
which is the singularity of \emph{this stage holding $\bar X$ fixed} and not
the singularity of the model: the equilibrium pole is
$\ell = \bar\rho+\psi_c$, obtained only after the session position is
solved out, and the two coincide only at $\rho = 1$.
\end{proof}

\begin{proof}[Proof of Proposition~\ref{prop:sessionpos}]
The session trader chooses $\bar x$ knowing $(g,u,z)$ and facing one
unresolved shock, the closing imbalance $z_c$; he then re-optimizes at the
auction and carries the resulting terminal position overnight. Writing total
profit as $\bar x(v - r_s) + x_c(v - r_c) - \tfrac{\gamma_A}{2}\sigma_A^2 X^2$
with $X = \bar x+x_c$ and substituting the optimal terminal holding
$X = (v-r_c)/(\gamma_A\sigma_A^2)$ of equation~\eqref{eq:auctionfoc}, the
certainty equivalent at the auction is
\begin{equation}
\mathrm{CE}_c \;=\;
\underbrace{\bar x\,(r_c - r_s)}_{\text{session leg}}
\;+\;
\underbrace{\frac{(v-r_c)^2}{2\gamma_A\sigma_A^2}}_{\text{auction leg, }\bar x\text{-free}} .
\label{eq:asessfoc}
\end{equation}
The auction leg is free of $\bar x$ but \emph{not} of $r_c$, and $r_c$ is
exactly the random variable the session leg is exposed to. Under CARA the
certainty equivalent of the sum is not the sum of certainty equivalents, so
$\bar x$ cannot be read off the session leg alone: the continuation value has
$\partial\,\mathrm{CE}_c/\partial r_c = -(v-r_c)/(\gamma_A\sigma_A^2) = -X$,
i.e.\ through the auction leg the trader is already short $X$ units of
closing-price risk, and he hedges that exposure in the session.

Evaluating the exact optimum. Write $r_c = \E_s[r_c] + a_{z_c}z_c$ and
$V \equiv \Var_s(r_c) = a_{z_c}^2\sigma_c^2$, the variance left unresolved
when the position is set. Maximizing
$\E_s\big[-\exp(-\gamma_A \mathrm{CE}_c)\big]$ is a Gaussian integral in
$z_c$ with a quadratic exponent; carrying it out and differentiating in
$\bar x$ gives the first-order condition
$(\E_s[r_c] - r_s)\big(1 + V/\sigma_A^2\big) = \gamma_A V \bar x - V(v -
\E_s[r_c])/\sigma_A^2$, i.e.
\begin{equation}
\bar x^i \;=\;
\underbrace{\frac{\E_s[r_c] - r_s}{\gamma_A V}}_{\text{carry leg}}
\;+\;
\underbrace{\frac{v - r_s}{\gamma_A \sigma_A^2}}_{\text{value leg}} .
\label{eq:asessdemand}
\end{equation}
The first term is the risk-priced carry of Lemma~\ref{lem:carry}. The
second is a pure value position and is the hedge just described: the trader
means to hold $\tau(v-r_c)$ overnight in any case, and whenever the session
price is already away from fundamental value he acquires part of that
holding early, at the same overnight risk price $\gamma_A\sigma_A^2$ that
governs the auction. Lemma~\ref{lem:carry} alone delivers only the first
term because it is stated for a trader who \emph{liquidates} at the second
venue; this trader does not.

Aggregating over the continuum with $\tau_s = n/(\gamma_A V)$ and
$\tau = n/(\gamma_A\sigma_A^2)$, and using the session price of
equation~\eqref{eq:session}, $r_s = g + \Lambda(\bar X+z)$,
\begin{equation}
\bar X \;=\; \tau_s\big(\E_s[r_c] - r_s\big)
\;+\; \tau\big(v - r_s\big) .
\label{eq:asessagg}
\end{equation}
The same $\tau$ governs both the value leg and the auction, so the session
and the close are priced off one risk aversion and one overnight variance;
$\tau_s$ enters only through the carry.

\begin{proposition}[the intraday accumulation profile]
\label{prop:profile}
Applying equation~\eqref{eq:asessdemand} at $t \in [0,T)$, the sector's
session position and the unresolved closing-price variance are
\begin{equation}
X_t \;=\; \underbrace{\frac{\E_t[r_c] - r_s(t)}{\gamma_A V_t}}_{\text{carry leg}}
\;+\; \underbrace{\frac{\E_t[v] - r_s(t)}{\gamma_A \sigma_0^2}}_{\text{value leg}},
\qquad
V_t \;=\; \Pi_u^2 \sigma_u^2 (T-t)
\;+\; a_z^2 \sigma_z^2 \Big(1 - \frac{t}{T}\Big)
\;+\; a_{z_c}^2 \sigma_c^2 .
\label{eq:profile}
\end{equation}
$V_t$ declines monotonically to $V_T = a_{z_c}^2 \sigma_c^2$, so the carry
leg's target rises through the session while the value leg's is priced at
a constant $\sigma_0^2$ and moves only with the news, and the profile
depends on primitives
only through $\sigma_u^2 T / (a_{z_c}^2 \sigma_c^2)$ and
$\sigma_0^2/(a_{z_c}^2\sigma_c^2)$.
\end{proposition}

\begin{proof}
Run the argument at interior $t$: equation~\eqref{eq:asessdemand} gives the
two legs of $X_t$, with $V_t = \Var_t(r_c)$ in the carry leg and the
overnight variance $\sigma_0^2$ in the value leg, the latter unaffected by
$t$ because the terminal position is carried to the same horizon regardless
of when it is opened. By Proposition~\ref{prop:loadings}, $r_c$ is linear in
$(g,u,z,z_c)$; at $t$ the unresolved components are private news of
variance $\sigma_u^2(T-t)$ and loading $\Pi_u$, session noise of variance
$\sigma_z^2(1-t/T)$ and loading $a_z$, and the closing imbalance of
variance $\sigma_c^2$ and loading $a_{z_c}$. Independence gives
equation~\eqref{eq:profile}; both $t$-dependent terms are strictly
decreasing, so $V_t \downarrow V_T = a_{z_c}^2\sigma_c^2 > 0$ and the carry
leg's target rises. Scaling by $V_T$ leaves a function of $t/T$,
$\sigma_u^2 T/(a_{z_c}^2\sigma_c^2)$ and $\sigma_0^2/(a_{z_c}^2\sigma_c^2)$
alone, the last entering only through the value leg.
\end{proof}

Substituting $\E_s[r_c] = M[r_s - \Lambda_c \bar X + \psi_c v]$ with
$z_c$ at its mean into equation~\eqref{eq:asessagg} and solving the linear
equation in $\bar X$, using $\tau_s\Lambda = \psi$ and
$\tau\Lambda = \psi_c/\rho$,
\begin{equation}
\bar X^\ast \;=\; \frac{1}{\Lambda\,D}
\Big[\,
\underbrace{\mathcal{S}\,\ell\, g}_{\text{public gap}}
\;+\;
\underbrace{\Psi\, u}_{\text{private news}}
\;+\;
\underbrace{(\mathcal{S}\ell - \Psi)\,\Lambda z}_{\text{noise}}
\,\Big] ,
\qquad
\Psi \;\equiv\; \chi\big[(\bar\rho+\psi_c) - (1-\mathcal{S})\,\ell\big] ,
\label{eq:asessload}
\end{equation}
where $D = \bar\rho+\psi_c-\ell$ and $\chi$, $\mathcal{S}$ are as in
equations~\eqref{eq:achi} and~\eqref{eq:asmap}.
The public-news term is taken at the opening auction and the other two during the
session, since $u$ and $z$ are unrealized when the auction prints: the first
underbrace is $X_o$ of equation~\eqref{eq:sessload} and the remaining two are
its $X_s$, so that $\bar X^\ast = X_o + X_s$.
Clearing against equation~\eqref{eq:impact}, $r_s = g + \Lambda(\bar X^\ast + z)$,
gives equation~\eqref{eq:sessprice}.

The structure is the paper's: the noise loading is the public-gap loading
net of the private-news loading, so it is the difference of a
pre-positioning motive and a liquidity-provision motive and changes sign
where they cross. On the stable region $D > 0$: the gap loading
$\mathcal{S}\ell/(\Lambda D) > 0$ is strictly proportional to $K$ and vanishes as
$K \to 0$ --- the speculator holds the public gap only because a
mechanical buyer is known to arrive; the private-news loading $\Psi$ is
strictly positive below the tipping threshold, since
$(\bar\rho+\psi_c) - (1-\mathcal{S})\ell > (\bar\rho+\psi_c) - \ell = D > 0$; and the
noise loading carries the sign of $\mathcal{S}\ell - \Psi$ and vanishes at
\begin{equation}
\ell \;=\; \frac{\psi_c\,(\bar\rho+\psi_c)}{\psi_c + \rho\,\mathcal{S}} ,
\label{eq:anoisethreshold}
\end{equation}
which replaces the equal-depth value $\ell = \psi_c$. Two limits check it.
As $\mathcal{S} \to 0$ the carry leg vanishes, the session position is pure value
trading, and the noise-chasing threshold rises to the tipping one $\bar\rho+\psi_c$: a value
trader fades noise at every dose and never chases it. At $\rho = 1$ and
$\psi_c$ small the threshold returns to $\psi_c$.
\end{proof}

\subsection{The assembled cycle}
\label{app:cycle}

\begin{proof}[Proof of Proposition~\ref{prop:loadings}: the equilibrium loadings]
Substituting equation~\eqref{eq:asessload} into the clearing equation of
the auction stage and collecting terms gives
$r_c = \Pi_g g + \Pi_u u + a_z z + a_{z_c} z_c$ with
\begin{equation}
\Pi_g = \frac{\bar\rho+\psi_c}{\bar\rho+\psi_c-\ell},
\;
\Pi_u = \chi\,\Pi_g,
\;
a_z = (1-\chi)\,\Lambda\,\Pi_g,
\;
a_{z_c} = \frac{\Lambda_c}{1+\psi_c-\ell} ,
\qquad
\chi \;=\; \frac{\psi_c}{\rho+\psi_c} .
\label{eq:aloadings}
\end{equation}
The substitution is the following two steps. $\bar X^\ast$ enters the
clearing equation with weight $+\Lambda$ through $r_s$ and
$-\Lambda_c$ through the unwind, so its net contribution is
$\Lambda(1-\rho)\bar X^\ast$, and the auction stage reads
\[
r_c\,(1-\ell+\psi_c)
\;=\; (1+\psi_c)\,g \;+\; \psi_c\,u \;+\; \Lambda z \;+\; \Lambda_c z_c
\;+\; (1-\rho)\,\Lambda \bar X^\ast .
\]
Substituting equation~\eqref{eq:asessload} for $\bar X^\ast$ and collecting,
each channel numerator carries the common factor $(1-\ell+\psi_c)$ against
the denominator $D = \bar\rho+\psi_c-\ell$, leaving the displays. Two
identities do the collapsing, both immediate from $\bar\rho = 1-\mathcal{S}(1-\rho)$
and $\chi = \psi_c/(\rho+\psi_c)$:
\[
(1-\chi)(1-\mathcal{S})\,\psi \;=\; \mathcal{S} ,
\qquad
\big(1+\rho(\tfrac{\psi_c}{\rho}+\psi)\big)
\;=\; \frac{\bar\rho+\psi_c}{1-\mathcal{S}} ,
\]
the first of which is what makes the gap loading of
equation~\eqref{eq:asessload} equal to $\mathcal{S}\ell$ exactly. The $z_c$ channel
never passes through $\bar X^\ast$ --- the closing imbalance is unforecastable
at the session, so the session position does not respond to it --- and its
numerator carries no session term at all, giving
$a_{z_c} = \Lambda_c/(1-\ell+\psi_c)$ with the auction-stage denominator
rather than $D$. The differing denominator is therefore a feature of the
timing, not an inconsistency.

Note that $\Pi_u/\Pi_g = \chi$ and $a_z/(\Lambda\Pi_g) = 1-\chi$ still
sum to one identically; what the two-leg demand changes is the value of
$\chi$, from $\psi_c/(\bar\rho+\psi_c)$ under a carry-only session to
$\psi_c/(\rho+\psi_c)$ here. The public channel does not see the
difference, and the next remark says why.

\begin{remark}[what the session demand does and does not move]
\label{rem:invariance}
Let $\mathcal{S}$ denote the session share of the public-news overshoot, the object
equation~\eqref{eq:asmap} measures. Conditional on $\mathcal{S}$, the public-news
loading $\Pi_g$, the benchmark loading $\Pi_g^{\,0}$, the closing-imbalance
loading $a_{z_c}$, and hence the predatory share $\mathcal{P}$ of
Proposition~\ref{prop:predshare}, are \emph{identical} whether the session
position is the carry of Lemma~\ref{lem:carry} alone or the two-leg demand
of equation~\eqref{eq:asessdemand}. The private-news and session-noise
loadings $\Pi_u$ and $a_z$, the revelation and noise-chasing thresholds, and the
inversion of $\psi$ from $\mathcal{S}$ are not.
\end{remark}

\begin{proof}
Under equation~\eqref{eq:asessagg} the effective depth entering $\Pi_g$ is
$\Lambda_{\text{eff}} = (1+a)(1+\rho(a+b))/(1+a+b)$ with $a = \psi_c/\rho$
and $b = \psi$, and $\Pi_g = \Lambda_{\text{eff}}/(\Lambda_{\text{eff}} -
\ell)$. Since $\mathcal{S} = b/(1+a+b)$, one has $1-\mathcal{S} = (1+a)/(1+a+b)$, and
substituting gives $\Lambda_{\text{eff}} = (1-\mathcal{S})(1+\psi_c+\rho\psi) =
\bar\rho + \psi_c$ with $\bar\rho = 1-\mathcal{S}(1-\rho)$ --- the same expression the
carry-only model produces at its own $\mathcal{S}$. Because every public-channel
object is a function of $(\ell,\psi_c,\rho,\mathcal{S})$ through
$\bar\rho + \psi_c$ alone, and $a_{z_c}$ contains no session term, all four
are invariant. The private channel is not, because $\chi$ depends on the
composition of the session position and not only on its size.
\end{proof}

The economic content is that $\mathcal{S}$ is a sufficient statistic for how much of
the closing overshoot the session leg is responsible for, whatever motive
produced it. The measurement of Section~\ref{sec:quant} reads $\mathcal{S}$ off the
window split and never uses $\psi$ structurally, so the calibration and
every counterfactual that runs through $\Pi_g$ are unaffected by the
distinction; what the distinction changes is the structural capacity one
would report behind a given $\mathcal{S}$, and the private-news predictions.

$a_z \ne a_{z_c}$ unless $\rho = 1$: session noise reaches the close
through the carried book at $\Lambda$ scaled by $\bar\rho$, while
closing-imbalance noise arrives at the auction at $\Lambda_c$; the two
variances must not be combined except at $\rho = 1$.

Dividing numerator and denominator of $\Pi_g$ by $\bar\rho+\psi_c$
gives $\Pi_g = 1/(1-\ell_{\text{eff}})$, and
\[
\Pi_g - 1 \;=\; \frac{\ell}{\bar\rho+\psi_c-\ell}
\;=\; \frac{(1+\psi)\,\ell}{\DN} \;>\; 0
\]
for every $\ell > 0$ on the stable region, with no threshold and no
restriction on either capacity. This is equation~\eqref{eq:Pi0}, and
equation~\eqref{eq:terminal} follows from
$X = (\psi_c/\Lambda_c)(v - r_c)$ on substituting the loadings.

Three properties used in the body's prose. $\Pi_u/\Pi_g = \chi$ and
$a_z/(\Lambda\Pi_g) = \bar\rho/(\bar\rho+\psi_c) = 1-\chi$, which sum to
one identically; the model is linear-quadratic with Gaussian shocks and $g$
is common knowledge, so the equilibrium map is linear and additive across
shocks, and $\Pi_u$, $a_z$ and $a_{z_c}$ coincide with the loadings of the
$g \equiv 0$ economy. At $\rho = 1$, $\bar\rho \equiv 1$ for every
$\psi$, so all four loadings are $\psi$-free; in general $\psi$
enters only through $\bar\rho$, with
$\partial\bar\rho/\partial\psi = (\rho-1)/(1+\psi)^2$ and limits
$\bar\rho \to 1$ ($\psi \to 0$) and $\bar\rho \to \rho$
($\psi \to \infty$); since $\Pi_g$ is decreasing in $\bar\rho$, session
capacity raises the overshoot when $\rho < 1$ and lowers it when
$\rho > 1$.
\end{proof}

\begin{proof}[Proof of Corollary~\ref{cor:overshoot}]
Immediate from equation~\eqref{eq:Pi0}: the numerator $\ell$ and the
denominator $\bar\rho + \psi_c - \ell$ are strictly positive on the stable
region; no threshold in $\ell$ and no restriction on $\psi_c$ or $\psi$.
As $\psi_c \to \infty$, $\Pi_g \downarrow 1$ without reaching it at finite
capacity.
\end{proof}

\begin{corollary}[when the arbitrage sector damps the overshoot]
\label{cor:pimbound}
For $\ell < 1$, $\Pi_g \le 1/(1-\ell)$ if and only if
$\psi_c \ge 1-\bar\rho$, with equality at $\psi_c = 1-\bar\rho$ for every
$\ell$. At $\bar\rho = 1$ the condition is $\psi_c \ge 0$ and the bound
always holds.
\end{corollary}

\begin{proof}
For $\ell < 1$ both $1-\ell$ and $\bar\rho+\psi_c-\ell$ are positive, so
$(\bar\rho+\psi_c)/(\bar\rho+\psi_c-\ell) \le 1/(1-\ell)$ cross-multiplies
sign-preservingly to $(\bar\rho+\psi_c)(1-\ell) \le \bar\rho+\psi_c-\ell$,
i.e.\ $\bar\rho + \psi_c \ge 1$, i.e.\ $\psi_c \ge 1-\bar\rho$, with
equality at $\psi_c = 1-\bar\rho$ for every $\ell$. The restriction
$\ell < 1$ is needed and is not implied by interiority: when
$\bar\rho+\psi_c > 1$ the stable region contains $\ell > 1$, where
$1/(1-\ell) < 0 < \Pi_g$ and the comparison is vacuous. At
$\bar\rho + \psi_c \le 1$ interiority already forces $\ell < 1$.
\end{proof}

\subsection{The closing loading on the concave curve}
\label{app:concaveclose}

Proposition~\ref{prop:loadings} is derived under linear absorption, and the
quantification of Section~\ref{sec:quant} prices each day on the concave
curve of Lemma~\ref{lem:curve}, whose measured curvature is not small. This
section states the equilibrium loading on that curve, its exact relation to
equation~\eqref{eq:Pi0}, and which of the two loadings a reduced-form
$\hat\Pi_g$ estimates.

The clearing identity is unchanged. The automaton's order is sized by the
realised close, so at a concession $p$ over the frictionless mark it demands
$\ell(g+p)$, and the absorbing side supplies $Q(p)$. Linear absorption
$Q(p) = (\bar\rho+\psi_c)\,p$ returns equation~\eqref{eq:Pi0}. In the
$\varphi_c = 2$ normalisation of Appendix~\ref{app:curveproof} the same side
supplies $Q(p) = p\,(\bar\rho + \psi_c + b\,|p|)$, written odd in $p$ so that
a down move prices on the same curve as an up move, and clearing
$\ell(g+p) = Q(p)$ becomes
\begin{equation}
b\,|p|\,p \;+\; A\,p \;=\; \ell\,g ,
\qquad
A \;\equiv\; \bar\rho + \psi_c - \ell ,
\label{eq:aquadclose}
\end{equation}
which is equation~\eqref{eq:quadclose} in the continuum parameterisation.

\begin{proposition}[the concave close]
\label{prop:aconcave}
Let $b>0$, $\ell>0$ and $g \neq 0$, and define the effective depth
$A_{\mathrm{e}} \equiv \sqrt{A^2 + 4b\ell|g|}$. Then
equation~\eqref{eq:aquadclose} has a unique solution, odd in $g$, and the
close carries
\begin{equation}
\Pi_g(g) \;=\; 1 + \frac{2\ell}{A + A_{\mathrm{e}}} ,
\qquad
\Pi_g^{\mathrm{m}}(g) \;\equiv\; \frac{\partial r_c}{\partial g}
\;=\; 1 + \frac{\ell}{A_{\mathrm{e}}} ,
\label{eq:aconcaveload}
\end{equation}
the secant and the marginal loading. Both are equation~\eqref{eq:Pi0} with
the linear depth $A$ replaced --- by the mean of $A$ and $A_{\mathrm{e}}$ in
the first case, by $A_{\mathrm{e}}$ in the second. Both exceed one, both are
strictly decreasing in $|g|$ and in $\psi_c$, and
$\Pi_g > \Pi_g^{\mathrm{m}}$. Moreover
\begin{enumerate}
\item[(i)] as $b \to 0$, $A_{\mathrm{e}} \to |A|$, and for $A>0$ both
loadings collapse to equation~\eqref{eq:Pi0};
\item[(ii)] as $|g| \to 0$ at $A>0$ they collapse to the same limit, so
equation~\eqref{eq:Pi0} is the vanishing-print reading of
equation~\eqref{eq:aconcaveload} and bounds it above at every positive print;
\item[(iii)] $A_{\mathrm{e}} > 0$ for either sign of $A$, so the equilibrium
exists and is unique at $\ell \ge \bar\rho + \psi_c$ as well, with secant
loading $1 + \sqrt{\ell/(b|g|)}$ at $A = 0$; and both loadings tend to one
as $|g| \to \infty$, the secant at rate $\sqrt{\ell/(b|g|)}$.
\end{enumerate}
\end{proposition}

\begin{proof}
Take $g>0$; oddness in $g$ gives the other sign. On $p>0$ the left side of
equation~\eqref{eq:aquadclose} is $bp^2 + Ap$, so the equation is a quadratic
with strictly negative constant term $-\ell g$ and therefore, for either sign
of $A$, exactly one positive root. Written free of cancellation that root is
$p = 2\ell g/(A + A_{\mathrm{e}})$, and $\Pi_g = 1 + p/g$ is the first
display. Differentiating equation~\eqref{eq:aquadclose} gives
$(2bp + A)\,\partial p/\partial g = \ell$, and $2bp = A_{\mathrm{e}} - A$
turns the bracket into $A_{\mathrm{e}}$, which is the second. Since
$A_{\mathrm{e}} > A$ whenever $b\ell|g| > 0$, we have
$A + A_{\mathrm{e}} < 2A_{\mathrm{e}}$ and the secant exceeds the marginal.
$A_{\mathrm{e}}$ is strictly increasing in $|g|$ and in $\psi_c$, and enters
both denominators positively, which signs the two comparative statics.
For (i) and (ii), $A_{\mathrm{e}} \to A$ under either limit, so
$2\ell/(A+A_{\mathrm{e}}) \to \ell/A$; the bound follows because the secant
is decreasing in $|g|$. For (iii), $A_{\mathrm{e}}$ is a square root of a
sum of a square and a positive term and so is strictly positive at $A \le 0$;
setting $A = 0$ gives $A_{\mathrm{e}} = 2\sqrt{b\ell|g|}$ and the stated
loading, and $A_{\mathrm{e}} \sim 2\sqrt{b\ell|g|}$ as $|g| \to \infty$ gives
the rate.
\end{proof}

\begin{remark}[which loading a reduced form estimates]
\label{rem:asecant}
$\Pi_g$ is a secant, so it depends on the size of the print at which it is
read. A loading of the close on news estimated by least squares over gaps
with $\E[g] = 0$ returns
$\mathrm{Cov}(r_c,g)/\Var(g) = \E[g^2\,\Pi_g(g)]/\E[g^2]$, the
$g^2$-weighted mean secant --- neither the marginal loading nor the secant at
any single print. A measured $\hat\Pi_g$ must therefore be read at the print
size of the sample that produced it; the calibration of
Appendix~\ref{app:calibration} reads it at the $|g| = 2\%$ reference day
throughout. The two loadings and the size dependence at the calibrated
primitives are
\begin{center}
\footnotesize
\begin{tabular}{lrrr@{\hskip 18pt}rrr}
\toprule
 & \multicolumn{3}{c}{SK Hynix} & \multicolumn{3}{c}{Samsung} \\
\cmidrule(lr){2-4}\cmidrule(lr){5-7}
$|g|$ & $\Pi_g$ & $\Pi_g^{\mathrm{m}}$ & $\psi_c = 0$
      & $\Pi_g$ & $\Pi_g^{\mathrm{m}}$ & $\psi_c = 0$ \\
\midrule
$1\%$    & 1.95 & 1.54 & 2.21 & 1.22 & 1.17 & 1.32 \\
$2\%$    & 1.70 & 1.38 & 1.83 & 1.18 & 1.13 & 1.24 \\
$3\%$    & 1.58 & 1.31 & 1.67 & 1.16 & 1.11 & 1.21 \\
$5.56\%$ & 1.43 & 1.23 & 1.48 & 1.13 & 1.08 & 1.16 \\
\bottomrule
\end{tabular}
\end{center}
at $(\ell,\psi_c,b) = (0.583,\,0.260,\,49.7)$ and
$(0.219,\,0.466,\,132.5)$, the last column of each pair being the
zero-capacity ceiling of Remark~\ref{rem:aceiling} and $5.56\%$ the median
absolute daily return of SK~Hynix over the post-launch window.
\end{remark}

\begin{remark}[the zero-capacity ceiling]
\label{rem:aceiling}
Because $\Pi_g$ is decreasing in $\psi_c$, equation~\eqref{eq:aconcaveload}
is bounded above by its $\psi_c = 0$ value at the same print, and a capacity
inverted from a measured overshoot is defined only below that ceiling. The
ceiling falls as the print grows: at SK~Hynix's primitives it is $1.83$ at
the reference day and reaches the pooled target $\hat\Pi_g = 1.66$ at
$|g| = 3.04\%$. The inversion is pooled across names
(Appendix~\ref{app:calibration}) and is never performed name by name.
\end{remark}

Two consequences for Section~\ref{sec:quant} follow. First, the linear
tipping threshold $\ell = \bar\rho + \psi_c$ bounds equation~\eqref{eq:Pi0}
alone: by Proposition~\ref{prop:aconcave}(iii) the concave close prices at
every dose, so what binds the replay is the ringing bound of
Proposition~\ref{prop:vol} on the compounded loading rather than a pole.
Second, by (ii) the linear form evaluated at calibrated primitives returns
the vanishing-print value and therefore overstates the loading at every
print actually observed: at SK~Hynix's primitives equation~\eqref{eq:Pi0}
gives $5.01$ against the $1.70$ that equation~\eqref{eq:aconcaveload} prices
at the reference day.

\begin{proof}[Proof of Proposition~\ref{prop:chaining}]
Write $x$ for $x_t$ and $x_+$ for $x_{t+1}$. Overnight capital clears a
share $\theta$ of the news error, so the next cycle's gap is
$g_+ = \varepsilon_+ - \theta e$. The news component of the close is
$\Pi_g g_+$ and the news component of value is $\varepsilon_+$, so
\[
e_+ \;=\; e + \Pi_g\,g_+ - \varepsilon_+
\;=\; e + \Pi_g(\varepsilon_+ - \theta e) - \varepsilon_+ ,
\]
which collects to equation~\eqref{eq:AR}: an autoregression in $e$ with
coefficient $1-\theta\Pi_g$, stationary if and only if the coefficient lies
inside the unit circle, i.e.\ $0 < \theta\Pi_g < 2$. The session shocks
$u$, $z$ and $z_c$ load on the close through the same day's loadings and
enter no later gap.

For the impulse response, set $u = z = z_c = 0$ and consider a unit public
shock. (i) $e^{(0)} = \Pi_g - 1$. (ii) Iterating,
$e^{(h)} = (1-\theta\Pi_g)^h e^{(0)}$; the open clears $\theta$ of the
standing error, so the overnight leg at horizon $h$ is
$-\theta e^{(h-1)} = -\theta(\Pi_g-1)(1-\theta\Pi_g)^{h-1}$, and the cycle
return in excess of value is
$e^{(h)} - e^{(h-1)} = -\theta\Pi_g\,e^{(h-1)} =
-\theta\Pi_g(\Pi_g-1)(1-\theta\Pi_g)^{h-1}$. (iii) Dividing the two, the
common factor $\theta(1-\theta\Pi_g)^{h-1}$ cancels and the ratio is
$\Pi_g$ exactly, at every horizon and every $\theta$. The no-arbitrage case
nests: $\psi_c = \psi = 0$ gives $\Pi_g = 1/(1-\ell)$.
\end{proof}

\subsection{Thresholds}
\label{app:ladderproof}

\begin{proof}[Proof of Proposition~\ref{prop:ladder}]
Each threshold solves the equation defining the behavior changing there,
using equation~\eqref{eq:aloadings}.

It is convenient to state all four as multiples of the effective depth
$\bar\rho+\psi_c$, which the two-leg demand makes possible: each threshold is
$(\bar\rho+\psi_c)$ times a factor in $(0,1]$, so the ladder is a statement
about the effective dose $\ell_{\text{eff}}$ alone.

\emph{Noise-chasing.} By equation~\eqref{eq:asessload} the loading of
$\bar X^\ast$ on $z$ is $\mathcal{S}\ell - \Psi$ with
$\Psi = \chi[(\bar\rho+\psi_c)-(1-\mathcal{S})\ell]$, so it vanishes at
$\mathcal{S}\ell = \chi[(\bar\rho+\psi_c)-(1-\mathcal{S})\ell]$, i.e.\ at
$\ell = (\bar\rho+\psi_c)\,\psi_c/(\psi_c+\rho \mathcal{S})$, which is
equation~\eqref{eq:anoisethreshold}, with factor
$\psi_c/(\psi_c+\rho \mathcal{S}) \in (0,1]$.

\emph{Oscillation.} $\theta\Pi_g = 2$ reads
$\theta(\bar\rho+\psi_c) = 2(\bar\rho+\psi_c-\ell)$, i.e.\
$\ell = (\bar\rho+\psi_c)(1-\theta/2)$. By
Proposition~\ref{prop:chaining} the chained recursion has coefficient
$1-\theta\Pi_g$, of modulus below one exactly when $0 < \theta\Pi_g < 2$;
since $\theta\Pi_g > 0$ always, the binding side is $\theta\Pi_g < 2$.

\emph{Revelation.} $\Pi_u = 1$ reads $\chi\Pi_g = 1$, i.e.\
$\chi(\bar\rho+\psi_c) = \bar\rho+\psi_c-\ell$, i.e.\
$\ell = (\bar\rho+\psi_c)(1-\chi) = (\bar\rho+\psi_c)\,\rho/(\rho+\psi_c)$,
with factor $1-\chi \in (0,1)$. Under a carry-only session demand $\chi$
would be $\psi_c/(\bar\rho+\psi_c)$ and this threshold would sit at $\bar\rho$;
the two agree at $\rho = 1$.

\emph{Tipping.} $D = 0$ at $\ell = \bar\rho+\psi_c$, factor one, where the
map $\ell \mapsto (\Pi_g,\Pi_u,a_z)$ has a pole.

\emph{Ordering.} Dividing through by $\bar\rho+\psi_c$, the four thresholds are
the factors
\[
\frac{\psi_c}{\psi_c+\rho \mathcal{S}}
\;<\; 1-\frac{\theta}{2}
\;<\; \frac{\rho}{\rho+\psi_c}
\;<\; 1 ,
\]
so the serial ordering is exactly the stated hypothesis, and revelation
$<$ tipping holds whenever $\psi_c > 0$. Reading the two outer
inequalities: the left is $\psi_c\,\theta < (2-\theta)\rho \mathcal{S}$ and the right
is $\theta(\rho+\psi_c) > 2\psi_c$. Both bind on $\psi_c$ relative to
$\rho$, so the ladder is serial when the auction capacity is small
relative to the depth wedge --- and, as in the carry-only case, the
condition is a hypothesis of the proposition rather than something the
calibrated primitives satisfy.
\end{proof}

\subsection{Excess volatility}
\label{app:volproof}

\begin{proof}[Proof of Proposition~\ref{prop:vol}]
By Proposition~\ref{prop:chaining}, $e_+ = (1-\theta\Pi_g)e +
(\Pi_g-1)\varepsilon_+$, so
$\Var(e) = (\Pi_g-1)^2\sigma_0^2/[1-(1-\theta\Pi_g)^2]$,
finite exactly when $0 < \theta\Pi_g < 2$. The close-to-close cycle return
is $R = \Pi_g\varepsilon + \Pi_u u + a_z z + a_{z_c} z_c - \theta\Pi_g e_-$,
so
\[
\Var(R) \;=\; \Pi_g^2\sigma_0^2 + \Pi_u^2\sigma_u^2 + a_z^2\sigma_z^2
+ a_{z_c}^2\sigma_c^2 + \theta^2\Pi_g^2\Var(e) .
\]
Substituting $\Var(e)$ and
$\theta^2\Pi_g^2/[1-(1-\theta\Pi_g)^2] = \theta\Pi_g/(2-\theta\Pi_g)
\equiv h$ gives equation~\eqref{eq:vol}.

Monotonicity. At fixed capacities and $\theta$, $\Pi_u = \chi\Pi_g$,
$a_z = (1-\chi)\Pi_g$ and $a_{z_c} = 1/(1+\psi_c-\ell)$ are positive and
increasing in $\ell$, so the private and noise channels rise. For the
public channel $\Pi_g^2 + h(\Pi_g-1)^2$, note $\theta\Pi_g < 2$ on the
stationary region, so $A \equiv 2-\theta\Pi_g > 0$ and
$\partial[\Pi_g^2 + h(\Pi_g-1)^2]/\partial\Pi_g =
(2/A^2)\big[(1-\theta)\Pi_g(4-\theta\Pi_g) + \theta\big] \ge
(2/A^2)\,\theta > 0$, using $\theta \le 1$ and $4 - \theta\Pi_g > 2$;
since $\Pi_g$ is increasing in $\ell$, the channel rises. Divergence as
$\theta\Pi_g \to 2$ follows from $h \to \infty$ and $\Pi_g > 1$ at
$\ell > 0$.
\end{proof}

\subsection{The impact curve}
\label{app:curveproof}

A base mass $\mu_c$ attends the auction unconditionally; behind it stands a
latent fringe ordered by mobilization cost, unit $j$ paying
$c_0\,j^{\varphi_c}$ per unit per event; set $\delta \equiv
\varphi_c/(\varphi_c + 2) \in (0,1)$. The slope at the touch is $\Lambda_c
\equiv \gamma S_T/\mu_c$. Write $p$ for the concession of the closing price
over $\tilde r$, work in units $\gamma S_T = \mu_c = 1$ (so $\Lambda_c = 1$
and the touch loop gain is $\ell = K$), and let $\beta \equiv
1/\sqrt{2\gamma S_T c_0}$, $p^{\ast} \equiv \mu_c/\beta$, $b \equiv
1/p^{\ast}$.

\begin{lemma}[the impact curve at close]
\label{lem:curve}
The closing auction absorbs $Q$ at a concession $P(Q)$ that is strictly
concave, has slope $\Lambda_c$ at the origin, and satisfies
$P \propto Q^{\delta}$ with $\delta \in (0,1)$ at scale. Its ratio of
marginal to average impact declines from one at the touch to $\delta$, and
where no capacity beyond the incumbents can be drawn the curve is
equation~\eqref{eq:impact} exactly.
\end{lemma}

\begin{proof}[Proof of Lemma~\ref{lem:curve}]
An attending unit facing concession $p$ holds $h = p/(\gamma S_T)$ and
collects surplus $p^2/(2\gamma S_T)$; unit $j$ attends iff the surplus
covers $c_0 j^{\varphi_c}$. The marginal attendee is therefore
$p^2/(2\gamma S_T) = c_0 j^{\varphi_c}$, so the fringe mass is
\begin{equation}
F(p) \;=\; \Big(\frac{p^2}{2\gamma S_T c_0}\Big)^{1/\varphi_c}
\;=\; (\beta p)^{2/\varphi_c} ,
\label{eq:afringe}
\end{equation}
with $\beta = 1/\sqrt{2\gamma S_T c_0}$ as defined above and $F = \beta p$
at $\varphi_c = 2$. Absorption is
$Q(p) = p\,[\mu_c + F(p)]/(\gamma S_T)$, strictly increasing and
convex, with inverse $P$ strictly concave and $P'(0) = \Lambda_c$. Since
$p F'(p) = (2/\varphi_c) F(p)$, marginal impact is
$P'(Q) = \gamma S_T/[\mu_c + (1 + 2/\varphi_c) F]$ against average
$P(Q)/Q = \gamma S_T/[\mu_c + F]$, so the marginal-to-average ratio
declines from one at $F = 0$ to $\delta$ as $F/\mu_c \to \infty$. Dropping
the base and inverting gives $P = \tilde\lambda\, Q^{\delta}$ with
$\tilde\lambda = (\gamma S_T)^{\delta} (2\gamma S_T c_0)^{(1-\delta)/2}$;
as $c_0 \to \infty$, $F \to 0$ and $P(Q) = \Lambda_c Q$.
\end{proof}

\begin{proof}[Proof of Proposition~\ref{prop:curve}]
(i) At $\varphi_c = 2$ the fringe is $F = \beta p$, absorption is
$\gamma S_T\, Q = p\,(\mu_c + \beta p)$, marginal impact
$\gamma S_T/(\mu_c + 2\beta p)$. Each speculator's auction first-order
condition carries own impact $\partial r_c/\partial x_c^i = P'/(1 - K P')
\equiv m_m \Lambda_m$ (curvature does not enter a first derivative), so the
symmetric rule is $x_c = (v - r_c)/(m_m \Lambda_m)$ with $v \equiv g + u$.
Substituting it with $r_c = \tilde r + p$ and $\bar Q = p(1 + bp)$ into the
clearing consistency $\bar Q = K r_c + N x_c$ yields
\begin{equation}
b\,(2N{+}1)\,p^2
\;+\; \big[(N{+}1)(1-\ell) - 2bN(v - \tilde r)\big]\,p
\;-\; \big[\ell\,\tilde r + N(1-\ell)(v - \tilde r)\big] \;=\; 0 ,
\label{eq:quadclose}
\end{equation}
whose positive root is the equilibrium concession; at $b \to 0$ the root is
$r_c = [\psi_c v + \tilde r]/[\psi_c + (1-\ell)]$, i.e.\
Proposition~\ref{prop:passthrough}. (The closed form is stated for the
$N$-trader specification of Appendix~\ref{app:largetraders}; the
continuum's counterpart follows by Proposition~\ref{prop:nesting} on
substituting $\psi_c = N(1-\ell)$, and clauses (i)--(iii) hold in both.)
Equivalently, decomposing the curve at the close into tangent and
intercept, $P(\bar Q) = \Lambda_m \bar Q + A$ with $A \equiv P(\bar Q) -
P'(\bar Q)\,\bar Q$, clearing reads $r_c = m_m(\tilde r + A + \Lambda_m
X_c)$ and the pass-through identity becomes
$\big[\psi_c + (1-\ell_m)\big] r_c = \psi_c v + (1-\ell_m)\,m_m(\tilde r + A)$.

(ii) $P' \le P'(0) = \Lambda_c$ gives $\ell_m \le \ell$. In the
automaton-dominated auction $\bar Q = K r_c$, average loop gain $\ell_a =
K P(\bar Q)/\bar Q = (r_c - \tilde r)/r_c < 1$, and $\ell_m =
[\text{marg}/\text{avg}]\,\ell_a \le \ell_a < 1$. Beyond the linear tipping
threshold the map $r \mapsto \tilde r + P(Kr + X_c)$ is concave with positive
intercept on the absorbing branch and sublinear at infinity, so it crosses
the diagonal exactly once, from above, with slope below one. The large-$K$
limit of $\ell_m$ is $\delta$; the interior peak is exhibited numerically
in the replication package.

(iii) At $\theta = 1$ the news-free skeleton alternates via
$e_+ = e - R(e)$, $R$ the odd close-response function, so a period-two
cycle of amplitude $e^{\ast}$ requires $R(e^{\ast}) = 2e^{\ast}$, i.e.\
$P(2K e^{\ast}) = e^{\ast}$. At $\varphi_c = 2$ the auction absorbs
$Q(p) = p\,(\bar\rho + \psi_c + b p)$ in the normalization above, the base
and the arbitrage sector together supplying the linear term that the
effective depth of equation~\eqref{eq:aloadings} carries. Two caveats
attach to that splice. The linear coefficient has been rescaled from
$\mu_c \equiv 1$ to $\bar\rho+\psi_c$ while the curvature coefficient
$b = 1/p^{\ast}$ is still struck against the unrescaled base, so
$p^{\ast}$ in equation~\eqref{eq:amplitude} is the crossover concession at
which the fringe equals the \emph{dealer base} $\mu_c$, not the effective
linear capacity; reading it as the latter rescales both displays by
$(\bar\rho+\psi_c)$. And $\bar\rho$ is an object of the linear two-venue
model, derived under $P' \equiv \Lambda_c$; carrying it onto the concave
curve assumes the session block is unchanged by curvature at the auction,
which is exact only in the small-$b$ neighbourhood. With those readings,
$Q(e^{\ast}) = 2\ell e^{\ast}$ gives $b e^{\ast} = 2\ell - \bar\rho -
\psi_c$, which is equation~\eqref{eq:amplitude}. It exists exactly above
the oscillation threshold of Proposition~\ref{prop:ladder} read at $\theta = 1$,
and is the unique attractor of $|e_t|$ since $R(e)/e$ is strictly
decreasing in $|e|$. The no-news fixed point solves
$Q(r^{\ast}) = \ell r^{\ast}$, giving
$r^{\ast} = (\ell - \bar\rho - \psi_c)\,p^{\ast}$, which exists exactly
above the tipping threshold; there the slope of the map at zero exceeds one, so
the undisplaced price is unstable, and at $r^{\ast}$ the slope is below
one. At $\bar\rho = 1$ and $\psi_c = 0$ the two displays collapse to
$(2\ell-1)p^{\ast}$ and $(\ell-1)p^{\ast}$.
\end{proof}

\subsection{Conduct}
\label{app:conductproof}

\begin{proposition}[the manufactured order]
\label{prop:manufactured}
A unit of net buying at the close enlarges the automaton's executed order by
\begin{equation}
\frac{\partial q_L}{\partial X_c} \;=\; \frac{\ell}{1-\ell}
\label{eq:manufactured}
\end{equation}
shares, exceeding one beyond $\ell = \tfrac12$. By
equation~\eqref{eq:auctionfoc} the corresponding \emph{individual}
derivative is zero for every trader.
\end{proposition}

\begin{proof}
Differentiating the auction fixed point $r_c = [\tilde r + \Lambda_c(X_c +
z_c)]/(1-\ell)$ in the aggregate purchase $X_c$ gives $\partial
r_c/\partial X_c = \Lambda_c/(1-\ell)$, and $q_L = K r_c$ gives
$\partial q_L/\partial X_c = K\Lambda_c/(1-\ell) = \ell/(1-\ell)$, which
equals one at $\ell = \tfrac12$, free of both capacities and of $\rho$.
Each trader is of measure zero, so $\partial r_c/\partial x^i = 0$ and
$\partial q_L/\partial x^i = 0$ for every $i$; the aggregate derivative is
a property of the equilibrium price system, not a sum of perceived
individual effects. The invariance $\partial q_{\text{prey}}/\partial x
\equiv 0$ that \citet{brunnermeierpedersen2005} use to separate predatory
trading from manipulation therefore fails for the aggregate at every
$\ell > 0$.
\end{proof}

\begin{proposition}[the self-reference share]
\label{prop:share}
Let the predetermined-reference benchmark size the automaton's order at the
session price, which the auction cannot move. The fraction of the closing
overshoot that exists only because the order is sized at the price it
moves is
\begin{equation}
s_R \;=\; \frac{\ell\,(1-\mathcal{S})}{\big[1 - \mathcal{S}(1-\rho)\big] + \psi_c - \ell\,\mathcal{S}}
\;=\; \frac{\ell\,(1-\mathcal{S})}{\bar\rho + \psi_c - \ell\,\mathcal{S}} ,
\label{eq:share}
\end{equation}
zero at $\ell = 0$, strictly increasing in $\ell$, equal to one exactly at
$\ell = \bar\rho+\psi_c$, and equal to $\ell/(1+\psi_c)$ when session
pre-positioning is absent. At $\rho = 1$ it collapses to
$\ell(1-\mathcal{S})/\big[(1+\psi_c) - \ell \mathcal{S}\big]$. Every primitive is held fixed
and only the keying of the order to the price it sets is removed, so a
reference the closing price cannot move sets $s_R = 0$ by construction.
\end{proposition}

\begin{proof}
In the benchmark, $q_L = Kr_s$ with $r_s = g +
\Lambda(\bar X+z)$, so $\partial q_L/\partial r_c = 0$ while the flow still
arrives. Clearing the auction gives $r_c(1+\psi_c) = r_s(1+\ell) +
\psi_c v - \Lambda_c \bar X + \Lambda_c z_c$. Solving the session fixed point
with the two-leg demand of equation~\eqref{eq:asessagg} and taking the
ratio,
\[
s_R \;=\; 1 - \frac{\Pi_g^{P}-1}{\Pi_g-1}
\;=\; \frac{\ell\,(1-\mathcal{S})}{\bar\rho + \psi_c - \ell\,\mathcal{S}} ,
\]
which is equation~\eqref{eq:share}. It is zero at $\ell = 0$; strictly
increasing in $\ell$, since the numerator rises and the denominator falls;
equal to one exactly at $\ell = \bar\rho+\psi_c$, where numerator and
denominator coincide; and at $\mathcal{S} = 0$ it is
$\ell/(\bar\rho+\psi_c)|_{\bar\rho=1} = \ell/(1+\psi_c) =
\ell_{\text{eff}}$. Note that $s_R$ is expressed in the measured share $\mathcal{S}$
and not in $\psi$: the determinant form $\ell/(\DN+\ell)$ that a
carry-only session demand produces is \emph{not} equivalent to the display
above once the session position carries a value leg, whereas the
$\mathcal{S}$-form is, and it is the $\mathcal{S}$-form the quantification uses.
\end{proof}

The same session solution delivers the measurement identity used in
Section~\ref{sec:quant}. Writing $w_g$ for the session price's loading on
the public gap in excess of one,
\begin{equation}
\underbrace{w_g}_{\text{session drift}} = \frac{\mathcal{S}\,\ell}{D},
\qquad
\underbrace{\Pi_g - 1 - w_g}_{\text{auction-window pop}} = \frac{(1-\mathcal{S})\,\ell}{D},
\qquad\text{so}\qquad
\frac{w_g}{\Pi_g - 1 - w_g} \;=\; \frac{\mathcal{S}}{1-\mathcal{S}}
\label{eq:adrift}
\end{equation}
exactly, at every loop gain and every $(\rho,\psi_c)$, where
$D = \bar\rho+\psi_c-\ell$. Equivalently, the session's share of the
overshoot is
\begin{equation}
\mathcal{S} \;\equiv\; \frac{w_g}{\Pi_g - 1}
\;=\; \frac{\rho\,\psi}{\rho + \psi_c + \rho\,\psi} ,
\qquad\text{so}\qquad
\psi \;=\; \frac{\mathcal{S}}{1-\mathcal{S}}\cdot\frac{\rho+\psi_c}{\rho},
\qquad
\bar\rho \;=\; 1 - \mathcal{S}(1-\rho) .
\label{eq:asmap}
\end{equation}
Two readings of this map must be kept apart. The step from the measured
$\mathcal{S}$ to $\bar\rho$, and through it to every public-channel object, is free
of $\ell$, $\psi_c$ and $\rho$ beyond what $\bar\rho$ itself displays: a
single measured share pins the \emph{price} consequences of the session
leg at every dose and both depths, which is all Section~\ref{sec:quant}
uses. The step from $\mathcal{S}$ to the structural capacity $\psi$ is not: it
carries the factor $(\rho+\psi_c)/\rho$, so recovering $\psi$ requires
$\psi_c$ and $\rho$ as well, and the implied $\psi$ is therefore
name-specific wherever $\psi_c$ is. Under a carry-only session demand the
factor would be one and the two readings would coincide; that coincidence
is an artifact of omitting the value leg, not a property of the model.

\begin{proof}[Proof of Proposition~\ref{prop:predshare}]
The benchmark economy is the equilibrium speculator with the carry leg of
equation~\eqref{eq:asessdemand} switched off and the value leg left intact,
i.e.\ $\tau_s \to 0$ with $\tau$ unchanged. This is exactly the benign
standard of Section~\ref{sec:conduct}: the benchmark trader's session
position is $\bar X = \tau(v - r_s)$, always directed against the mispricing,
and he corrects at the auction precisely as the equilibrium speculator
does. It is a restriction on the trader's \emph{motive}, not on the price
system, which is why the two-leg demand gives the benign benchmark a
structural definition; a carry-only session model has no such
decomposition, and setting $\psi = 0$ there removes the session position
altogether rather than only its pre-positioning component.

Solving the session fixed point at $\tau_s = 0$ gives $\mathcal{S} = 0$ and hence
$\bar\rho = 1$, so
$\Pi_g^{\,0} = (1+\psi_c)/(1+\psi_c-\ell)$: the value leg trades against
the gap and moves the close, but it does so symmetrically at both venues
and leaves no depth wedge behind. Then
\[
\Pi_g - \Pi_g^{\,0}
= \frac{(\bar\rho+\psi_c)(1+\psi_c-\ell) - (1+\psi_c)(\bar\rho+\psi_c-\ell)}
       {(\bar\rho+\psi_c-\ell)(1+\psi_c-\ell)}
= \frac{\ell\,(1-\bar\rho)}{(\bar\rho+\psi_c-\ell)(1+\psi_c-\ell)} ,
\]
and dividing by $\Pi_g - 1 = \ell/(\bar\rho+\psi_c-\ell)$ gives
$\mathcal{P} = (1-\bar\rho)/(1+\psi_c-\ell)$, which is
$\mathcal{S}(1-\rho)/(1+\psi_c-\ell)$ by $\bar\rho = 1-\mathcal{S}(1-\rho)$; the identity
$\mathcal{S} = \psi/(1+\psi)$ is equation~\eqref{eq:asmap}. The signs require
$\ell < 1+\psi_c$ in addition to interiority, which is automatic in the
empirically relevant case $\rho \le 1$ (then
$\bar\rho \le 1$ and $\ell_{tip} = \bar\rho+\psi_c \le 1+\psi_c$) but not
in general: at $\rho > 1$ the band
$\ell \in (1+\psi_c,\,\bar\rho+\psi_c)$ is interior with $1-\bar\rho$ and
$1+\psi_c-\ell$ both negative, so $\mathcal{P} > 0$ there despite
$\rho > 1$, while the benchmark itself has no interior equilibrium
($\Pi_g^{\,0} < 0$). On $\rho \le 1$: $\mathcal{P} > 0$ iff $\bar\rho < 1$
iff $\rho < 1$ and $\psi > 0$; the stated monotonicities are immediate
from the closed form, and
$1 - \mathcal{P} = (\bar\rho+\psi_c-\ell)/(1+\psi_c-\ell) \in (0,1]$ on the
interior region, vanishing at the pole. For the quantity margin,
$r_c(1-\ell+\psi_c) = \tilde r + \psi_c v$ with
$\partial\tilde r/\partial \bar X = \Lambda - \Lambda_c$ gives
$\partial q_L/\partial \bar X = K(\Lambda-\Lambda_c)/(1-\ell+\psi_c) > 0$
for $\rho < 1$.
\end{proof}

\subsection{Welfare}
\label{app:welfareproof}

\begin{proof}[Proof of Proposition~\ref{prop:loss}]
The automaton trades $K r_c$ at $r_c$ against value $v$, so per unit of
rebalancing capital $W_h = \E[r_c(v-r_c)]$. Substituting the loadings and
$v = g + u$ with $g = \varepsilon$ in an isolated cycle,
\[
W_h = \E[r_c v] - \E[r_c^2]
= \big(\Pi_g\sigma_0^2 + \Pi_u\sigma_u^2\big)
- \big(\Pi_g^2\sigma_0^2 + \Pi_u^2\sigma_u^2 + a_z^2\sigma_z^2
+ a_{z_c}^2\sigma_c^2\big),
\]
which is equation~\eqref{eq:loss}. The public channel
$-\Pi_g(\Pi_g-1)\sigma_0^2 < 0$ at every $\ell > 0$ by
Corollary~\ref{cor:overshoot}; the private channel
$-\Pi_u(\Pi_u-1)\sigma_u^2$ is positive exactly when $\Pi_u < 1$, i.e.\
below $\ell = \bar\rho$; the noise channels are non-positive.

For the threshold, write $T = \bar\rho+\psi_c$ and $D = T-\ell > 0$. With
$\Pi_g = T/D$, $\Pi_u = \chi T/D$ and $a_z = (1-\chi)T/D$ in return units,
the identities $\Pi_g(\Pi_g-1) = T\ell/D^2$ and
$\Pi_u(\Pi_u-1) = \chi T[\ell - (1-\chi)T]/D^2$ give $W_h = -N(\ell)/D^2$
with
\[
N(\ell) \;=\; T\Big\{
\underbrace{\big[\sigma_0^2 + \chi\,\sigma_u^2\big]}_{\textstyle A\,>\,0}\,\ell
\;-\; \chi(1-\chi)\,T\,\sigma_u^2 \;+\; (1-\chi)^2\,T\,\sigma_z^2 \Big\} ,
\]
whose root is
\begin{equation}
\ell^{\ast} \;=\;
\frac{(\bar\rho+\psi_c)\,(1-\chi)\,
      \big[\chi\,\sigma_u^2 \;-\; (1-\chi)\,\sigma_z^2\big]}
     {\sigma_0^2 \;+\; \chi\,\sigma_u^2} ,
\qquad
\chi \;=\; \frac{\psi_c}{\rho+\psi_c} .
\label{eq:alosthreshold}
\end{equation}
This is equation~\eqref{eq:losthreshold} written in the general form. Under
a carry-only session demand, where $\chi = \psi_c/T$ and $1-\chi =
\bar\rho/T$, it reduces term by term to
$\bar\rho(\psi_c\sigma_u^2 - \bar\rho\sigma_z^2)/[(\bar\rho+\psi_c)\sigma_0^2
+ \psi_c\sigma_u^2]$; with the two-leg demand the same expression is
evaluated at $\chi = \psi_c/(\rho+\psi_c)$ instead. The revelation threshold is
$(1-\chi)T$ in both, so the comparison below is unaffected in form.
Substituting the reduced form,
\[
N(\ell) \;=\; \underbrace{\big[(\bar\rho+\psi_c)\sigma_0^2
+ \psi_c\sigma_u^2\big]}_{\textstyle A\,>\,0}\,\ell
\;-\; \psi_c\bar\rho\,\sigma_u^2 \;+\; \bar\rho^2\sigma_z^2
\qquad(\text{carry-only }\chi) .
\]
(the noise term is $\bar\rho^2\sigma_z^2$ and not
$\bar\rho^2\Lambda^2\sigma_z^2$ because, by the units convention of
Appendix~\ref{app:notation}, $\sigma_z$ is here the variance of the
\emph{return} displacement $\Lambda z$ and the relevant loading is
$a_z = \bar\rho/D$ of equation~\eqref{eq:aunits}).
$N$ is affine with strictly positive slope, hence has the unique root
$\ell^{\ast}$ of equation~\eqref{eq:losthreshold}, and $W_h$ crosses zero
once, from above. The threshold is positive exactly when $W_h(0) > 0$,
i.e.\ exactly when $\psi_c\sigma_u^2 > \bar\rho\,\sigma_z^2$; when that
fails, $\ell^{\ast} \le 0$ and holders lose at every positive dose. It
lies strictly below the revelation threshold $(1-\chi)T$ at every variance mix:
\begin{equation}
1 \;-\; \frac{\ell^{\ast}}{(1-\chi)\,T}
\;=\; \frac{\sigma_0^2 \;+\; (1-\chi)\,\sigma_z^2}
           {\sigma_0^2 \;+\; \chi\,\sigma_u^2} \;>\; 0 ,
\label{eq:albelowrev}
\end{equation}
which at the carry-only $\chi$ is the display
$\big[\sigma_0^2(\psi_c+\bar\rho) + \bar\rho\sigma_z^2\big]\big/
\big[\sigma_0^2(\psi_c+\bar\rho) + \psi_c\sigma_u^2\big]$ and the threshold is
$\bar\rho$.

Two qualifications to the comparative statics stated with the proposition.
First, $\ell^{\ast}$ is decreasing in $\sigma_0/\sigma_u$ --- equivalently,
the display above is increasing in it --- only where $\ell^{\ast} > 0$:
\[
\frac{\partial \ell^{\ast}}{\partial \sigma_0}
\;=\; \frac{-2\bar\rho\,\sigma_0(\psi_c+\bar\rho)
      \big(\psi_c\sigma_u^2 - \bar\rho\,\sigma_z^2\big)}
     {\big[\sigma_0^2(\psi_c+\bar\rho) + \psi_c\sigma_u^2\big]^2} ,
\]
whose sign is the sign of $-(\psi_c\sigma_u^2 - \bar\rho\sigma_z^2)$, i.e.\
of $-\ell^{\ast}$. Second, $\ell^{\ast}$ is \emph{not} increasing in
$\bar\rho$ in general. Differentiating,
\[
\operatorname{sign}\frac{\partial\ell^{\ast}}{\partial\bar\rho}
\;=\;
\operatorname{sign}\Big[
\psi_c^2\,\sigma_u^2\big(\sigma_0^2+\sigma_u^2\big)
\;-\; 2\psi_c\bar\rho\,\sigma_z^2\big(\sigma_0^2+\sigma_u^2\big)
\;-\; \bar\rho^2\sigma_0^2\sigma_z^2 \Big] ,
\]
which is positive unconditionally at $\sigma_z = 0$ but turns negative on
an open set with $\ell^{\ast} > 0$: at
$(\bar\rho,\psi_c,\sigma_0,\sigma_u,\sigma_z) = (0.9,\,0.5,\,1,\,1,\,0.6)$,
$\ell^{\ast} = 0.083$ and $\partial\ell^{\ast}/\partial\bar\rho = -0.122$.
The numerator $\bar\rho(\psi_c\sigma_u^2 - \bar\rho\sigma_z^2)$ is a
downward parabola in $\bar\rho$ peaking at
$\bar\rho = \psi_c\sigma_u^2/(2\sigma_z^2)$, so the monotonicity holds
exactly on
$\sigma_z^2 < \psi_c^2\sigma_u^2(\sigma_0^2+\sigma_u^2)\big/
\big[2\psi_c\bar\rho(\sigma_0^2+\sigma_u^2) + \bar\rho^2\sigma_0^2\big]$
and fails above it. Since the Korean calibration sits at small
$\sigma_z$ relative to $\psi_c\sigma_u$, the stated sign is the empirically
relevant one, but it is a side condition and not a theorem.

The single crossing does not come from
monotonicity of $W_h$, which fails: the private channel
$\psi_c(\bar\rho-\ell)/D^2$ has derivative
$\psi_c(\bar\rho-\psi_c-\ell)/D^3$ and rises on $\ell < \bar\rho-\psi_c$
whenever $\bar\rho > \psi_c$; the argument above shows this occurs only
while $W_h$ is on one side of zero.

The statement is exact at $\sigma_c = 0$ and is a bound otherwise. The
closing-imbalance channel is priced on the auction book and carries the
denominator $1+\psi_c-\ell$ rather than $D$, so $W_h = -N(\ell)/D^2 -
a_{z_c}^2\sigma_c^2$ does not collapse into a single affine numerator and
equation~\eqref{eq:losthreshold} is not the crossing point when
$\sigma_c > 0$. Because the omitted term is strictly negative and
increasing in $\ell$, it moves the crossing \emph{down}: $\ell^{\ast}$ of
equation~\eqref{eq:losthreshold} is an upper bound on the dose at which
holders begin to lose, so $W_h < 0$ \emph{whenever} $\ell > \ell^{\ast}$,
and the proposition's converse should be read at $\sigma_c = 0$. That
$W_h$ still crosses zero exactly once when $\sigma_c > 0$ is supported by
a $3\times10^5$-draw random search over
$(\bar\rho,\psi_c,\rho,\sigma_0,\sigma_u,\sigma_z,\sigma_c)$ that returned
no parameter vector with more than one sign change, but it is not proved
here: the affine-numerator argument covers only $\sigma_c = 0$.
\end{proof}

\begin{proof}[Proof of Proposition~\ref{prop:ledger}]
Every transfer is an exchange at a common price. During the session, $\bar X$
and $z$ are absorbed by dealers at $r_s$; at the auction, $K r_c$, $X_c$
and $z_c$ are absorbed by dealers at $r_c$; overnight, all positions are
marked at $v$. Writing each agent's profit as $\sum$ (quantity)
$\times$ ($v$ minus execution price) and summing across agents, the
price terms cancel event by event because quantities sum to zero at each
common price, and the $v$ terms cancel because net quantities sum to zero;
total profit is identically zero for every realization of $(g,u,z,z_c)$.

For the incidence, take expectations with equation~\eqref{eq:aloadings}.
The automaton's loss grows like $\Pi_g^2$ --- its gross flow is $K r_c$ and
$r_c$ itself carries the multiplier --- while the sector's position is
capped by equation~\eqref{eq:auctionfoc} at $\tau(v-r_c)$, which grows only
like $\Pi_g$. The sector's share of the holders' loss,
$\mathcal{A}(\ell) \equiv \E[\text{arbitrage profit}]/(-\E[\text{holder
loss}])$, therefore declines in $\ell$ wherever holders in fact lose, i.e.\
on $\ell > \ell^{\ast}$. It does not decline to a small number. As the dose
approaches the tipping threshold every profit is dominated by its
$O(D^{-2})$ term, the shock variances cancel between numerator and
denominator, and the share converges to a closed form free of the variance
mix,
\begin{equation}
\lim_{\ell \uparrow \bar\rho+\psi_c} \mathcal{A}(\ell)
\;=\; \frac{\psi_c \;+\; \rho\,\mathcal{S}\,(1-\mathcal{S})}{\bar\rho + \psi_c} ,
\label{eq:aincidence}
\end{equation}
where $\mathcal{S}$ is the session share of equation~\eqref{eq:asmap}. To see it,
write $T = \bar\rho+\psi_c$ and evaluate at $\ell = T$: the three session
loadings of equation~\eqref{eq:asessload} then share the common factor
$\hat C \equiv T\big[g + \chi u + (1-\chi)\Lambda z\big]$, since
$\Psi = \chi T \mathcal{S}$ and $\mathcal{S}\ell = sT$ there, so
$r_c \simeq \hat C/D$ and $\bar X \simeq \mathcal{S}\,\hat C/(\Lambda D)$ --- the
session position is the share $\mathcal{S}$ of the closing displacement, which is
what $\mathcal{S}$ means. Substituting into
$\bar X(v-r_s) + X_c(v-r_c) = X(v-r_c) + \bar X(r_c-r_s)$, using
$r_s \simeq \mathcal{S}\hat C/D$, and dividing by $-K\E[\hat C^2]/D^2$ gives
$\big[\tau + \mathcal{S}(1-\mathcal{S})/\Lambda\big]/K$, which is
equation~\eqref{eq:aincidence} after multiplying through by $\Lambda_c$. Since $1-\mathcal{S}+\rho \mathcal{S}^2 > 0$ on
$\mathcal{S} \in [0,1)$, the limit lies strictly in $(0,1)$, so the market makers
retain a strictly positive residual share at the boundary. At the
calibration of record it is $0.75$ for Samsung and $0.63$ for SK~Hynix:
near the boundary the two sectors split the holders' loss, with the
speculators taking the larger part.

At the other end the ratio is not informative about incidence. Below
$\ell^{\ast}$ the holders do not lose at all and $\mathcal{A}$ changes sign;
just above it $\mathcal{A}$ exceeds one --- the speculators earn multiples
of the holders' loss, because most of their profit is sourced from the noise
traders rather than from the complex, and over that range the market makers
can lose money. Only the limit~\eqref{eq:aincidence} and the decline toward
it are properties of the model rather than of the variance mix.

The sector's auction certainty equivalent is
$\tfrac12\tau\,\E[(r_c-v)^2]$. Restricted to the three channels priced on
$D$ this is proportional to $\psi_c/(\bar\rho+\psi_c-\ell)^2$, with
$\partial[\psi_c/D^2]/\partial\psi_c = (\bar\rho-\ell-\psi_c)/D^3$, so it
rises in capacity up to $\psi_c = \bar\rho-\ell$ and falls thereafter; the
$z_c$ channel carries the denominator $1+\psi_c-\ell$ instead and shifts
that turning point, so the proportionality and the stated turning point
hold exactly only at $\sigma_c = 0$. In either case it is a competitive
return to risk-bearing with no markup component.
\end{proof}

At fixed $\ell$, $W_h$ is hump-shaped in $\psi_c$ (the public and noise
channels improve while $\Pi_u = \chi\Pi_g$ rises toward one, extinguishing
the momentum subsidy) and strictly decreasing in $\psi$, which enters
only through $\bar\rho$.

\section{Additional Theoretical Results and Extensions}
\label{app:theory}

Table~\ref{tab:status} lists the model's price predictions, the test this
paper runs, and the verdict. The quantity, conduct, welfare and dynamic
predictions are settled in the body rather than tabulated: the loop gain, the
capacity ratio and the absorption ratio are measured in
Section~\ref{sec:calibration}; the pre-positioning share is $47\%$ pooled; the
AUM ratchet holds at a compounding slope of $1.100$; and realized tracking
error does not reveal the distortion. Three predictions are left untested
because the required data channel is closed --- the split of the transfer
between dealers and other counterparties is not separately identified, the
private-news sign flip at the exit threshold needs an intra-session instrument for
$u$, and the same-day displacement per unit of flow is null at this power. One
is \emph{refuted}: liquidity-provider withdrawal does not widen the disparity,
which peaks in June and bears no relation to it.

\subsection{Incidence: the ledger}
\label{app:ledgerprop}

\begin{proposition}[the ledger]
\label{prop:ledger}
The profits and losses of the automaton, the noise traders, the
speculators, and the market makers sum to zero state by state. In
expectation, on the range where holders lose ($\ell > \ell^{\ast}$) the
speculators' share of the holders' loss declines in $\ell$, from a
multiple of that loss just above $\ell^{\ast}$ --- where the speculators
are paid mostly by the noise traders --- toward the limit
$(\psi_c + \rho\,\mathcal{S}(1-\mathcal{S}))/(\bar\rho+\psi_c)$ at the stability boundary,
which is free of the variance mix and strictly interior to $(0,1)$. The
balance accrues to the market makers as inventory compensation. At the
calibrated primitives the boundary split is $0.75$/$0.25$ for Samsung and
$0.63$/$0.37$ for SK~Hynix.
\end{proposition}

The proof is in Appendix~\ref{app:welfareproof}. The market makers' column
is the return to attending the auction, the entry incentive
Appendix~\ref{app:rho} closes on.

\subsection{Attendance and the capacity ratio}
\label{app:rho}

Attending capacity $\mu_c$, and with it $\rho$, is an equilibrium object:
capacity attends the auction until the profit per unit of attending capacity
equals the participation cost. Write $\Pi_{MM}(\ell)$ for the market
makers' column in Proposition~\ref{prop:ledger} and substitute
$\mu_c = \gamma S_T K/\ell$.

\begin{proposition}[free-entry attendance]
\label{prop:attendance}
Let capacity attend until
\begin{equation}
\Phi(\ell) \;\equiv\; \frac{\Pi_{MM}(\ell)\,\ell}{\gamma S_T}
\;=\; c\,K .
\label{eq:attendance}
\end{equation}
Entry is stable exactly where $\Phi$ is increasing. Assume in addition that
$\Pi_{MM}$ is non-decreasing in $\ell$, so that
$\varepsilon_\Phi \equiv d\log\Phi/d\log\ell = 1 + d\log\Pi_{MM}/d\log\ell
\ge 1$; then at any stable solution:
\emph{(i)} $d\log\ell/d\log K = 1/\varepsilon_\Phi < 1$, equivalently
$d\log\Lambda_c/d\log K = 1/\varepsilon_\Phi - 1 < 0$;
\emph{(ii)} a lever that cuts $K$ raises $\Lambda_c$, so the loop gain
falls by less than the cut, by the factor $1/\varepsilon_\Phi$;
\emph{(iii)} with the rebalance split across two auctions receiving $K_o$ and
$K_c$, $\kappa \equiv \Lambda_c/\Lambda_o
= \kappa_0 (K_o/K_c)^{1-1/\varepsilon_\Phi}$;
\emph{(iv)} entry is stabilizing, so the measured loop gain already embeds
the observed regime's attendance, and only counterfactuals that move flow
are affected.
\end{proposition}

\begin{proof}
(i) Differentiate equation~\eqref{eq:attendance} logarithmically at
$\Phi' > 0$; the second form follows from $\Lambda_c = \ell/K$. The strict
inequalities require $\varepsilon_\Phi > 1$, which is the maintained
assumption above and not implied by stability: $\Phi' > 0$ delivers only
$\varepsilon_\Phi > 0$, and on $\varepsilon_\Phi \in (0,1)$ the signs in
(i)--(ii) reverse, so that entry would amplify rather than damp a change in
$K$. The assumption is the natural one here because $\Pi_{MM}$ is a
concession collected on a gross flow that scales with $\Pi_g$, but we do not
derive it. (ii) Immediate from (i). (iii) Apply (i) at each venue and
divide; this integrates a local elasticity into a global power law, so it
requires in addition that $\Phi$ be a power function with the same exponent
and the same participation cost at both venues. (iv) By (i),
$d\ell/dK > 0$ but damped relative to fixed attendance. The measurement
statement is not definitional: it is the claim that $\hat\ell$ is an
equilibrium object of the observed attendance regime, which holds because
$\hat\ell$ is estimated from realized closing prices in that regime, and it
carries no implication for a counterfactual that moves $K$. Because $\Phi$
admits multiple stable solutions (see below), (i)--(iii) are local to
whichever obtains.
\end{proof}

$\Phi$ is not monotone: it is negative and falling at low loop gains and
steeply positive at high ones. At a shallow close it has a single turning
point and the stable solution above it is unique; at a deep one it can turn
several times, so the attendance game can admit multiple equilibria and the
comparative statics are local to whichever obtains.
Consequently $\hat\rho$ and $\hat\ell$ are valid for the observed regime
and not transportable to a counterfactual that changes the loop gain;
$\varepsilon_\Phi$ is bounded rather than point-identified (fixed
attendance is $\varepsilon_\Phi \to \infty$), and
Section~\ref{sec:schedule} reports counterfactuals as a band. The
cross-sectional prediction is tested: the auction-to-continuous impact
ratio is $0.24$--$0.29$ for the treated stocks against $0.41$ for matched
controls (Appendix~\ref{app:impact}), deeper exactly where the predictable
order arrives, a comparison the liquidity-provider exemption cannot
generate since it covers both groups. Across the launch the auction's traded
value rose two fifths while its spread did not move
(Table~\ref{tab:attendance}); the threshold ladder and the ratchet are
stated at fixed attendance, so each threshold maps to a critical $K$ and the
ratchet's crossing probabilities are upper bounds.

\subsection{The equilibrium on the impact curve}
\label{app:curve}

\begin{proposition}[the close on the impact curve]
\label{prop:curve}
Let attendance costs be distributed as in Appendix~\ref{app:curveproof},
and let $p^{\ast}$ be the crossover concession at which the fringe mass
equals the base.
\emph{(i) The close.} At $\varphi_c = 2$ (the square-root law,
$\delta = \tfrac12$, bracketed by the measured exponents), the fringe is
linear in the concession and the equilibrium concession solves the
quadratic~\eqref{eq:quadclose}, whose positive root collapses to
Proposition~\ref{prop:passthrough} as the fringe vanishes.
\emph{(ii) No pole.} At the equilibrium price $\ell_m \le \ell$ always, and
in the automaton-dominated auction $\ell_m < 1$ with no condition: beyond the
linear tipping threshold the close saturates rather than ceasing to exist. As
the complex grows, $\ell_m$ rises, peaks below one at intermediate scale,
and declines to $\delta$.
\emph{(iii) Bounded ringing.} Above the oscillation threshold the news-free
dynamics converge to a period-two oscillation of closed-form amplitude; at
$\varphi_c = 2$ and $\theta = 1$,
\begin{equation}
e^{\ast} = \big(2\ell - \bar\rho - \psi_c\big)\,p^{\ast}
\;\;\Big(\ell > \tfrac{\bar\rho+\psi_c}{2}\Big),
\qquad
r^{\ast} = \big(\ell - \bar\rho - \psi_c\big)\,p^{\ast}
\;\;\big(\ell > \bar\rho+\psi_c\big),
\label{eq:amplitude}
\end{equation}
where $r^{\ast}$ is a self-referential displacement of the close on
news-free days, the undisplaced price being unstable there. The two
conditions are the oscillation and tipping thresholds of
Proposition~\ref{prop:ladder}, and at $\bar\rho = 1$, $\psi_c = 0$ the
displays reduce to $(2\ell-1)p^{\ast}$ and $(\ell-1)p^{\ast}$.
\end{proposition}

By clause (i) every first-order condition at the equilibrium price is the
linear one at $(\Lambda_m, m_m)$ shifted by the tangent wedge, so
Propositions~\ref{prop:passthrough}--\ref{prop:ladder} are read at the
operating size, and the three primitives separate: the fitted exponent
identifies $\varphi_c$, the touch slope $\mu_c$, and the level $c_0$
(Table~\ref{tab:g4calibration}). Clause (iii) prices the days on which
$\theta\Pi_g \ge 2$, where the linear recursion has no stationary
magnitude.

\subsection{The AUM ratchet}
\label{app:ratchet}

A leveraged fund's assets compound at the leveraged rate, so
\begin{equation}
A_+ = A(1 + Lr)
\quad\Longrightarrow\quad
K_+ = K(1+Lr),
\qquad
\ell_+ = \ell\,(1+Lr) ,
\label{eq:aratchet}
\end{equation}
with no free parameter.

\begin{proposition}[the ratchet]
\label{prop:ratchet}
With $K$ evolving by equation~\eqref{eq:aratchet}, the complex crosses the
stability threshold with positive probability from starting loop gains at which
the fixed-$K$ system never crosses it. Simulating 60 cycles at $L=2$ with a
2.5\% daily return scale ($\psi_c = 0.5$, $\rho = 0.43$, $\psi$ tracking
the loop gain, 20{,}000 paths), the crossing probability is $0.000$ at
fixed $K$ for every starting value, against $0.001$, $0.025$, $0.135$ and
$0.381$ at endogenous $K$ from $\ell_0 = 0.20$, $0.30$, $0.40$ and $0.50$.
\end{proposition}

The mean loop gain falls along the simulated paths (volatility drag) while
the crossing probability rises: the threshold is a right-tail statement.
The empirical counterpart and the policy consequence are in
Appendix~\ref{app:ratchetsim}.

\subsection{The venue split $\varphi$}
\label{app:phi}

The mandate is sized at the reference close, $K r_c$, whatever the split. A share
$\varphi$ of it executes in the closing auction that sets $r_c$; the residual
$1-\varphi$ executes in single-stock futures, which print at 15:45, after $r_c$ is
fixed (Section~\ref{sec:close}). Only the share executing at the reference
participates in the auction's fixed point, so on the concave close of
Appendix~\ref{app:curveproof} the loop gain at the split is
\[
\ell(\varphi) \;=\; \Lambda_c(\varphi K)\,\varphi K
\;=\; \bar\ell\,\Big(\frac{\varphi}{\varsigma}\Big)^{\hat\delta} ,
\]
since $\Lambda_c(Q) = P(Q)/Q \propto Q^{\hat\delta-1}$ at scale and $\varsigma$ is the
split at which $\bar\ell$ is measured. The residual leg prints against its own book,
enters no fixed point, and is priced at that venue's measured depth rather than at an
assumed one (Appendix~\ref{app:venue}). The status quo is $\varphi = \varsigma = 0.49$,
the spot slice of Appendix~\ref{app:families} applied to worldwide $K$; the
swap-pass-through accounting of that appendix instead implies at-the-close shares of
$0.44$ (Samsung) and $0.76$ (SK~Hynix), which the quantification does not adopt.
$\varphi$ is pinned by funding constraints, the transaction-tax asymmetry, and
short-sale restrictions, and is therefore exogenous to the predation game.

\begin{proposition}[the venue split]
\label{prop:phi}
On the interior branch --- $\ell(\varphi) < \bar\rho + \psi_c$ on the linear branch, and
everywhere on the concave one, where the print is bounded:
\begin{enumerate}
\item[(i)] $\ell(\varphi)$ is strictly increasing in $\varphi$, and so are the closing
loading, the manufacturing leverage, and the excess volatility of
Proposition~\ref{prop:vol}. On the linear branch
$\Pi_0(\varphi) = (\bar\rho+\psi_c)/(\bar\rho+\psi_c-\ell(\varphi))$.
\item[(ii)] The residual leg enters neither the loading nor the mandate: the fund settles
at $r_c$ and delivers $L\,r_c$ at every $\varphi$, so tracking error is zero throughout.
\item[(iii)] Holder cost is the sum of the transfer at the reference, increasing in
$\varphi$ through (i), and the residual leg's execution cost, which at a fixed return
path is proportional to $(1-\varphi)^{1+\delta_f}$ for $\delta_f$ the concavity of the
venue that clears it, and is decreasing. It is therefore not monotone in general, and
the sign of its derivative is set by the two venues' depths rather than by the split.
\item[(iv)] On the linear branch an interior equilibrium requires $\varphi < \varsigma\,
[(\bar\rho+\psi_c)/\bar\ell\,]^{1/\hat\delta}$; on the concave branch the print is
bounded and no such restriction binds.
\end{enumerate}
\end{proposition}

\begin{proof}
(i) $\ell(\varphi)$ is a positive power of $\varphi$ and hence increasing; the loading is
increasing in $\ell$ on either branch, and Proposition~\ref{prop:vol} is increasing in the
loading. (ii) The residual prints after $r_c$, so it appears in no equation determining
$r_c$, and the net asset value is struck at $r_c$, so the delivered multiple is unchanged.
(iii) The transfer inherits (i). The residual order is $(1-\varphi)K|r|$ and its concession
is proportional to $[(1-\varphi)K|r|]^{\delta_f}$, so quantity times concession is
proportional to $(1-\varphi)^{1+\delta_f}$ at a fixed path; a sum of an increasing and a
decreasing function need not be monotone. (iv) Solve $\ell(\varphi) < \bar\rho+\psi_c$ for
$\varphi$; on the concave branch $P$ is bounded above, so no pole exists.
\end{proof}

\subsection{A large trader at both venues}
\label{app:phigame}

The quantification's residual leg is cleared by the venue and is not available to the
speculator. Relaxing that gives a three-stage game, solved here for a single large trader
rather than the continuum of the body: a share $\varphi$ executes at the close and
$1-\varphi$ at a displaced venue anchored to it,
$r_f = r_c + \Lambda_f[(1-\varphi)K r_c + x_f]$ with $\Lambda_f = \rho_f\Lambda$, at
which the trader may also trade. Only the share executing at the close participates in
the fixed point, so the same-price multiplier is $m_\varphi = 1/(1-\varphi\ell)$. The
loci below are properties of that game and are not the calibration; $\rho_f$ is swept
rather than measured.

\begin{proposition}[the split with a large trader at both venues]
\label{prop:phigame}
The three-stage game (session, auction, displaced venue) with closed-loop play
has a unique interior linear equilibrium wherever $\varphi\ell < 1$, both
second-order conditions hold, and no pole has been crossed, the pole locus
being
\[
\mathcal{P}(\ell,\varphi) \;\equiv\;
\ell^2\rho_f^2(1-\varphi)^2 + 2\ell\rho\rho_f(1+\varphi)
+ \rho^2 - 4\rho\rho_f + \rho_f \;=\; 0 ,
\]
and it collapses to the core as $\rho_f \to \infty$ with $\varphi = 1$. On
that region:
\begin{enumerate}
\item[(i)] Even at $\varphi = 1$, the existence of a usable displaced venue
raises the overshoot: $\Pi_g$ increases as $\rho_f$ falls.
\item[(ii)] The speculator's closing-auction trade flips from
provision-side to automaton-side as $\varphi$ falls below a threshold
$\varphi_{bang}(\ell)$, increasing in the loop gain below the dose at which
it saturates at one.
\item[(iii)] Holder loss is U-shaped in $\varphi$; the minimizer
$\varphi^\ast(\ell)$ is interior and strictly decreasing in the loop gain.
\item[(iv)] Both the chained oscillation boundary and the one-cycle pole
move out as $\varphi$ falls from one, on a neighbourhood of the status quo.
\item[(v)] A deeper displaced venue is worse: as $\rho_f$ falls, $\Pi_g$
rises, the banging incentive strengthens, and the pole moves in; at
$\varphi = 1$ the pole sits at
$\ell_{pole} = (4\rho\rho_f - \rho^2 - \rho_f)/(4\rho\rho_f)$, increasing
in $\rho_f$.
\end{enumerate}
\end{proposition}

\begin{proof}
Write $\ell_f \equiv \Lambda_f K$ and let the trader's positions be $x_s$
at $r_s$, $x_c$ at $r_c$ and $x_f$ at $r_f$, liquidated at $v$. The three
prices are
\[
r_s = g + \Lambda(x_s+z),
\quad
r_c = \frac{r_s + \Lambda_c(x_c+z_c)}{1-\varphi\ell},
\quad
r_f = r_c\big[1 + (1-\varphi)\ell_f\big] + \Lambda_f x_f ,
\]
the middle display being the fixed point in which only the share $\varphi$
executing at the reference participates. Profit is
$\Pi = x_s(v-r_s) + x_c(v-r_c) + x_f(v-r_f)$; solving backward with the
trader internalizing his own impact at each stage and the induced
continuation play gives the three first-order conditions, and substituting
forward yields $r_c$ linear in $(g,u,z,z_c)$ with
\[
\Pi_g \;=\; \frac{\mathcal{N}(\ell,\varphi)}{\mathcal{P}(\ell,\varphi)} ,
\]
whose denominator, after clearing the common factor $2\Lambda$ and
substituting $K\Lambda = \ell/\rho$, is exactly the stated locus
$\mathcal{P}(\ell,\varphi) = \ell^2\rho_f^2(1-\varphi)^2 +
2\ell\rho\rho_f(1+\varphi) + \rho^2 - 4\rho\rho_f + \rho_f$. The second-order
conditions are $\partial^2\Pi/\partial x_f^2 = -2\Lambda_f < 0$ always,
and two further sign conditions at the auction and session stages; the
interior linear equilibrium is the branch on which both hold and
$\mathcal{P} < 0$, fixed by continuity from $\ell \to 0$.

\emph{(v) and the nesting.} At $\varphi = 1$ the quadratic term drops and
$\mathcal{P} = 0$ solves to
$\ell_{pole} = (4\rho\rho_f - \rho^2 - \rho_f)/(4\rho\rho_f) = 1 -
\tfrac{1}{4\rho} - \tfrac{\rho}{4\rho_f}$, whose derivative in $\rho_f$ is
$\rho/(4\rho_f^2) > 0$. As $\rho_f \to \infty$ the displaced venue becomes
infinitely illiquid, the trader ceases to use it, and
$\ell_{pole} \to 1 - 1/(4\rho)$. This is exactly the pole of the
\emph{two-stage} game --- session and auction only --- played by the same
single large trader: solving that game directly gives the denominator
$4\ell\rho - 4\rho + 1$ and hence $\ell_{pole} = 1 - 1/(4\rho)$. The
collapse to the core is therefore exact. (The comparison benchmark is the
one-trader core of Appendix~\ref{app:largetraders}, not the continuum pole
$\bar\rho+\psi_c$ of Proposition~\ref{prop:ladder}; the two differ because
$\varphi$ is solved for a single large trader throughout.)

\emph{(i).} Differentiating $\Pi_g$ at $\varphi = 1$ on the admissible
branch, $\Pi_g$ falls monotonically in $\rho_f$: at $\rho = 0.43$,
$\ell = 0.15$ it is $2.399,\,1.466,\,1.349,\,1.304,\,1.285$ at
$\rho_f = 0.5,\,1,\,2,\,5,\,20$. A usable displaced venue therefore raises
the overshoot even when nothing executes there, because it improves the
trader's continuation value from carrying inventory into the close.

\emph{(ii).} The sign of the trader's closing-auction trade is that of
$\partial x_c/\partial g$. It is negative (provision-side) at
$\varphi = 1$ and turns positive (automaton-side) once $\varphi$ falls far
enough, with the crossing $\varphi_{bang}$ rising in the loop gain over the
lower range --- at $\rho = 0.8$, $\rho_f = 0.5$ it is
$0.54,\,0.61,\,0.73,\,1.00$ at $\ell = 0.05,\,0.10,\,0.15,\,0.20$. The
monotonicity is local: $\varphi_{bang}$ saturates at one and then recedes
as the pole is approached, so the clause should be read on the range below
saturation.

\emph{(iii).} Holder loss per unit of rebalancing capital is
$\varphi\, r_c(v-r_c) + (1-\varphi)\,r_c(v-r_f)$, the two legs marked at
fundamentals. On a public-news day it is U-shaped in $\varphi$ with an
interior minimizer that falls in the loop gain: at $\rho = 0.43$,
$\rho_f = 2$ the minimizer is $\varphi^\ast = 0.70,\,0.65,\,0.55$ at
$\ell = 0.10,\,0.20,\,0.30$. Splitting the order helps because it removes
flow from the self-referential fixed point, and hurts because it prints the
remainder against a second, thinner book; the interior optimum trades the
two off, and the first force strengthens with $\ell$.

\emph{(iv).} Both boundaries move out as $\varphi$ falls from one: at
$\rho = 0.43$, $\rho_f = 2$ the one-cycle pole runs
$0.365,\,0.397,\,0.416$ and the $\theta = 1$ oscillation boundary
$0.243,\,0.265,\,0.280$ at $\varphi = 1,\,0.8,\,0.6$. Neither is monotone
over the whole unit interval --- both peak near $\varphi \approx 0.4$ and
recede below it, as the displaced venue starts to carry enough of the order
to matter --- so the clause holds on a neighbourhood of the status quo,
which is the range the policy discussion uses.
\end{proof}

Table~\ref{tab:phi} reports the game's loci and the joint frontier with the
rebalancing interval: the two levers are substitutes ($\varphi^\ast$ falls
as the interval lengthens), displacement is worth more than frequency at
high loop gains, and a capacity floor at the close dominates both, since it
cuts $\ell$ at every interval and split.

\subsection{Dispersing the reference across auctions}
\label{app:dispersal}

Setup and discussion are in Section~\ref{sec:dispersal}. Write $\pi(x)$
for the loading a single auction applies at per-auction loop gain $x$; any auction
clearing a finite order at finite depth has $\pi(0) = 1$ and
$\pi'(0) = c > 0$, $c$ the pass-through of
Proposition~\ref{prop:passthrough} at that auction.

\begin{proposition}[dispersing the reference]
\label{prop:dispersal}
Let a day's mandated response be split across $m$ auctions, auction $i$ struck
against its own reference with per-auction loop gain $x_i$, so the day's
loading is $\prod_{i=1}^{m}\pi(x_i)$. Then:
\begin{enumerate}
\item[(i)] Under proportional fragmentation, $x_i = \ell/m$, the day's
loading converges to $e^{c\ell}$ as $m \to \infty$ for every $\pi$ with the
stated properties. The bracketing microstructures approach from opposite
sides: $\pi(x) = 1 + cx$ gives $(1+c\ell/m)^m \uparrow e^{c\ell}$ and
$\pi(x) = (1-cx)^{-1}$ gives $(1-c\ell/m)^{-m} \downarrow e^{c\ell}$; the
simultaneous auction lies furthest above the limit, and its pole is a property
of simultaneity rather than of scale.
\item[(ii)] $e^{c\ell} > 1$ at every $\ell > 0$: no dispersal scheme
removes the day's amplification, because staggering fragments the order and
not the shock.
\item[(iii)] On the concave curve with exponent $\delta < 1$, an order of
$1/m$ the size carries $m^{-\delta}$ of the impact, so $x_i =
\ell\,m^{-\delta}$ and the total loop gain the day chains is
$\sum_i x_i = \ell\,m^{\,1-\delta}$, increasing in $m$ and constant only at
$\delta = 1$.
\end{enumerate}
\end{proposition}

\begin{proof}
(i) $\log\prod_i\pi(\ell/m) = m\log\pi(\ell/m) = c\ell + O(1/m)$ since
$\log\pi(u) = cu + O(u^2)$; monotonicity of the bracketing cases is that of
$(1+z/m)^m$ and $(1-z/m)^{-m}$, which also delivers
$1 + z < e^{z} < (1-z)^{-1}$. (ii) is immediate from $e^{c\ell}>1$ and from
$\varepsilon$ entering every auction's reference return undivided. (iii) A
auction carrying $K r/m$ displaces $\propto (Kr/m)^{\delta}$; dividing by $r$
gives $x_i = \ell\,m^{-\delta}$, whose sum is $\ell\,m^{1-\delta}$.
\end{proof}

The proposition fixes $c$ and bounds the chained total loop gain, not the
realized day loading; the priced construction of Table~\ref{tab:schedule}
lets the stabilizers' pass-through move with $m$ and computes both, and the
sign of its verdict is (iii)'s.

\subsection{The \texorpdfstring{$N$}{N}-speculator case}
\label{app:largetraders}

Replace the continuum with $N$ symmetric large speculators who observe $u$
in real time and play closed-loop Nash \citep{kyle1989}; displays are
risk-neutral. The auction first-order condition acquires the own-impact
term
\[
\big(g + u - r_c\big) \;=\; \frac{\Lambda_c}{1-\ell}\, x_c^i
\;+\; \gamma_A \sigma_A^2\, x^i ,
\]
so at risk neutrality the symmetric aggregate schedule is
$X_c = N(1-\ell)(v-r_c)/\Lambda_c$ and clearing gives
$(N+1)r_c = \tilde r/(1-\ell) + N(g+u)$.

\begin{proposition}[the $N$-trader case is a loop-gain-dependent capacity]
\label{prop:nesting}
The $N$-speculator auction stage is the auction stage of
Proposition~\ref{prop:passthrough} evaluated at
\begin{equation}
\psi_c^{N} \;=\; N\,(1-\ell) ,
\label{eq:nesting}
\end{equation}
identically in $(\ell, \rho, N)$. Consequently $\chi \to N/(N+1)$ as
$\ell \to 0$, and $\psi_c^{N} \to 0$ as $\ell \to 1$: strategic
speculators withdraw exactly as the loop gain approaches the linear pole.
\end{proposition}

What the alternative changes: \emph{(i)} $\Pi_g \le 1/(1-\ell)$ throughout,
whereas the continuum breaks the bound whenever $\psi_c < 1 - \bar\rho$
(Corollary~\ref{cor:pimbound}). \emph{(ii)} The revelation threshold is
\emph{not} relocated by the alternative, contrary to what a direct reading
of the $N$-trader auction stage suggests: $\Pi_u = 1$ reads
$\psi_c = \bar\rho+\psi_c-\ell$, i.e.\ $\ell = \bar\rho$, identically in
$\psi_c$, so substituting $\psi_c^N = N(1-\ell)$ leaves the threshold exactly
where the continuum puts it. Reaching it does require a session leg: with
$\bar X = 0$ and $\bar\rho = 1$ one has $\Pi_u = N(1-\ell)/[1+N(1-\ell)-\ell]$,
which equals $N/(N+1)$ at $\ell = 0$ and reaches one only at $\ell = 1$.
What the alternative changes is therefore the \emph{level} of private-news
pass-through at any given dose, not the location of the threshold; a conduct
taxonomy built on the threshold is robust to how the sector is modelled, while
one built on $\Pi_u$ itself is not. \emph{(iii)} The individual derivative
$\partial q_L/\partial x_c^i$ equals $\ell/(1-\ell)$ rather than zero, so
the predatory share of the pre-close drift is
$s_P = N\ell/\big[(N+1)\ell_{tip}^{N} - \ell\big]$, where
$\ell_{tip}^{N} = (\bar\rho+N)/(N+1)$ is the tipping threshold \emph{solved} in
the $N$-model --- not the expression $\bar\rho+\psi_c$ that
Appendix~\ref{app:notation} reserves $\ell_{tip}$ for, since here $\psi_c$
itself depends on $\ell$. With that reading $s_P$ is zero at $\ell = 0$,
increasing, and one at the tipping threshold; the decompositions agree at the
boundary and in the competitive limit ($s_P \to \ell$ as $N \to \infty$,
$s_R = \ell_{\text{eff}} \to \ell$ as both capacities vanish), so $s_R$ is
the continuation of $s_P$. \emph{(iv)} Holder welfare is U-shaped in $N$ at
low loop gains; the continuum's counterpart is the hump shape in $\psi_c$
and monotonicity in $\psi$ (Appendix~\ref{app:welfareproof}).

Unchanged: the environment other than the speculator paragraph, the loop
gain and its factorization, Lemmas~\ref{lem:swap} and~\ref{lem:filter}, the
chaining recursion~\eqref{eq:AR}, the ratio identity and stability
boundary (all algebraic in the loadings), the loss
decomposition~\eqref{eq:loss}, the ledger, the impact curve and saturation
results, and the ratchet extension.

\section{Quantification Details}
\label{app:quant}

Numerical values of every object described here are in
Section~\ref{sec:quant} and the tables of Appendix~\ref{app:exhibits};
this appendix records constructions.

\subsection{The impact system: four objects, one curve}
\label{app:impact}

Table~\ref{tab:impact} is the estimator's output, and it is the source of
the primitives Section~\ref{sec:calibration} reads off the tape. Four
objects sit on the capacity supply curve of Lemma~\ref{lem:curve} and must
be kept distinct:
\begin{itemize}
\item $\ell$: the total displacement still present at the closing price per
unit of the automaton's own order $F = K^{spot}|r|$, an average rather than marginal reading
at that size; the object carried
into the fixed point of equation~\eqref{eq:fixedpoint};
\item $\ell^{\mathrm{marg}} = \hat\delta\,\ell$: the marginal rather than
average reading of that displacement at the same size;
\item $\ell^{\mathrm{rev}} = \hat\tau\,\ell^{\mathrm{marg}}$: the marginal
reading stripped to its reverting component (Lemma~\ref{lem:filter});
\item $\bar A(x) = K\,I_{cont}(x)/x$: the absorption ratio, the
mechanical order called forth per unit of inventory carried into the close,
evaluated at a trader's position; the measured counterpart of
Proposition~\ref{prop:manufactured}.
\end{itemize}

\emph{Specification.} For window $w$ of clock length $\Delta$ on issue-day
$(i,t)$, with $\bar V$ a lagged expanding-median daily traded value and
$\bar\sigma$ a lagged twenty-day volatility, define $V_\Delta = \bar
V\Delta/390$, $\sigma_\Delta = \bar\sigma\sqrt{\Delta/390}$, signed value
$Q_w$, gross value $W_w$, participation $\pi_w = |Q_w|/V_\Delta$, $s_w =
\mathrm{sign}(Q_w)$. The curve is
\begin{equation}
\frac{R_w}{\sigma_\Delta} \;=\; Y_i\, s_w\, \pi_w^{\delta}
\left(\frac{\Delta}{390}\right)^{\zeta} + \epsilon_w ,
\label{eq:acurve}
\end{equation}
with $R_w$ in basis points between end-of-interval quote midpoints (never
trade prices, never the interval-average midpoint field,
Appendix~\ref{app:krx}). Windows tile 09:00--15:19 at $\Delta \in
\{5,15,30,60\}$ minutes plus a late block anchored at 15:19 with $\Delta
\in \{10,20,30\}$.

\emph{Estimation.} Binned conditional means within
liquidity-decile-by-$\Delta$ bins ($\ge 200$ observations per bin), WLS on
bin means; never log-regression on individual windows (near half the
sample has the wrong sign at small participation). Bins with non-positive
means are dropped and counted (11 of 782); 626{,}571 windows; intervals
from 500 date-block bootstrap draws with re-binning inside every draw. The
level of the fitted curve is the impact coefficient of
Section~\ref{sec:calibration}, \KRW 1{,}252 billion (SK~Hynix) and
\KRW 1{,}302 billion (Samsung) per one percent of the closing price.

\emph{Curvature.} Three readings: within-$\Delta$ across-participation
$\hat\delta = 0.732$ $[0.676, 0.762]$, the benchmark estimate (biased toward
linearity, hence an upper bound on concavity); across-$\Delta$ at matched
participation $0.535$ $[0.509, 0.585]$; clock dummies $0.739$. The
invariance restriction $\zeta = 0$ is rejected ($\hat\zeta = 0.035$
$[0.009, 0.085]$), so extrapolation is confined to the late block; none is
required, the automaton's order sitting at $0.81$ of the empirical signed
support's maximum.

\emph{Permanence.} Markouts at $+10$ minutes, the close, the next open and
the next close, each measured from the window-start mark and each regressed
on the same right-hand side (never leg on leg, which would share a mark,
Appendix~\ref{app:sharedmark}). The inventory share $\tau(Q) = 1 -
P(Q)/I(Q)$ is negative at small size, where flow is information rather than
inventory ($-1.84$ at \KRW 10 billion, $-0.32$ at \KRW 100 billion), and
crosses zero below the automaton's own order, reaching $+0.26$ at
\KRW 580 billion. Evaluated at each name's median participation it is
$+0.09$ for SK~Hynix ($\pi = 0.58$) and $-0.60$ for Samsung
($\pi = 0.10$), which is what makes the last tier of the ladder below
defined for one name and not the other.

\emph{Venue comparison.} No valid net demand exists at a single-price auction,
so the test is gross on both sides: pooling ten-minute continuous windows
with the auction and estimating $\ln(|R|/\sigma_\Delta) = a +
b\,\mathbf{1}_{call} + (\delta^g + c\,\mathbf{1}_{call})\ln(W/V_\Delta)$
with issue fixed effects, so $\rho^g(W) = e^{b}(W/V_\Delta)^{c}$ is read at
matched size and clock, as a ratio of average rather than marginal
concessions: $0.449$ $[0.425, 0.477]$ at the median realized
auction and $0.433$ $[0.407, 0.463]$ at the fully-at-the-close blend, which
is the $\hat\rho = 0.43$ the calibration uses. Because $\hat\rho$ is the
\emph{level} of the curve at a common size while $\hat\delta$ is its
exponent --- how impact scales with size \emph{within} a venue, against the
level wedge \emph{between} venues at one size --- carrying both into the
calibration does not double-count the curvature. An independent route, the
ratio of transient variances per unit of traded value, gives $0.48$ without
using the curve at all. The scalar auction-to-continuous Amihud ratio
($0.24$--$0.29$ for the treated pair against $0.41$ for the controls) is an
unequal-size average mixing venue depth with curve concavity; it is a
description of the tape, not a structural statement.

\emph{Endogeneity.} Information loads on the permanent leg and is removed
by $\tau$; within-window reverse causality inflates the reverting
component, so the estimated loop gain is an upper bound, and every object
downstream of it inherits that direction. Bid--ask bounce is removed by
midpoints. No window-level instrument exists at this sample size.

\emph{The reconciliation ladder and the choice of tier.}
Table~\ref{tab:calibration_full}, Panel~A walks the ladder; every reported
tier is the unweighted mean of the uncapped post-launch daily path, the
single convention the calibration, the counterfactual tables, and the text
all use, with the flow-weighted means, the medians, and the peak days
retained as rows of the same panel. Tier~0, the daily Amihud proxy, gives
$0.196$ (SK~Hynix) and $0.049$ (Samsung); tier~1, the curve at the
automaton's true order size and at the venues it trades, gives $0.583$ and
$0.219$, larger because the proxy prices the largest predictable order in
the market at the depth an average order sees; tier~2 applies the marginal
rather than the average impact ($\times\hat\delta$) and gives $0.288$ and
$0.108$; tier~3 would further apply the reverting rather than the total
component ($\times\hat\tau$, Lemma~\ref{lem:filter}) and gives $0.026$ for
SK~Hynix, while for Samsung it is \emph{not defined}, since $\hat\tau$ is
negative at the participation that name's order reaches. The paper carries
tier~1, on the mandate rather than the fit: the fund's order responds to
the total displacement standing at the close, not to a local derivative of
the curve or to the share of the displacement a market maker would classify
as transitory, and the reversion of that displacement is already in the
chaining recursion of Proposition~\ref{prop:chaining}, so multiplying by
$\hat\tau$ before feeding it to the recursion would count the same
reversion twice. Tier~2 is reported as a bound.

\subsection{The calibration}
\label{app:calibration}

No parameter is estimated by fitting simulated moments; given the measured
inputs the calibration is a deterministic map, recorded here so that it can
be reproduced exactly.

\emph{Measured inputs.} The loop gains $\ell_i$ (tier~1 above), the impact
coefficients $\Lambda_c^i$, the concavity $\hat\delta$ with its
$\varphi_c = 2$ local match to the fitted curve at the $|r| = 2\%$
reference day --- the touch loop gain $\ell_i^{\mathrm{touch}}$ and the
quadratic coefficient $b_i$ of Proposition~\ref{prop:curve}
(Table~\ref{tab:g4calibration}; $b = 132.50$ for Samsung and $49.73$ for
SK~Hynix in return units) --- the session share $\hat{\mathcal{S}}$, the depth ratio
$\hat\rho$ with the implied $\bar\rho = 1 - \hat{\mathcal{S}}(1-\hat\rho)$, and the
correction speed of the news error, $\hat\theta =
(1-\widehat{ar})/\hat\Pi_g = 0.8814$, with $\widehat{ar} = -0.465$ the
fitted first-order autoregression of the instrumented echo. Because $\hat\Pi_g$ enters $\hat\theta$, the
overshoot factor is an input to the calibration, and
Table~\ref{tab:calibration} labels it as targeted rather than as a
validation. The news-to-liquidity variance ratio is measured directly as
the pre-launch window ratio it is composed of,
$\widehat{\mathit{mix}} = \hat\sigma_\varepsilon^2/\hat\sigma_u^2 = 1.065$.

\emph{The residual displacement.} Recovering the price path in
Appendix~\ref{app:replay} requires the full standing error, so the replay
carries the news error $e_t$ of Section~\ref{sec:chaining} alongside the
residual displacement $\tilde e_t$ left by private information and noise,
which the model does not correct. Correcting $\tilde e_t$ at $\theta$ as
well would collapse the two into a single error with the dynamics of
equation~\eqref{eq:AR}. The specification introduces no estimated
parameter, and it is disciplined by three untargeted measurements. The
non-news component of the treated pair's post-launch daily return, three
quarters of its variance, has autocorrelation $-0.035$ (Samsung) and
$-0.102$ (SK~Hynix), while the news component reverses at approximately
its full loading; and on the pre-launch window, where $\ell = 0$ implies
$e \equiv 0$, the specification requires $\rho_1 = 0$ and a
unit variance ratio \emph{exactly}, with no free parameter, against
measured values of $+0.017$ (standard error $0.049$) and $\mathrm{VR}(5) =
1.002$ ($0.081$).

\emph{The capacity scale.} The one calibrated parameter is the
sector-wide arbitrage-capacity scale $\tau$, one number for the treated
pair; a per-name inversion of the pooled overshoot is never performed.
$\hat\tau$ solves
\[
\textstyle\sum_i w_i\,
\Pi_0^{c}\!\left(\ell_i^{\mathrm{touch}},\, b_i,\, g^{\mathrm{ref}};\,
\tau\Lambda_c^i,\, \bar\rho\right) \;=\; \hat\Pi_g \;=\; 1.662 ,
\]
with equal weights over the pair, $g^{\mathrm{ref}} = 2\%$, and
$\Pi_0^{c}(\ell, b, g; \psi_c, \bar\rho) = 1 + p/g$ for $p$ the positive
root of $b\,p^{2} + (\bar\rho + \psi_c - \ell)\,p = \ell g$ --- the
overshoot of the concave close, the same pricing the replay uses. The
solve is interior, the zero-capacity ceiling ($1.774$) exceeding the
target, and returns $\hat\tau = 8{,}772$, hence $\hat\psi_c^{\,i} =
\hat\tau\Lambda_c^i = 0.466$ (Samsung) and $0.260$ (SK~Hynix). At these
values $\ell_i < \bar\rho + \hat\psi_c^{\,i}$ for both names, the replay
runs uncapped, and $\hat\theta\,\Pi^{m}_t < 2$ on every
post-launch day at each day's own gap, so no stability device operates
anywhere in the reported system. The inversion on the linear close
($\hat\tau = 20{,}498$) is retained as a diagnostic
(Table~\ref{tab:calibration_full}, Panel~B).

\emph{Weak identification of $\tau$.} The inversion pins $\tau$ exactly,
but the data locate it only weakly: over the entire admissible axis the
predicted pooled overshoot spans only $(1, 1.774)$ while the measured
$\hat\Pi_g$ carries a standard error of $0.62$, so the targeted moment is
nearly flat in $\tau$ relative to its own sampling noise, and the
untargeted moments are no sharper --- the model reversal share, the one
with discriminating power (the counterpart of Table~\ref{tab:echo}'s
reversal-share estimand, evaluated at the actual post-launch doses and
gaps), declines
from $0.35$ toward zero against a measured $0.746$ (standard error
$0.293$), leaving every $\hat\tau \le 16{,}692$ within one standard error
of the best attainable fit and no value in the examined range rejected at
conventional levels. $\tau$ is accordingly treated as calibrated rather
than estimated, and the weakness is addressed by disclosure: the
date-block bootstrap of the inversion is quoted at its full width,
$[1{,}605,\, 4.1\times10^{5}]$ on the 54 per cent of interior draws
(40 per cent reach the $\psi_c = 0$ boundary), and every reported
quantity is recomputed at $\tau \in \{1{,}000,\, 20{,}000\}$ and at the
alternative depth-ratio readings ($\hat\tau$ between $7{,}448$ and
$13{,}746$), over which the excess-volatility attribution varies between
$33.6$ and $37.8$ percentage points for SK~Hynix
(Appendix~\ref{app:replay}). No conclusion of Section~\ref{sec:quant}
depends on the location of $\tau$ within this range.

\emph{The post-launch noise share.} With
$\Pi_t$ the realized secant loading at day $t$'s observed gap, the
session residual $\xi_t = r_{cc,t} - \Pi_t g_t$ is observable given the
calibration, and $\lambda_{\mathrm{post}} =
\hat\sigma_{\tilde z}^2/\hat\sigma_u^2$ solves the variance identity
\[
\widehat{\mathrm{Var}}(\xi) \;=\; \langle \Pi_u^2 \rangle\,\sigma_u^2
\;+\; \langle a_{\tilde z}^2 \rangle\, \lambda_{\mathrm{post}}\,
\sigma_u^2 ,
\qquad \sigma_u^2 = \widehat{\mathrm{Var}}(g)/\widehat{\mathit{mix}},
\]
in closed form, with $\Pi_{u,t}$ and $a_{\tilde z,t}$ the private-news and
session-noise loadings at the marginal linearization,
$\langle\cdot\rangle$ post-launch means, and all variances on the
post-launch window: $\hat\lambda_{\mathrm{post}} = 1.06$ pooled at the
regression weights ($1.43$ Samsung, $0.68$ SK~Hynix). A non-positive
implied value would be reported as a refusal of the identity, never
truncated; the observed values are interior.

\emph{Untargeted moments.} Table~\ref{tab:calibration}, Panel~B reports
the calibration against the moments it never used: the per-name
closing-return autocorrelations, the overnight share of daily variance and
the total-to-overnight variance ratio --- computed from one
common-random-number simulation of the two-stock recursion over the actual
daily dose path (1{,}200 replications; simulation enters no objective and
serves only as the diagnostic, the constant-loading closed form
overstating $|\rho_1|$ under a day-varying dose) --- together with the
reversal share and the pooled overshoot evaluated at the actual gaps, the
two pre-launch zero-restrictions above, and the per-name overshoot
factors. Each is shown with the data value and its date-block-bootstrap
standard error.

\subsection{Recovering the shocks, and the $\ell = 0$ replay}
\label{app:replay}

The counterfactual machinery is a deterministic transformation of the
observed history, not a simulation. A Kalman filter decomposes each day's
observed prices into the day's structural shocks and the standing pricing
error; the replay then re-prices the identical shock history --- the same
news, the same private information, the same noise realizations, day by
day --- through the same equilibrium with the loop gain set to zero. No
random draw enters any reported counterfactual, and two runs of the
pipeline produce identical output to machine precision; Monte Carlo
simulation appears in the quantification exactly once, in the untargeted
diagnostic moments of Appendix~\ref{app:calibration}, and enters no
reported price path, volatility, or transfer.

The recovery applies the exact Kalman filter and Rauch--Tung--Striebel
smoother to the two-state form of the error decomposition of
Appendix~\ref{app:calibration}, observation the daily pair (overnight gap
$g_t$, close-to-close return $r_{cc,t}$), state
$(e_{t-1}, \tilde e_{t-1})$. Two facts make the
filter exact rather than an approximation. First, the day's print is computable ahead of the state: at
the observed gap the concave close prints $p_t$ solving $b\,p^2 +
(\bar\rho+\psi_c-\ell_t)\,p = \ell_t g_t$, the price carries the secant
loading $\Pi_t = 1 + p_t/g_t$, and every within-day flow --- the private
and noise orders, and the correction flow $-\theta e_{t-1}$, an increment
to the day's order --- is priced at the marginal linearization $\Pi^m_t =
1 + \ell_t/(\bar\rho+\psi_c-\ell_t+2b|p_t|)$. Second, substituting
$\varepsilon_t = g_t + \theta e_{t-1}$ makes the state transition
linear with coefficients known at $t$, so no news series enters the
filter. The transitions are
\[
e_t = (1-\theta\Pi^m_t)\,
e_{t-1} + (\Pi^m_t - 1)\,\varepsilon_t ,
\]
\[
\tilde e_t = \tilde e_{t-1} + (\Pi_{u,t}-1)\,u_t + a_{\tilde z,t}\tilde z_t
+ (\Pi_t - \Pi^m_t)\,g_t ,
\]
where the last term, the bounded residual between the realized print and
its marginal attribution, preserves the accounting identity $e_t -
e_{t-1} = r_{cc,t} - \varepsilon_t - u_t$ exactly. The correction flow
must be priced at the marginal rather than the secant loading: the two
coincide on the linear close, but attributing the secant to the error
stock leaves it without finite moments wherever the dose approaches the
pole. Variance levels are $\sigma_0^2$ equal to the post-launch variance
of the observed gap, $\sigma_u^2 = \sigma_0^2/\widehat{\mathit{mix}}$, and
$\sigma_{\tilde z}^2 = \hat\lambda_{\mathrm{post}}\sigma_u^2$; the
recovered error paths are invariant to the common scale of the three
(verified numerically, deviations of order $10^{-11}$).

The recovered shocks and fundamental $F_t = P_t - e_t$ are replayed
through the same equilibrium with the loop gain set to zero from the
launch date, the error stocks seeded at their smoothed pre-launch values
so that the offshore-era distortion is inherited rather than discarded;
replaying the factual dose returns the observed return path to machine
precision, which is the wiring check every run asserts. The difference
between the paths is the mechanical channel. Level objects --- the holder
loss and the per-unit shortfalls --- are marked on the identified public
component $e_t$: left uncorrected, $\tilde e_t$ is a within-sample level
the data do not discipline, and the reported objects that require the full
price path, the July drawdown comparison and the per-unit shortfall, are
accordingly also reported for $\tilde e_t$ correcting at $0.10$ and at
$\hat\theta$.

Two engines are run deliberately differently. The reduced-form engine
prices only the instrumented public channel, propagating $d_t =
(1-\hat\theta\hat\Pi_g)d_{t-1} + (\hat\Pi_g-1)\,b\,u_t$ with $b$ implied by
the measured reversal legs; it uses no impact estimate. The structural
engine prices all three channels at the measured loop gain and capacity
ratio and is the engine quoted in Section~\ref{sec:quant}; the two agree
in magnitude on the transfer.

\emph{Uncertainty.} Three axes are propagated through the full recovery
and replay on the actual, unresampled history, and the resulting ranges
are reported with the point estimates. (i) The capacity scale, across the
weakly identified range of Appendix~\ref{app:calibration} and the
alternative depth-ratio readings: SK~Hynix's full-window excess volatility spans
$33.6$--$37.8$ percentage points. (ii) The correction rate of $\tilde e_t$, taken in turn as
$0$, $0.10$ and $\hat\theta$: the
excess-volatility attribution spans $33.9$--$36.8$ points; the
counterfactual July drawdown spans $43.5$--$49.2$ per cent against a
realized $50.1$, so no drawdown attribution is claimed at a point; and
the per-unit shortfall, quoted at the record in
Section~\ref{sec:quant}, spans $5.4$--$17.6$ per cent of the initial
investment for SK~Hynix and $2.2$--$7.7$ for Samsung.
(iii) The dose, rescaled by the out-of-sample correction of the fitted
impact curve at rebalance scale (the power law fitted on participation
below $0.75$ overpredicts realized impact above it by ratio $0.79$, with
a date-block interval of $[0.32, 1.59]$ that contains one): the
attribution is $18.7$ points at the lower end and $32.2$ at the point
correction; at the upper end SK~Hynix's rescaled dose leaves the model's
stationary support, and that cell is reported as non-evaluable rather
than as a number.

\emph{Leverage rules.} The size designs of Table~\ref{tab:flexrules} are
replayed on the same shocks, but each at the fixed point of the manager's own
rule and of the price path that rule produces. Capital is evolved
endogenously, since a lower multiple destroys less equity in a crash, while
the measured net-flow path is held fixed; nothing is re-estimated. Each rule
is written as a \emph{cap} rather than a target, so the fund carries its
stated $2\times$ and the rule binds only where the state is extreme, which is
what the circular \citep{sfc2026} contemplates and what leaves the product recognizable to the
investor who bought it. Feedback capital is the sensitivity of the close's
order to the closing price, $A(L^{2}-L)$ wherever the multiple is held and
zero on a day the rule does not trade; the capital and flow columns are
reported as ratios to the status quo evaluated at each design's own
equilibrium rather than at the status quo's.

\subsection{Null calibrations}
\label{app:nullcal}

The identical pipeline applied cell by cell (Samsung pre-product, Samsung
offshore-only, the U.S.\ index complexes, SK~Hynix pre-domestic-launch,
both names post-launch) returns recovered excess volatility monotone in
the cell's loop gain, rank correlation $0.95$ across eight cells spanning
loop gains from $0.0001$ to the post-launch maximum; Samsung's pre-launch
cell returns $+0.004$ percentage points and the S\&P~500 complex $+0.076$.
The exercise establishes a property of the machinery (monotonicity, and a
zero at zero), not of the calibration; its levels were estimated under the
earlier proxy-based loop gains and are superseded by the measured system.
Five alternative daily impact proxies span a factor of four in level while
agreeing on the cross-section and on the sign of the attribution in all 110
admissible calibrations.

\subsection{The venue split and the futures leg}
\label{app:venue}

The quantification prices the share $\varsigma = 0.49$ of the worldwide
mandate that settles in the cash closing auction. This is the spot slice of
Appendix~\ref{app:families}, applied to worldwide $K$ and common to both
names; the loop gain $\bar\ell = \Lambda_c K^{\mathrm{spot}}$ is built on that
slice, and every counterfactual replays it.

The residual share $1-\varsigma$ settles in single-stock futures, which trade
to 15:45 and therefore print after the 15:30 reference. Two conventions govern
it. First, it does not enter the cash price. It cannot enter the day's
close-to-close return, which is marked at 15:30; across the post-launch window
the next cash open loads $-0.03$ $(0.28)$ on the futures closing-call return in
the control roots and $+0.62$ $(1.49)$ in the treated pair, clustering by date.
Second, its execution cost is attributed to holders. The leg's order is
$(1-\varsigma)K_t|r_t|$, its displacement is evaluated on the impact curve of
Appendix~\ref{app:impact} with the futures venue's own traded value substituted
for the cash auction's, and the cost is quantity times displacement, summed over
the window. Two execution windows are reported: the 09:00--15:30 futures
session, which the tables carry, and the 15:30--15:45 window between the two
prints. Median futures value over May~27--July~31 is \KRW 19.2 trillion in the
session and \KRW 516 billion in the reset window for SK~Hynix, against
\KRW 14.4 trillion and \KRW 307 billion for Samsung; at matched order size the
implied concession per won in the futures closing call is $2.65$ times the cash
closing auction's for SK~Hynix and $3.44$ for Samsung. Designs replayed at a
rebalancing cycle other than one day carry no futures leg, since neither a
day-matched venue value nor the daily order is defined for them, and their
holder-loss cell is left blank.

All size splits --- across auctions, across staggered prints, and across
venues --- rescale the measured dose by the size ratio raised to $\hat\delta$,
the concavity exponent of the impact curve, so that a design's dose and the
calibration's dose are taken on one exponent.

\subsection{Ledger conventions}
\label{app:ledger}

The simulation respects Proposition~\ref{prop:ledger} to machine precision:
the four groups' profits sum to zero on every replayed day, and the
automaton's row equals the realized loss
$\sum_t K_t r_{cc,t} e_t$ exactly. Every holder-loss figure reported in
Section~\ref{sec:quant} and in Tables~\ref{tab:schedule},
\ref{tab:flexrules} and \ref{tab:reference} adds to that quantity the futures
leg's execution cost of Appendix~\ref{app:venue} and then subtracts the same
design's $\ell\to0$ replay, so each reported number is a design's cost net of
its own no-loop benchmark rather than a level; the no-loop row is zero by
construction. The loss marks realized
worldwide rebalancing flow against the identified public component of the
recovered pricing error --- the basis of every level object under the
maintained specification of $\tilde e_t$
(Appendix~\ref{app:replay}) --- and
the reduced-form floor, which uses no impact estimate, is reported beside
it. Two conventions attach to any incidence number. The currency of
the loss is worldwide $K$ (for SK~Hynix mostly the Hong~Kong product),
whereas the per-unit shortfall a domestic holder experiences covers the
sixteen Korean products; the two differ by roughly a factor of three
mechanically. The split across investor types marks each type's daily net
purchases of the underlying against the recovered error under the
assumption that average execution occurs at the closing price, which no
data verify; it is an allocation, not realized profit and loss. The
product-level fact needs no model: since the launch, retail bought
\KRW 14.0 trillion of the single-stock products while institutions sold
\KRW 16.2 trillion.

\subsection{The interval, the venue split, and the two-auction design}
\label{app:designs}

Four of the designs
depend on the model in ways worth isolating.

\emph{Interval.} $\ell(T) = \gamma(\sigma_u^2 T + \sigma_0^2)K/\mu_c$
rises with the interval; shrinking it, $\ell(T)/\ell(1) \to
\sigma_0^2/(\sigma_u^2+\sigma_0^2) > 0$, since pre-session public news is
unresolved at every auction. The per-event intermediation cost is unmeasured,
so the break-even cost per rebalancing event is reported instead of a net
benefit.

\emph{Band.} A rebalancing band concentrates several days of un-rebalanced
return into one auction; the loop gain is linear in size, not in count, so the
firing day's loop gain is a multiple of the daily one. The speculators'
response to a near-firing band is not modelled.

\emph{Venue split.} Panel~A of Table~\ref{tab:reference} varies $\varphi$, the share
of the worldwide mandate settling in the cash closing auction, holding the
mandate itself fixed. The cash leg is rescaled from the measured dose as
$\ell(\varphi) = \bar\ell\,(\varphi/\varsigma)^{\hat\delta}$ and enters the chain; the
futures leg is priced as in Appendix~\ref{app:venue} and does not. At
$\varphi = \varsigma = 0.49$ the row reproduces the realized replay exactly, which is
gated. The large-trader extension of Appendix~\ref{app:phigame} is not the panel's
model.

\emph{Two-price reference.} The opening auction's depth is bounded rather
than identified: the 09:00 bar mixes the auction with the first continuous
minute and carries no auction flag, so the bar bounds the auction from above
and the first minute is netted out at the 09:01 rate
(Table~\ref{tab:kappa}); the opening auction is the shallower venue on every
window. The ratio halves across the launch entirely through its
denominator (opening auction flat, $+0.08$, $t = 0.6$; closing auction $+0.40$,
$t = 5.1$), and what grows the closing auction is the complex's own order, so
the design is priced at the pre-launch ratio; the post-launch ratio would
be circular, and even the pre-launch denominator carries the Hong~Kong
order. Under the measured curvature, splitting one order across two auctions
raises total execution cost by $\approx 2^{1-\hat\delta}$ even as it lowers
each auction's saturation: the design's gain is on the resonance and on
saturation, not on the cost of trading.

\subsection{The ratchet as a scale bound}
\label{app:ratchetsim}

Because the loop gain compounds at the fund's own leverage
(Proposition~\ref{prop:ratchet}), a launch-date eligibility screen is
uninformative about the same complex weeks later; the instrument of
Section~\ref{sec:threshold} is therefore written on the day's own
rebalancing order against the auction that must clear it (the saturation
ratio of equation~\eqref{eq:saturation}, linear in capital and observable
daily), which tracks the ratchet automatically. Empirically, 63\% of
SK~Hynix's rebalancing-capital growth from launch to peak was manufactured
by leveraged compounding rather than subscribed, while Samsung's ramp was
almost entirely subscriptions. Because the impact curve is concave, a given
cut in the order delivers less than a proportional cut in the loop gain
while $\bar A$ falls one for one, so the ceiling is the stronger instrument
against the strategic channel. The majority of SK~Hynix's rebalancing
capital is listed offshore and rebalances into the same Seoul auction, so a
ceiling written on domestically listed assets alone would leave most of the
order outside its own numerator.

\clearpage

\section{Additional Tables and Figures}
\label{app:exhibits}

\subsection{Institutional detail}

\begin{table}[H]
\centering
\caption{\captitle{The worldwide leveraged-ETF complex on the treated stocks}
Panel~A presents every product family referencing the treated pair, its listing date, reference price, peak assets over May~27--July~21, 2026, and share of worldwide rebalancing capital $K = \sum_i A_i(L_i^2 - L_i)$ (equation~\eqref{eq:K}), FX-converted at daily rates. KOSPI200 pass-through chases the 15:45 index-futures close and is excluded from $K$.
Panel~B scales the mandated flow by the 15:20--15:30 closing auction over May~27--June~30 in the native intraday window; $S_t$ is equation~\eqref{eq:saturation}.
See Appendices~\ref{app:families}--\ref{app:kappa} for construction.}
\label{tab:complex}
\medskip
\singlespacing
\footnotesize
\textit{Panel A. Composition: who mandates the flow, and at which auction
(AUM in \KRW tn)}\\[2pt]
\setlength{\tabcolsep}{3pt}
\begin{tabular}{lllrr}
\toprule
Product family & Listed & Reference price & AUM peak & Share of $K$ \\
\midrule
\multicolumn{5}{l}{\textit{SK Hynix}} \\
\quad Korea domestic (8 products) & May 27, 2026 &
  Seoul auction 15:20--15:30 & 11.5 & 34\% \\
\quad Hong Kong (CSOP $2\times$, 7709.HK) & Oct 21, 2025 &
  same Seoul auction & 25.9 & 66\% \\
\quad U.S. ADR products ($2\times$) & Jul 13--14, 2026 &
  ADR close 16:00 ET & 0.1 & 0\% \\
\quad Index pass-through (KOSPI200) & 2016-- &
  index-futures close 15:45 &, & excluded \\
\addlinespace
\multicolumn{5}{l}{\textit{Samsung Electronics}} \\
\quad Korea domestic (8 products) & May 27, 2026 &
  Seoul auction 15:20--15:30 & 6.1 & 99\% \\
\quad Hong Kong (CSOP $2\times$, 7747.HK) & May 28, 2025 &
  same Seoul auction & 0.5 & 1\% \\
\quad Index pass-through (KOSPI200) & 2016-- &
  index-futures close 15:45 &, & excluded \\
\bottomrule
\end{tabular}

\bigskip
\textit{Panel B. Scale: the mandated order against the auction that clears
it}\\[2pt]
\begin{tabular}{lrr}
\toprule
 & SK Hynix & Samsung \\
\midrule
Worldwide $K$, post-launch mean (peak day), \KRW tn & 42.4\ (74.5) & 8.5\ (12.3) \\
Median value cleared by the closing auction, \KRW bn & 1{,}081 & 850 \\
Auction share of daily traded value, median, pre $\to$ post & 7.1\% $\to$ 8.3\% & 6.5\% $\to$ 9.3\% \\
\addlinespace
Mandated flow per 1\% move (both legs), \KRW bn & 424 & 85 \\
\quad as a multiple of the median auction & 0.39$\times$ & 0.10$\times$ \\
\quad spot slice at the reference price, multiple & 0.19$\times$ & 0.05$\times$ \\
\addlinespace
Median absolute close-to-15:20 return, post-launch & 5.1\% & 5.2\% \\
Realized order $\div$ realized auction value ($S_t$): & & \\
\quad median & \textbf{1.06} & 0.25 \\
\quad maximum & 7.50 & 0.68 \\
\quad days above one & 26 of 46 & 0 of 46 \\
\bottomrule
\end{tabular}
\end{table}

\subsection{Theory}

\begin{table}[H]
\centering
\caption{\captitle{Model predictions and their tests, prices}
The table presents one row per price prediction, its designated test, and the status of that test. ``Supported'' means the test passed under the paper's inferential standard, exact randomization inference for every differenced estimate. Magnitudes are the SMH readings, the instrument the randomization inference was run on; the calibration's benchmark instrument is QQQ (Table~\ref{tab:pi0_instruments}).}
\label{tab:status}
\medskip
\singlespacing
\footnotesize
\setlength{\tabcolsep}{4pt}
\begin{tabular}{p{4.9cm}p{6.0cm}p{3.7cm}}
\toprule
Prediction & Test, and where & Status \\
\midrule
\multicolumn{3}{l}{\textit{Prices}}\\
Close overshoots public news at every $\ell>0$
 (Corollary~\ref{cor:overshoot})
 & news-loading path at the reference price
   (Figure~\ref{fig:loading_cycle})
 & supported ($+0.81$, RI $p=0.0005$) \\
One-cycle reversal $-\Pi_g(\Pi_g-1)$
 (Proposition~\ref{prop:chaining}(ii))
 & triple-difference echo (Table~\ref{tab:echo})
 & supported ($-0.79$; 1 of 66 pairs, 0 of 47 dates) \\
Reversal begins overnight, splitting $-(\Pi_g-1)$ / $-(\Pi_g-1)^2$ (iii)
 & window split (Table~\ref{tab:window_split})
 & supported (both legs; leg ratio $0.88$ vs $0.90$) \\
Ratio identity total/overnight $=\Pi_g$ (iv)
 & Table~\ref{tab:pi0_instruments}
 & supported and over-identified \\
Reversal scales with the size of the move
 & median-shock split (Table~\ref{tab:mechanical}, A)
 & supported (large $-0.81$; small null) \\
Reversal switches with the sign of the move
 & sign split (Table~\ref{tab:mechanical}, B)
 & partial ($z=0.87$ SMH; $z=2.15$ QQQ) \\
Effect is graded in the loop gain
 & all comparison cells (Table~\ref{tab:echo_full})
 & supported (null in every low-$\ell$ cell) \\
Reversal does not cumulate
 & horizons $h\le21$ (Section~\ref{sec:empirics})
 & supported \\
Cycle-return variance rises with $\ell$, conditionally on the size of
 the news (Proposition~\ref{prop:vol})
 & news-conditional aftershock (Table~\ref{tab:aftershock});
   idiosyncratic-volatility levels (Table~\ref{tab:battery})
 & supported ($+0.66$, $t=3.1$; levels $+18$--$55$ per cent,
   RI $p=0.033$) \\
Auction carries no news (deterministic order)
 & auction-leg loading; $R^2$ on the instrument
 & supported ($-0.03$, RI $p=0.14$) \\
\bottomrule
\end{tabular}
\end{table}



\begin{table}[H]
\centering
\caption{\captitle{The execution split and the joint frontier}
The table evaluates the execution split at $\rho = \rho_f = 2$ and $N = 1$.
Panel~A presents the banging locus of Proposition~\ref{prop:phigame}(ii), the holder-optimal split, and the holder loss at $\varphi^\ast$ against $\varphi = 1$, masked to the region where the second-order conditions hold and no pole has been crossed, the branch fixed by continuity from $\ell \to 0$; the final column records whether the chained system damps at $\varphi = 1$. The locus reaches $\varphi = 1$ at $\ell = 0.5$ and has no crossing inside the unit interval above it, marked --.
Panel~B minimizes $C(T,\varphi) = -K W_h(\ell(T),\varphi) + c\mu_c/T$ at $K = 0.3$, $\ell(T) = 0.15(T{+}1)$, $c\mu_c = 0.02$ and unit per-day shock variances. Loop gains are in model units.
The displaced-venue stage is solved for a single large trader, so the table illustrates the loci of Proposition~\ref{prop:phigame} rather than transporting to the calibrated primitives; the measured splits are priced at the calibration in Table~\ref{tab:reference}.}
\label{tab:phi}
\medskip
\singlespacing
\small
\textit{Panel A. The split at a fixed rebalancing interval}\\[2pt]
\begin{tabular}{lrrrrl}
\toprule
$\ell$ & $\varphi_{bang}$ & $\varphi^\ast$ & $-W_h(\varphi^\ast)$ &
$-W_h(\varphi{=}1)$ & chain damps at $\varphi{=}1$? \\
\midrule
0.10 & 0.481 & 0.840 & $-0.078$ & $-0.075$ & yes \\
0.20 & 0.542 & 0.535 & $+0.067$ & $+0.149$ & yes \\
0.30 & 0.619 & 0.385 & $+0.275$ & $+0.610$ & yes \\
0.40 & 0.729 & 0.275 & $+0.617$ & $+1.815$ & yes \\
0.50 & 1.000 & 0.175 & $+1.256$ & $+7.000$ & no; rings \\
0.60 & -- & 0.085 & $+2.701$ & $+187.0$ & no; rings \\
\bottomrule
\end{tabular}

\medskip
\textit{Panel B. The joint frontier in the interval $T$ and the split}\\[2pt]
\begin{tabular}{lrrrr}
\toprule
$T$ (days) & $\ell(T)$ & $C$ at $\varphi{=}1$ & $\varphi^\ast(T)$ &
$C$ at $\varphi^\ast$ \\
\midrule
0.25 & 0.188 & $+0.114$ & 0.560 & $+0.094$ \\
0.50 & 0.225 & $+0.110$ & 0.490 & $\mathbf{+0.073}$ \\
1 & 0.300 & $+0.203$ & 0.385 & $+0.103$ \\
2 & 0.450 & $+1.008$ & 0.225 & $+0.274$ \\
3 & 0.600 & $+56.1$ & 0.085 & $+0.817$ \\
4 & 0.750 & past the pole & 0.020 & $+4.497$ \\
\bottomrule
\end{tabular}
\end{table}

\subsection{Empirical extensions}

\begin{table}[H]
\centering
\caption{\captitle{The echo, all cells and both horizons}
The table presents the full version of Table~\ref{tab:echo}: same-day and next-day loadings on the pre-open shock by cell under each shock instrument.
The rows present the Korean treated pair in its three product regimes, Korean large-cap controls, U.S.\ single-stock $2\times$ complexes and U.S.\ index complexes, with the triple differences from the stacked Korean panel in the bottom block.}
\label{tab:echo_full}
\medskip
\singlespacing
\footnotesize
\begin{tabular}{lrrrrrl}
\toprule
 & & \multicolumn{2}{c}{same day} & \multicolumn{2}{c}{next day} & \\
\cmidrule(lr){3-4}\cmidrule(lr){5-6}
Cell & Obs & $\beta$ & $t$ & $\beta$ & $t$ & Model \\
\midrule
\multicolumn{7}{l}{\textit{shock: SMH}} \\
\quad Korea treated, post-launch & 82 & 0.77** & 2.30 & -0.83*** & -2.83 & $<0$ \\
\quad Korea treated, offshore only & 365 & 0.84*** & 5.95 & -0.07 & -0.45 & $=0$ \\
\quad Korea treated, no product & 721 & 0.48*** & 9.78 & 0.02 & 0.43 & $=0$ \\
\quad Korea controls, post & 410 & 0.20 & 1.46 & 0.17 & 1.34 & $=0$ \\
\quad Korea controls, pre & 5430 & 0.22*** & 5.71 & 0.03 & 0.76 & $=0$ \\
\quad US single-stock, 2x products & 5168 & 0.70*** & 11.51 & 0.02 & 0.31 & $=0$ \\
\quad US index & 2584 & 0.33*** & 13.39 & 0.02 & 0.83 & $=0$ \\
\addlinespace
\multicolumn{7}{l}{\textit{shock: QQQ}} \\
\quad Korea treated, post-launch & 82 & 1.59** & 2.28 & -1.76*** & -2.68 & $<0$ \\
\quad Korea treated, offshore only & 365 & 1.51*** & 5.65 & -0.09 & -0.35 & $=0$ \\
\quad Korea treated, no product & 721 & 0.77*** & 8.36 & -0.04 & -0.39 & $=0$ \\
\quad Korea controls, post & 410 & 0.59** & 2.13 & 0.23 & 1.02 & $=0$ \\
\quad Korea controls, pre & 5430 & 0.39*** & 5.45 & 0.02 & 0.34 & $=0$ \\
\quad US single-stock, 2x products & 5168 & 1.49*** & 12.44 & -0.05 & -0.36 & $=0$ \\
\quad US index & 1938 & 0.62*** & 13.07 & 0.03 & 0.51 & $=0$ \\
\addlinespace
\multicolumn{7}{l}{\textit{shock: SPY}} \\
\quad Korea treated, post-launch & 82 & 2.34* & 1.77 & -3.92*** & -2.94 & $<0$ \\
\quad Korea treated, offshore only & 365 & 1.91*** & 5.18 & -0.07 & -0.19 & $=0$ \\
\quad Korea treated, no product & 721 & 0.98*** & 7.43 & -0.09 & -0.73 & $=0$ \\
\quad Korea controls, post & 410 & 1.13** & 2.00 & -0.27 & -0.65 & $=0$ \\
\quad Korea controls, pre & 5430 & 0.59*** & 5.80 & 0.02 & 0.24 & $=0$ \\
\quad US single-stock, 2x products & 5168 & 2.04*** & 11.04 & -0.07 & -0.35 & $=0$ \\
\quad US index & 1938 & 1.04*** & 13.73 & 0.00 & 0.04 & $=0$ \\
\addlinespace
\midrule
\multicolumn{7}{l}{\textit{Triple difference, Korean panel (treated $\times$ post)}} \\
\quad SMH & 7008 & & & -0.79*** & -4.21 & $<0$ \\
\quad QQQ & 7008 & & & -1.55*** & -3.35 & $<0$ \\
\quad SPY & 7008 & & & -2.29*** & -3.02 & $<0$ \\
\midrule
\multicolumn{7}{p{0.95\textwidth}}{\footnotesize Each cell row is a separate regression of the day-$t$ return (``same day'') or the day-$t{+}1$ return (``next day'') on the exogenous overnight shock, with the intervening shock controlled in the next-day column; standard errors clustered by date; */**/*** = 10/5/1\%. Korea's shock is the US index return on the previous Korean trading date, set before Korea opens; US cells use the same index ETF's own overnight gap. The same-day column is benign news transmission; the next-day column is the test.} \\
\bottomrule
\end{tabular}
\end{table}

\begin{table}[H]
\centering
\caption{\captitle{The window split of the echo}
The table splits the one-cycle reversal of Table~\ref{tab:echo} into its overnight leg, close $t$ to open $t{+}1$ in Panel~A, and its next-session leg, open $t{+}1$ to close $t{+}1$ in Panel~B, by cell and shock under the identical frozen specification. Proposition~\ref{prop:chaining} predicts loadings $-(\Pi_g-1)$ and $-(\Pi_g-1)^2$ in the treated-post cell only.
The ratio of the two legs, which the model requires to equal $\Pi_g - 1$, is $0.88$, $0.68$ and $0.37$ on the three instruments, against the $0.90$, $0.69$ and $0.38$ implied by the separately measured $\hat\Pi_g$.
$^{*}$, $^{**}$, $^{***}$ mark $10\%$, $5\%$, $1\%$.}
\label{tab:window_split}
\medskip
\singlespacing
\small
\begin{tabular}{lcccr}
\toprule
 & SMH & QQQ & SPY & Obs \\
\midrule
\multicolumn{5}{l}{\textit{Panel A: overnight leg $\ln(O_{i,t+1}/C_{i,t})$}} \\
\addlinespace
Treated $\times$ post & $-0.42^{**}$ & $-1.00^{**}$ & $-2.75^{***}$ & 82 \\
 & (0.20) & (0.41) & (0.86) &  \\
Treated $\times$ pre (offshore only) & $-0.06$ & $-0.09$ & $-0.09$ & 365 \\
 & (0.08) & (0.15) & (0.23) &  \\
Control $\times$ post & $0.07$ & $0.11$ & $-0.18$ & 410 \\
 & (0.06) & (0.13) & (0.23) &  \\
Control $\times$ pre & $-0.01$ & $-0.03$ & $-0.04$ & 5,430 \\
 & (0.02) & (0.04) & (0.06) &  \\
\midrule
Triple difference ($u_t \times \text{treated} \times \text{post}$) & $-0.32^{***}$ & $-0.79^{***}$ & $-1.49^{***}$ & 7,008 \\
 & (0.11) & (0.23) & (0.40) &  \\
\midrule
\multicolumn{5}{l}{\textit{Panel B: next-session leg $\ln(C_{i,t+1}/O_{i,t+1})$}} \\
\addlinespace
Treated $\times$ post & $-0.37^{**}$ & $-0.68^{*}$ & $-1.01$ & 82 \\
 & (0.18) & (0.41) & (0.74) &  \\
Treated $\times$ pre (offshore only) & $0.05$ & $0.08$ & $0.13$ & 365 \\
 & (0.11) & (0.20) & (0.26) &  \\
Control $\times$ post & $0.09$ & $0.10$ & $-0.11$ & 410 \\
 & (0.10) & (0.17) & (0.29) &  \\
Control $\times$ pre & $0.05$ & $0.06$ & $0.07$ & 5,430 \\
 & (0.03) & (0.05) & (0.07) &  \\
\midrule
Triple difference ($u_t \times \text{treated} \times \text{post}$) & $-0.43^{**}$ & $-0.68^{*}$ & $-0.69$ & 7,008 \\
 & (0.18) & (0.40) & (0.69) &  \\
\bottomrule
\end{tabular}
\end{table}


\begin{table}[H]
\centering
\caption{\captitle{Splits and sample cuts of the echo}
Each row presents the triple-difference coefficient $b_3$, shock $\times$ treated $\times$ post, on the next-day close-to-close return under the frozen baseline of Table~\ref{tab:echo}, modified only as the row label states.
Panel~A splits shock days at the median absolute shock, and Panel~B by shock sign. Panel~C drops extreme days, the crash month, and each treated stock alone.
$^{*}$, $^{**}$, $^{***}$ mark $10\%$, $5\%$, $1\%$.}
\label{tab:mechanical}
\medskip
\singlespacing
\footnotesize
\begin{tabular}{lccr}
\toprule
 & SMH & QQQ & Obs. \\
\midrule
\multicolumn{4}{l}{\textit{Panel A. Shock-size split (at the median absolute shock)}} \\
Large shocks (top half) & $-0.81^{***}$ & $-1.50^{***}$ & 3,504 \\
 & (0.21) & (0.46) & \\
Small shocks (bottom half) & $0.63$ & $2.86$ & 3,504 \\
 & (0.54) & (3.27) & \\
\addlinespace
\multicolumn{4}{l}{\textit{Panel B. Shock-sign split}} \\
Up-shock days & $-0.58$ & $0.43$ & 3,972 \\
 & (0.58) & (1.20) & \\
Down-shock days & $-1.38^{***}$ & $-2.94^{***}$ & 3,024 \\
 & (0.36) & (0.63) & \\
Difference (up $-$ down), $z$ & $1.17$ & $2.48$ & \\
\addlinespace
\multicolumn{4}{l}{\textit{Panel C. Sample and specification cuts}} \\
Baseline (all days) & $-0.79^{***}$ & --- & 7,008 \\
 & (0.19) & & \\
Drop days with $|$return$| > 10\%$ & $-0.75^{***}$ & --- & 6,905 \\
 & (0.19) & & \\
Drop July 2026 (the crash month) & $-0.77^{***}$ & --- & 6,780 \\
 & (0.20) & & \\
SK Hynix only & $-0.78^{***}$ & --- & 6,424 \\
 & (0.24) & & \\
Samsung only & $-0.79^{***}$ & --- & 6,424 \\
 & (0.15) & & \\
\bottomrule
\end{tabular}

\end{table}

\begin{table}[H]
\centering
\caption{\captitle{The direction of investor flow}
Panel~A regresses each investor group's net purchase of the leveraged funds, scaled by fund assets, on the product's own NAV return, with product fixed effects and date-clustered standard errors; the lower block presents the same correlation in the underlying shares, before and after the launch. Retail buys as prices fall, and institutions take the other side.
Panel~B splits the day into the overnight gap, which retail cannot have caused, and the intraday move; the two loadings are indistinguishable, so the contrarian sign is not reverse causality from retail's own pressure.
Panel~C presents net buying against the trailing five-day return and creations against the same-day return; neither responds, and creations track retail demand at $t-2$ instead, the liquidity-provider inventory chain.
Estimation samples differ across rows with the lag depth. $^{*}$, $^{**}$, $^{***}$ mark $10\%$, $5\%$, $1\%$.}
\label{tab:flowchasing}
\medskip
\singlespacing
\footnotesize
\begin{tabular}{lcccc}
\toprule
 & Retail & Institutions & Foreigners & Other \\
\midrule
\multicolumn{5}{l}{\textit{Panel A. Net buying on the LETF tape, per won of fund AUM}} \\
Same-day NAV return & $-0.301^{***}$ & $+0.292^{***}$ & $+0.010$ & $-0.000$ \\
 & (0.028) & (0.028) & (0.007) & (0.004) \\
NAV return, $t-1$ & $+0.080^{**}$ & $-0.077^{*}$ & $-0.002$ & $-0.001$ \\
 & (0.040) & (0.041) & (0.009) & (0.002) \\
NAV return, $t-2$ & $+0.074^{**}$ & $-0.060^{*}$ & $-0.007$ & $-0.006^{**}$ \\
 & (0.032) & (0.032) & (0.005) & (0.003) \\
\addlinespace
Samsung stock, corr.\ with return (pre-launch) & $-0.639$ & $+0.388$ & $+0.514$ & $-0.184$ \\
\quad post-launch & $-0.748$ & $+0.650$ & $+0.455$ & $-0.510$ \\
SK Hynix stock, corr.\ with return (pre-launch) & $-0.644$ & $+0.614$ & $+0.489$ & $-0.705$ \\
\quad post-launch & $-0.783$ & $+0.633$ & $+0.494$ & $-0.631$ \\
\addlinespace
\multicolumn{5}{l}{\textit{Panel B. The same day split into its predetermined and its contemporaneous half}} \\
Overnight gap (predetermined) & $-0.296^{***}$ & $+0.303^{***}$ & $-0.008$ & $+0.000$ \\
 & (0.047) & (0.043) & (0.013) & (0.005) \\
Intraday move & $-0.315^{***}$ & $+0.280^{***}$ & $+0.033^{**}$ & $+0.001$ \\
 & (0.045) & (0.049) & (0.013) & (0.006) \\
\addlinespace
\multicolumn{5}{l}{\textit{Panel C. Longer horizons, and the primary market}} \\
 & \multicolumn{2}{c}{Estimate} & \multicolumn{2}{l}{} \\
Retail net buying on the trailing 5-day return & \multicolumn{2}{c}{$+0.003$} & \multicolumn{2}{l}{} \\
 & \multicolumn{2}{c}{(0.016)} & \multicolumn{2}{l}{} \\
Creation rate on the same-day return & \multicolumn{2}{c}{$+0.014$} & \multicolumn{2}{l}{} \\
 & \multicolumn{2}{c}{(0.029)} & \multicolumn{2}{l}{} \\
\quad on the return at $t-2$ & \multicolumn{2}{c}{$-0.222^{***}$} & \multicolumn{2}{l}{} \\
 & \multicolumn{2}{c}{(0.031)} & \multicolumn{2}{l}{} \\
Creation rate on retail net buying at $t-2$ & \multicolumn{2}{c}{$+0.659^{***}$} & \multicolumn{2}{l}{\footnotesize the provider-inventory chain} \\
 & \multicolumn{2}{c}{(0.063)} & \multicolumn{2}{l}{} \\
\midrule
Product-days estimated & \multicolumn{4}{c}{640--720} \\
Date clusters & \multicolumn{4}{c}{46} \\
\bottomrule
\end{tabular}
\end{table}

\begin{table}[H]
\centering
\caption{\captitle{The echo against the same-sector control group}
The table re-estimates the specification of Table~\ref{tab:echo} on the treated pair and the ten technology-peer controls of Appendix~\ref{app:inference}; the treated rows are identical to Table~\ref{tab:echo} by construction.
The peers' own post-launch cell carries a reversal roughly a quarter of the treated cell's, spillover together with their own low-loop-gain index pass-through, so this triple difference is a lower bound: it is $-0.48$, $-1.03$ and $-1.65$ on the next-day total against the large-cap group's $-0.79$, $-1.55$ and $-2.29$.
Standard errors are clustered by date.}
\label{tab:echo_tech}
\medskip
\singlespacing
\footnotesize
\setlength{\tabcolsep}{4.5pt}
\begin{tabular}{lccccccr}
\toprule
 & \multicolumn{3}{c}{Next-day (total)} & \multicolumn{3}{c}{Overnight leg} & \\
\cmidrule(lr){2-4}\cmidrule(lr){5-7}
 & SMH & QQQ & SPY & SMH & QQQ & SPY & Obs. \\
\midrule
Treated $\times$ post & $-0.83$ & $-1.76$ & $-3.92$ & $-0.44$ & $-1.05$ & $-2.88$ & 82 \\
 & (0.29) & (0.66) & (1.34) & (0.21) & (0.43) & (0.93) & \\
Treated $\times$ pre & $+0.01$ & $-0.03$ & $-0.07$ & $-0.02$ & $-0.06$ & $-0.07$ & 1,086 \\
 & (0.06) & (0.10) & (0.14) & (0.04) & (0.07) & (0.10) & \\
Tech peers $\times$ post & $-0.20$ & $-0.43$ & $-1.48$ & $-0.10$ & $-0.25$ & $-1.04$ & 410 \\
 & (0.17) & (0.33) & (0.68) & (0.09) & (0.19) & (0.33) & \\
Tech peers $\times$ pre & $+0.08$ & $+0.07$ & $+0.05$ & $-0.01$ & $-0.04$ & $-0.05$ & 5,430 \\
 & (0.04) & (0.07) & (0.10) & (0.02) & (0.04) & (0.06) & \\
\midrule
Triple difference & $-0.48$ & $-1.03$ & $-1.65$ & $-0.20$ & $-0.54$ & $-1.02$ & 7,008 \\
 & (0.17) & (0.40) & (0.67) & (0.09) & (0.21) & (0.35) & \\
\bottomrule
\end{tabular}
\end{table}

\begin{table}[H]
\centering
\caption{\captitle{The correction speed}
The table presents two readings of the correction speed per instrument and sample, the ratio of the $h{=}2$ to the $h{=}1$ cycle return and a minimum-distance fit to all nine profile coefficients.
The chained pricing error is an autoregression with coefficient $1-\theta\Pi_g$ (Proposition~\ref{prop:chaining}), so the horizon profile identifies $\theta$ given $\Pi_g$. The ratio identity is exactly free of $\theta$, so none of this feeds back into $\hat\Pi_g$.}
\label{tab:theta}
\medskip
\singlespacing
\small
\begin{tabular}{llrrrrr}
\toprule
Shock & Sample & $\hat\Pi_g$ & AR coef & $\hat\theta$ (AR) & $\hat\theta$ (min.\ dist.) & $\hat\theta\hat\Pi_g$ \\
\midrule
SMH & treated-post & 1.85 & -0.29 & 0.70 & 0.63 & 1.11 \\
SMH & DDD & 1.89 & -0.31 & 0.69 & 0.64 & 1.18 \\
QQQ & treated-post & 1.66 & -0.47 & 0.88 & 0.69 & 1.19 \\
QQQ & DDD & 1.68 & -0.48 & 0.88 & 0.69 & 1.23 \\
SPY & treated-post & 1.35 & -0.55 & 1.15 & 0.74 & 1.30 \\
SPY & DDD & 1.33 & -0.62 & 1.22 & 0.75 & 1.32 \\
\midrule
\multicolumn{7}{p{0.95\textwidth}}{\footnotesize AR coef $=1-\theta\Pi_g$, the ratio of the $h{=}2$ to the $h{=}1$ cycle return following a public shock (Prop.\ 22). $\hat\theta$(AR) $=(1-\text{AR})/\hat\Pi_g$; the minimum-distance column fits $(\Pi_g,\theta)$ to all nine profile coefficients ($h=1,2,3$ $\times$ overnight/session/total), weighted by inverse bootstrap variance. $1<\hat\theta\hat\Pi_g<2$ is the ringing region: damped, sign-alternating echoes.} \\
\bottomrule
\end{tabular}
\end{table}

\begin{table}[H]
\centering
\caption{\captitle{The volatility design across control sets}
The table presents the volatility difference-in-differences of Section~\ref{sec:volvol} under each control group: treated $\times$ period coefficients on the log absolute daily return, raw and residualized on rolling two-factor models with betas through $t-1$, with stock and date fixed effects.
The rows present the large-cap and technology-peer groups of Appendix~\ref{app:inference}, the eight KOSPI shares nearest the listing screen, and their union.
The columns present joint Wald $p$-values for zero pre-launch coefficients and the pooled treated $\times$ post estimate on residual volatility, with date-clustered $t$ and randomization rank.
Comparable controls restore parallel trends, and the estimate is $+55$ per cent against the broad union (exact $p = 0.033$) and $+18$ per cent, imprecise, against the technology peers alone.}
\label{tab:battery}
\medskip
\singlespacing
\footnotesize
\setlength{\tabcolsep}{4pt}
\begin{tabular}{lcccccc}
\toprule
 & & \multicolumn{3}{c}{Joint pre-trend $p$} & \multicolumn{2}{c}{Pooled DiD, $\log|\varepsilon|$} \\
\cmidrule(lr){3-5}\cmidrule(lr){6-7}
Control set & $N$ & raw & residual & no-prod.\ era & $b$ $(t)$ & RI rank ($p$) \\
\midrule
Never-eligible large caps & 10 & $<0.001$ & $0.007$ & $0.016$ & $+0.70$ $(+5.1)$ & 1 of 66 $(0.015)$ \\
Nearest eligibility screen & 8 & $<0.001$ & $0.175$ & $0.045$ & $-0.03$ $(-0.2)$ & 25 of 45 $(0.556)$ \\
Tech/semiconductor peers & 10 & $0.025$ & $0.326$ & $0.407$ & $+0.16$ $(+1.2)$ & 15 of 66 $(0.227)$ \\
Union of all sets & 23 & $<0.001$ & $0.234$ & $0.480$ & $+0.44$ $(+3.6)$ & 10 of 300 $(0.033)$ \\
\bottomrule
\end{tabular}

\end{table}

\begin{table}[H]
\centering
\caption{\captitle{Close-specific volatility}
The table presents difference-in-differences estimates, treated $\times$ post with stock and date fixed effects and date-clustered standard errors, for the log absolute close-to-close return, the log absolute open-to-open return, and their within-stock difference $z$, which nets out every volatility component common to the two marks \citep{frenchroll1986}.
The pre-trend columns present joint Wald $p$-values.
The rise loads on the two marks nearly equally, and the close-specific component is statistically zero. $^{*}$, $^{**}$, $^{***}$ mark $10\%$, $5\%$, $1\%$.}
\label{tab:closespec}
\medskip
\singlespacing
\small
\begin{tabular}{lcccc}
\toprule
 & \multicolumn{2}{c}{vs.\ large caps} & \multicolumn{2}{c}{vs.\ tech peers} \\
\cmidrule(lr){2-3}\cmidrule(lr){4-5}
Outcome & $b$ $(t)$ & pre-trend $p$ & $b$ $(t)$ & pre-trend $p$ \\
\midrule
Close-anchored, $\log|r^{cc}|$ & $+0.52^{***}$ $(+3.3)$ & $<0.001$ & $+0.22$ $(+1.4)$ & $0.024$ \\
Open-anchored, $\log|r^{oo}|$ & $+0.45^{***}$ $(+3.0)$ & $<0.001$ & $+0.08$ $(+0.6)$ & $0.157$ \\
Difference $z$ (close $-$ open) & $+0.07$ $(+0.4)$ & $0.453$ & $+0.14$ $(+0.8)$ & $0.731$ \\
\addlinespace
$z$, excluding July 2026 & $+0.15$ $(+0.6)$ & & $+0.25$ $(+1.0)$ & \\
RI rank of $z$ (66 pairs) & 19 of 66 & & 16 of 66 & \\
\bottomrule
\end{tabular}

\end{table}

\begin{table}[H]
\centering
\caption{\captitle{News-conditional aftershock volatility}
The table presents the next-day absolute return $|r_{t+1}|$ on the absolute pre-open shock $|\varepsilon_t|$ interacted with treated $\times$ post, estimated with stock and date fixed effects, lead $|\varepsilon_{t+1}|$ interactions as controls, and date-clustered standard errors.
The coefficient is the treated-specific change at the launch in how strongly yesterday's news size predicts today's movement: $+0.66$ ($t = 3.08$) against the large-cap group and $+0.81$ ($t = 3.09$) against the technology peers, both surviving exclusion of the crash month. The true pair ranks 1st and 3rd of 66, and no placebo date of 47 is as large. $^{*}$, $^{**}$, $^{***}$ mark $10\%$, $5\%$, $1\%$.}
\label{tab:aftershock}
\medskip
\singlespacing
\small
\begin{tabular}{lcccc}
\toprule
Shock instrument & \multicolumn{2}{c}{vs.\ large caps} & \multicolumn{2}{c}{vs.\ tech peers} \\
\midrule
$|$SMH$|$ & $+0.66^{***}$ & $(0.21)$ & $+0.81^{***}$ & $(0.26)$ \\
$|$QQQ$|$ & $+1.15^{***}$ & $(0.43)$ & $+1.42^{***}$ & $(0.50)$ \\
$|$SPY$|$ & $+1.13$ & $(0.69)$ & $+1.39^{*}$ & $(0.78)$ \\
\addlinespace
$|$SMH$|$, excluding July 2026 & $+0.55^{**}$ & $(0.25)$ & $+0.84^{***}$ & $(0.31)$ \\
RI: true-pair rank (66 pairs) & \multicolumn{2}{c}{1 of 66} & \multicolumn{2}{c}{3 of 66} \\
RI: placebo launch dates $\ge$ true & \multicolumn{2}{c}{0 of 47} & \multicolumn{2}{c}{---} \\
Joint pre-launch slope path $p$ & \multicolumn{2}{c}{0.77} & \multicolumn{2}{c}{0.72} \\
\bottomrule
\end{tabular}

\end{table}

\begin{table}[H]
\centering
\caption{\captitle{The migration into the reference windows}
The table presents difference-in-differences estimates, treated $\times$ post with contract or stock and date fixed effects and date-clustered standard errors, on shares of daily traded value in percentage points.
The futures outcomes are the 15:45 closing auction's share and the final-thirty-minute share, over April~1--June~30, 2026, on the treated pair and 22 Tier-A roots. The cash outcomes are the 15:20--15:30 auction's share and the final-thirty-minute continuous share, over May~4--June~30, on the treated pair and 91 Tier-A names.
The final column presents exact randomization inference over treated-pair relabelings. The futures closing auction carries the sharp result; the cash auction-share rise does not clear the pair permutation and is reported as a null. $^{*}$, $^{**}$, $^{***}$ mark $10\%$, $5\%$, $1\%$.}
\label{tab:quantity}
\medskip
\singlespacing
\footnotesize
\setlength{\tabcolsep}{4pt}
\begin{tabular}{lcccc}
\toprule
Share of session traded value (pp) & Treated pre & Treated post & DiD & RI rank ($p$) \\
\midrule
Futures: closing-auction (15:45) share & 0.28 & 1.36 & $+0.96^{***}$ & 4 of 276 $(0.014)$ \\
 & & & (0.09) & \\
Futures: final-30-minute share & 11.41 & 14.69 & $+1.89^{**}$ & 21 of 276 $(0.076)$ \\
 & & & (0.92) & \\
Cash: closing-auction share & 7.34 & 9.75 & $+0.95^{**}$ & 648 of 4278 $(0.151)$ \\
 & & & (0.38) & \\
Cash: final-30-minute continuous share & 6.01 & 6.68 & $-0.60^{*}$ & 812 of 4278 $(0.190)$ \\
 & & & (0.34) & \\
\bottomrule
\end{tabular}

\end{table}


\begin{table}[H]
\centering
\caption{\captitle{Six estimators of the closing-auction impact coefficient}
Each row presents the capital that displaces the closing price by one percent and the implied loop gain at the complex's own order on a one-percent day.
The model's $\Lambda_c$ requires signed net order flow, flow in the window where the complex trades, and displacement still present at the close. Row~(1), the Amihud ratio, fails the first, and row~(2), the square-root law, uses no data of ours; together they bracket the plausible range. Rows~(3)--(6) present the native tape at four venue-horizon combinations, and row~(6) is the only one satisfying all three properties and the estimator carried through Section~\ref{sec:quant}.}
\label{tab:impact_menu}
\medskip
\small
\begin{tabular}{lrrrr}
\toprule
Estimator of $\Lambda_c$ & \multicolumn{2}{c}{\KRW bn moving the close 1\%} & \multicolumn{2}{c}{implied $\ell$ at a 1\% move} \\
\cmidrule(lr){2-3}\cmidrule(lr){4-5}
 & Samsung & SK Hynix & Samsung & SK Hynix \\
\midrule
(1) daily Amihud{,} gross turnover & 1{,}690 & 2{,}023 & 0.02 & 0.10 \\
(2) square-root law{,} $c=1$ & 356 & 343 & 0.34 & 0.76 \\
(3) tape{,} mid-session $\to$ window end & 186 & 189 & 0.35 & 1.03 \\
(4) tape{,} mid-session $\to$ close & 200 & 199 & 0.39 & 0.99 \\
(5) tape{,} pre-close $\to$ block end & 2{,}204 & 2{,}064 & 0.18 & 0.36 \\
(6) tape{,} pre-close $\to$ close & 1{,}302 & 1{,}252 & 0.18 & 0.40 \\
\bottomrule
\end{tabular}
\end{table}


\begin{figure}[H]
\centering
\includegraphics[width=.92\textwidth]{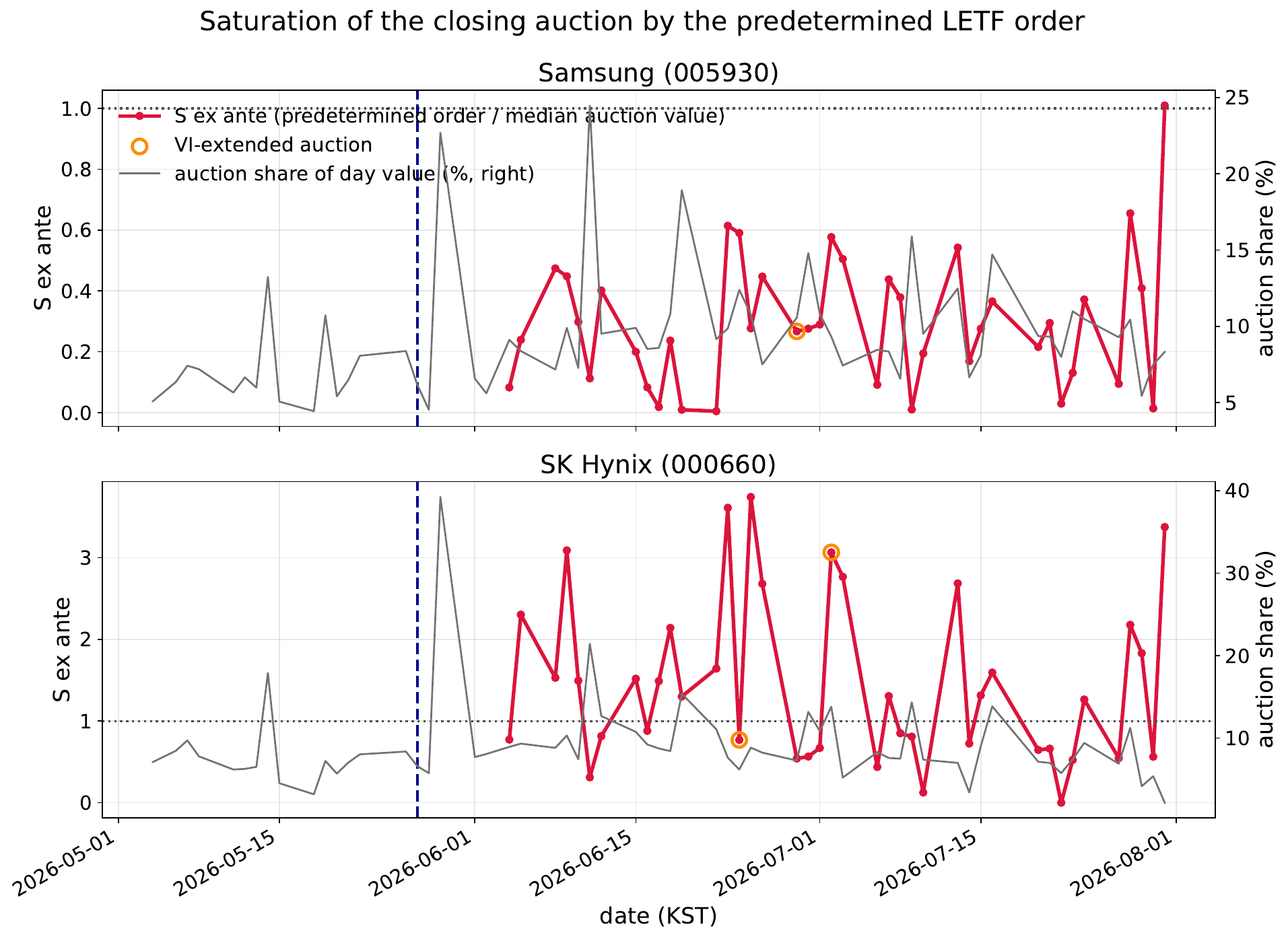}
\caption{\captitle{Auction saturation over time}
The figure plots the saturation ratio $S_t$ of equation~\eqref{eq:saturation} for SK~Hynix around the May~27, 2026 launch: near zero before, jumping at it, and pinned around one after (post-launch median $1.06$, maximum $7.50$, above one on 26 of 46 days).
Vertical lines mark the launch and the July~16 listing freeze.}
\label{fig:saturation}
\end{figure}

\subsection{Quantification details}

\begin{table}[H]
\centering
\caption{\captitle{The calibration on the endogenous impact curve}
The table reads the measured impact system through the concave-impact layer of Proposition~\ref{prop:curve} at its closed-form case $\varphi_c = 2$, per treated name at the $|r| = 2\%$ reference day.
Panel~A maps the fitted curve into the layer's primitives at the complex's own order size, at the pooled $\hat\psi_c$ of Section~\ref{sec:calibration} rather than at a capacity fitted here.
Panel~B presents the three loop gains a concave curve separates, touch, average at the complex's own size, and marginal, with four readings of the overshoot: the concave close at zero capacity and at $\hat\psi_c$ --- the pricing the calibration inverts and the replay uses --- the linear chain at $\hat\psi_c$ as a diagnostic, and the pooled measured $\hat\Pi_g$ they are held against.
Panel~C prices in closed form the regimes the linear model cannot, the sustained-ringing amplitude and the no-news bootstrap price of equation~\eqref{eq:amplitude}, evaluated at $\theta = 1$ where the news-free skeleton is defined, so its ring crossing is $(\bar\rho+\psi_c)/2$; Proposition~\ref{prop:ladder} locates the same threshold at the $\theta$-general $(\bar\rho+\psi_c)(1-\theta/2)$, which at the measured $\hat\theta = 0.88$ sits slightly higher.}
\label{tab:g4calibration}
\medskip
\small
\begin{tabular}{lcc}
\toprule
 & SK Hynix & Samsung \\
\midrule
\multicolumn{3}{l}{\textit{Panel A: layer parameters at the operating point ($\varphi_c = 2$, $|r| = 2\%$)}} \\
local exponent at own size $\hat\delta$ (curve, within-$\Delta$) & 0.732 & 0.732 \\
fringe/base concession ratio $x$ & 0.58 & 0.58 \\
print concession $\bar p = \ell_a|r|$ (bp) & 117 & 44 \\
crossover concession $p^{*}$ (bp) & 201 & 75 \\
crossover size $Q^{*}$ (KRW bn) & 1,749 & 365 \\
arbitrage capacity $\hat\psi_c$ (stage~02, pooled; \emph{not} fitted here) & 0.260 & 0.466 \\
\midrule
\multicolumn{3}{l}{\textit{Panel B: the three doses, and the overshoot (untargeted throughout)}} \\
 & & Targeted \\
touch dose $\ell$ (linear core) & 0.92 & 0.35 \\
average dose at own size $\ell_a$ (stage~07 ladder) & 0.58 & 0.22 \\
marginal dose at own size $\ell_m = \hat\delta\,\ell_a$ & 0.43 & 0.16 \\
concave-close $\Pi_g$ at $\psi_c = 0$ \hfill (no) & 2.21 & 1.34 \\
concave-close $\Pi_g$ at $\hat\psi_c$ \hfill (no) & 2.06 & 1.27 \\
linear-chain $\Pi_g$ at $\hat\psi_c$ (the replay's loading) \hfill (no) & 4.80 & 1.30 \\
pooled measured $\hat\Pi_g$ (ratio identity; QQQ) \hfill (no) & \multicolumn{2}{c}{1.66} \\
\midrule
\multicolumn{3}{l}{\textit{Panel C: regimes in closed form}} \\
ring crossing $\ell > (\bar\rho+\psi_c)/2$ / pole $\ell > \bar\rho+\psi_c$ \hfill ($\theta = 1$) & 0.37 / 0.74 & 0.47 / 0.94 \\
ringing at $\psi_c = 0$: $e^{*} = (2\ell - \bar\rho)p^{*}$ (bp) & 275 & 16 \\
bootstrap print at $\psi_c = 0$: $r^{*} = (\ell - \bar\rho)p^{*}$ (bp) & 89 & -- \\
at $\hat\psi_c$: $e^{*} = (2\ell - \bar\rho - \psi_c)p^{*}$, $r^{*} = (\ell - \bar\rho - \psi_c)p^{*}$ (bp) & 222, 37 & --, -- \\
\quad at across-$\Delta$ $\hat\delta = 0.53$: $e^{*}$, $r^{*}$ (bp; $\psi_c = 0$) & 148, 70 & 19, 8 \\
\bottomrule
\end{tabular}
\end{table}


\begin{table}[H]
\centering
\caption{\captitle{The opening auction against the closing auction}
Panel~A measures $\kappa$, the opening auction's depth relative to the closing auction, as the ratio of the two auctions' traded value. The 09:00 bar mixes the auction with the first continuous minute (Appendix~\ref{app:krx}), so it bounds the auction from above and the headline nets the first minute out at the 09:01 rate; every reading is an upper bound, the conservative direction for the two-print counterfactuals of Table~\ref{tab:schedule}.
Panel~B runs each leg separately in logs against the ten never-eligible large caps, with stock and date fixed effects and date-clustered standard errors: the opening auction does not move across the launch while the closing auction grows two fifths, so counterfactuals that move the order are calibrated on the pre-launch window.
The tape runs from 2026-05-04 to 2026-06-30 on 926 issues.}
\label{tab:kappa}
\medskip
\singlespacing
\small
\begin{tabular}{lrr}
\toprule
 & Samsung & SK Hynix \\
\midrule
\multicolumn{3}{l}{\textit{Panel A. $\kappa = W^{\text{open}}/W^{\text{close}}$, median of daily ratios}} \\
pre-launch $^\dagger$ & 0.75 & 0.59 \\
full sample & 0.39 & 0.36 \\
post-launch & 0.30 & 0.28 \\
\quad trading days & 15 / 59 / 44 & 15 / 56 / 41 \\
\quad upper bound, whole 09:00 bar (full sample) & 0.50 & 0.48 \\
\midrule
\multicolumn{3}{l}{\textit{Panel B. Which auction moved: treated versus Tier-A controls, log level}} \\
 & $\beta^{\text{DiD}}$ & $t$ \\
opening auction value & +0.089 & +0.82 \\
closing auction value & +0.479 & +8.13 \\
continuous session value & +0.404 & +6.10 \\
$\kappa$ & -0.388 & -3.26 \\
\bottomrule
\end{tabular}
\end{table}


\begin{table}[H]
\centering
\caption{\captitle{The attendance response of the closing auction}
The table presents the attendance visibility test for Proposition~\ref{prop:attendance}, on the treated pair and the ten never-eligible large caps with stock and date fixed effects and date-clustered standard errors. $S$ is the predetermined saturation $|K^{\text{spot}}_{t-1}\hat r_t| / \text{med}(W^{\text{auc}})_{t-1}$, with a lagged denominator and $\hat r$ ending at 15:10, so outcome and regressor share no price mark.
The two outcomes are traded value, which the mandated order contaminates, and the effective spread, which it does not. They agree at the intensive margin, where a higher loop gain brings more value and a tighter spread, and part at the extensive one, where value jumps two fifths across the launch while the spread does not move.
The depth response is therefore bounded rather than pinned, and Table~\ref{tab:schedule} prices the two-print design at both ends. The positive spread reading is controls-dependent, and the close-bar spread fields are partly degenerate (Appendix~\ref{app:krx}).}
\label{tab:attendance}
\medskip
\singlespacing
\small
\begin{tabular}{lrr}
\toprule
 & $\beta$ & ($t$) \\
\midrule
\multicolumn{3}{l}{\textit{Panel A. The extensive margin: treated $\times$ post}} \\
closing auction, traded value & +0.444 & (+7.89) \\
closing auction, effective spread & +0.129 & (+1.57) \\
\quad memo: opening auction, traded value & +0.089 & (+0.82) \\
\quad memo: continuous session, traded value & +0.404 & (+6.10) \\
\midrule
\multicolumn{3}{l}{\textit{Panel B. The intensive margin: the predetermined dose $S$}} \\
closing auction, traded value & +0.226 & (+3.73) \\
\quad + $|\hat r|$ and lagged auction size & +0.047 & (+0.74) \\
closing auction, effective spread & +0.256 & (+1.15) \\
\quad + $|\hat r|$ and lagged auction size & +0.076 & (+0.33) \\
\midrule
\multicolumn{3}{l}{\textit{Panel C. As a depth response, $d\log\Lambda_c/dS$}} \\
extensive margin, via traded value & -0.573 & -- \\
extensive margin, via spread & +0.166 & -- \\
intensive margin, via traded value & -0.047 & -- \\
intensive margin, via spread & +0.076 & -- \\
\quad memo: $d\log\hat\rho^{\,g}/d\log W$ (stage 47, across sizes) & -0.092 & -- \\
\bottomrule
\end{tabular}
\end{table}

\begin{table}[H]
\centering
\caption{\captitle{The schedule counterfactuals priced on both treated names}
The table runs Table~\ref{tab:schedule}'s three timing levers for Samsung as well as SK~Hynix, columns and method as there, with the excess volatility and the mean pricing error restored, the mandate-tracking-error column omitted, zero for every design here, and the manufacturing leverage carried under its symbol $\bar A$.
Panel~A is priced at the design's most favorable, as in the main table: dispersal worsens both names monotonically in $m$, and the smaller complex deteriorates more slowly.}
\label{tab:schedule_both}
\medskip
\singlespacing
\scriptsize
\setlength{\tabcolsep}{3pt}
\begin{tabular}{lrrrrrr@{\hskip 14pt}rrrrrr}
\toprule
 & \multicolumn{6}{c}{Samsung} & \multicolumn{6}{c}{SK Hynix} \\
\cmidrule(lr){2-7}\cmidrule(lr){8-13}
 & $\bar\ell$ & $\bar A$ & vol & excess & mean $|e|$ & holder loss & $\bar\ell$ & $\bar A$ & vol & excess & mean $|e|$ & holder loss \\
 &  &  & (\%) & (pp) & (\%) & (\KRW bn) &  &  & (\%) & (pp) & (\%) & (\KRW bn) \\
\midrule
\multicolumn{13}{l}{\textit{Benchmarks}} \\
status quo$^\dagger$ & 0.219 & 0.34 & 115.3 & 13.3 & 0.33 & 79 & 0.583 & 1.77 & 136.7 & 36.6 & 0.99 & 1{,}380 \\
no loop ($\ell\to0$) & 0.000 & 0.00 & 102.0 & 0.0 & 0.00 & 0 & 0.000 & 0.00 & 100.0 & 0.0 & 0.00 & 0 \\
\addlinespace
\multicolumn{13}{l}{\textit{Panel A. Staggered reference prints: $m$ equal funds, each rebalancing at its own intraday moment against its own print}} \\
$m = 2$ & 0.132 & 0.34 & 118.9 & 16.9 & 0.40 & 98 & 0.351 & 1.77 & 146.9 & 46.9 & 1.21 & 1{,}858 \\
$m = 3$ & 0.098 & 0.34 & 121.2 & 19.2 & 0.44 & 110 & 0.261 & 1.77 & 152.0$^{\S}$ & -- & -- & 2{,}098$^{\S}$ \\
$m = 5$ & 0.067 & 0.34 & 124.4 & 22.4 & 0.49 & 128 & 0.180 & 1.77 & 155.7$^{\S}$ & -- & -- & 2{,}235$^{\S}$ \\
\addlinespace
\multicolumn{13}{l}{\textit{Panel B. The rebalancing interval: the mandate re-formed every $T$ days, on the cumulative $T$-day return}} \\
$T=2$ & 0.296 & 0.34 & 117.1 & 14.0 & 0.47 & -- & 0.675 & 1.77 & 143.2 & 41.0 & 1.36 & -- \\
\addlinespace
\multicolumn{13}{l}{\textit{Panel C. Two reference prints: the order split between the opening and closing auctions}} \\
equal depth & 0.261 & 0.34 & 109.7 & 7.7 & 0.21 & 46 & 0.665 & 1.77 & 118.2 & 18.2 & 0.43 & 497 \\
endogenous attendance & 0.285 & 0.34 & 110.3 & 8.3 & 0.23 & 51 & 0.738 & 1.77 & 120.5 & 20.5 & 0.52 & 592 \\
pre-launch depth & 0.285 & 0.34 & 110.4 & 8.4 & 0.24 & 53 & 0.736 & 1.77 & 121.0 & 21.0 & 0.56 & 635 \\
\bottomrule
\end{tabular}
\end{table}

\begin{table}[H]
\centering
\caption{\captitle{The reference counterfactuals priced on both treated names}
The table runs Table~\ref{tab:reference}'s execution split and averaged reference prices for Samsung as well as SK~Hynix, columns and method as there, with the excess volatility and the mean pricing error restored, the mandate-tracking-error column omitted, and the manufacturing leverage carried under its symbol $\bar A$.
The measured split differs sharply across the pair, $\hat\varphi = 0.76$ for SK~Hynix against $0.44$ for Samsung, so the same counterfactual $\varphi$ is a cut for one name and close to the status quo for the other.}
\label{tab:reference_both}
\medskip
\singlespacing
\scriptsize
\setlength{\tabcolsep}{2pt}
\begin{tabular}{lrrrrrr@{\hskip 8pt}rrrrrr}
\toprule
 & \multicolumn{6}{c}{Samsung} & \multicolumn{6}{c}{SK Hynix} \\
\cmidrule(lr){2-7}\cmidrule(lr){8-13}
 & $\bar\ell$ & $\bar A$ & vol & excess & mean $|e|$ & holder loss & $\bar\ell$ & $\bar A$ & vol & excess & mean $|e|$ & holder loss \\
 &  &  & (\%) & (pp) & (\%) & (\KRW bn) &  &  & (\%) & (pp) & (\%) & (\KRW bn) \\
\midrule
\multicolumn{13}{l}{\textit{Benchmarks}} \\
status quo$^\dagger$ & 0.219 & 0.34 & 115.3 & 13.3 & 0.33 & 79 & 0.583 & 1.77 & 136.7 & 36.6 & 0.99 & 1{,}380 \\
no loop ($\ell\to0$) & 0.000 & 0.00 & 102.0 & 0.0 & 0.00 & 0 & 0.000 & 0.00 & 100.0 & 0.0 & 0.00 & 0 \\
\addlinespace
\multicolumn{13}{l}{\textit{Panel A. The execution split $\varphi$: two auctions of the same depth, $\varphi$ into the reference auction}} \\
$\varphi = 1.00$ & 0.22/0.00 & 0.34 & 115.3 & 13.3 & 0.33 & 79 & 0.58/0.00 & 1.77 & 136.7 & 36.6 & 0.99 & 1{,}380 \\
$\varphi = 0.75$ & 0.18/0.08 & 0.34 & 117.8 & 15.8 & 0.50 & 160 & 0.47/0.21 & 1.77 & 146.7 & 46.7 & 1.46 & 2{,}908 \\
$\varphi = 0.50$ & 0.13/0.13 & 0.34 & 117.5 & 15.6 & 0.55 & 211 & 0.35/0.35 & 1.77 & 146.4 & 46.4 & 1.61 & 3{,}702 \\
$\varphi = 0.25$ & 0.08/0.18 & 0.34 & 116.1 & 14.1 & 0.57 & 252 & 0.21/0.47 & 1.77 & 142.2 & 42.1 & 1.62 & 4{,}150 \\
$\varphi = 0.00$ & 0.00/0.22 & 0.34 & 111.0 & 9.1 & 0.47 & 253 & 0.00/0.58 & 1.77 & 125.2 & 25.2 & 1.23 & 3{,}411 \\
\addlinespace
\multicolumn{13}{l}{\textit{Panel B. Reference averaging: the mandate settles on an average of several prints}} \\
same-day average (open, close)$^{\ddagger}$ & 0.109 & 0.22 & 122.1 & 20.1 & 0.93 & 309 & 0.291 & 1.49 & \multicolumn{4}{c}{\textit{no damped eq.}} \\
MA(3) settlement mark & 0.073 & 0.11 & 107.5 & 5.5 & 0.16 & 11 & 0.194 & 0.59 & 116.2 & 16.2 & 0.51 & 237 \\
MA(5) settlement mark & 0.044 & 0.07 & 105.6 & 3.6 & 0.11 & 5 & 0.117 & 0.35 & 110.7 & 10.7 & 0.36 & 161 \\
\bottomrule
\end{tabular}
\end{table}

\begin{table}[H]
\centering
\caption{\captitle{The calibration in full}
The table presents four panels mirroring Table~\ref{tab:calibration}.
Panel~A measures the loop gain as a ladder: tier~0 the daily Amihud proxy; tier~1, the one carried forward, the size-aware closing-impact curve at each day's own order and venue split; tier~2 marginal rather than average ($\times\hat\delta$); and tier~3 ($\times\hat\tau$, the reverting component alone), which is defined for SK~Hynix and dashed for Samsung, whose order reaches a participation at which $\hat\tau$ is still negative.
Panel~B carries the one calibrated parameter by both routes: the concave pooled inversion $\hat\tau = 8{,}772$ of record (Appendix~\ref{app:calibration}), with its bootstrap interval, the profile bound, and the share of draws at the boundary, and the linear-close inversion $\hat\tau = 20{,}498$ retained as a legacy diagnostic; the row $\hat\psi_c^{\,i} = \hat\tau\Lambda_c^i$ turns one sector-wide capacity into two per-name ones.
Panel~C presents the restrictions the calibration never targeted, the per-name model overshoot against the pooled measured $\hat\Pi_g$, the leg ratio under all three instruments, and the depth and loop gain that would reproduce $\hat\Pi_g$ exactly; its per-name inversion rows are a sensitivity of the map, not a per-name test.
Panel~D presents what the self-reinforcement explains. Its volatilities are on the price-level convention of the shock-recovery replay and run about two points above the log convention the policy tables use; its per-unit row marks both replays at the recovered fundamental rather than at each one's own displaced close, which is the basis Section~\ref{sec:effect} reports.}
\label{tab:calibration_full}
\medskip
\singlespacing
\scriptsize
\begin{tabular}{lrr}
\toprule
 & Samsung & SK Hynix \\
\midrule
\multicolumn{3}{l}{\textit{Panel A. The measured loop gain $\ell = \Lambda_c K$, tier by tier}}\\
(0) Figure~\ref{fig:dose} proxy: Amihud $\times$ $K$ & 0.049 & 0.196 \\
(1) tape impact at own size and venue $^\dagger$ & 0.219 & 0.583 \\
\quad median day / peak day & 0.24 / 0.28 & 0.61 / 0.79 \\
\quad flow-weighted mean (weights $\propto K_t$) & 0.235 & 0.608 \\
(2) $\times\,\hat\delta$ (marginal) & 0.108 & 0.288 \\
(3) $\times\,\hat\tau$ (reverting component only) & -- & 0.026 \\
\quad $\hat\tau$ at the median participation & -0.60 at $\pi = 0.10$ & +0.09 at $\pi = 0.58$ \\
\midrule
\multicolumn{3}{l}{\textit{Panel B. The calibrated capacity: one pooled $\hat\tau$, two routes}}\\
$\hat\tau$, concave pooled inversion (calibration of record) & \multicolumn{2}{c}{8{,}772} \\
\quad bootstrap 95\% interval (concave draws) & \multicolumn{2}{c}{[1{,}605, 410{,}559]} \\
\quad range within one s.e.\ of the best untargeted fit & \multicolumn{2}{c}{$\hat\tau \le 16{,}692$, open at the low end} \\
\quad share of bootstrap draws at the concave boundary & \multicolumn{2}{c}{40\%} \\
$\hat\psi_c^{\,i} = \hat\tau\,\Lambda_c^i$ & 0.466 & 0.260 \\
\quad zero-capacity pooled concave $\Pi_g$ (ceiling) vs measured & \multicolumn{2}{c}{1.77 vs 1.66} \\
$\hat\tau$, linear inversion (legacy diagnostic) & \multicolumn{2}{c}{20{,}498} \\
\quad legacy bootstrap 95\% interval / boundary share & \multicolumn{2}{c}{[8{,}232, 138{,}574]; 0\%} \\
\quad share of draws with $\hat\Pi_g \le 1$ (no overshoot) & \multicolumn{2}{c}{6\%} \\
\quad $\hat\tau$ at the tape tier (linear / concave) & \multicolumn{2}{c}{20{,}498 / 8{,}772} \\
\quad $\hat\tau$ at the marginal tier (linear / concave) & \multicolumn{2}{c}{10{,}873 / 0$^{b}$} \\
\quad $\hat\tau$ at the amihud tier (linear / concave) & \multicolumn{2}{c}{0$^{b}$ / 0$^{b}$} \\
\midrule
\multicolumn{3}{l}{\textit{Panel C. The untargeted moments, and the solved repairs}}\\
model $\Pi_g$ at $\hat\psi_c$, linear chain (per name, untargeted) & 1.30 & 4.80 \\
\quad zero-capacity ceiling $m = 1/(1-\bar\ell)$ & 1.28 & 2.40 \\
model $\Pi_g$ at $\hat\psi_c$, concave close (per name, untargeted) & 1.27 & 2.06 \\
pooled measured $\hat\Pi_g$ (QQQ) & \multicolumn{2}{c}{1.66 [0.84, 3.27]} \\
next-session\,:\,overnight leg, data / $\hat\Pi_g-1$ (QQQ) & \multicolumn{2}{c}{0.68 / 0.69} \\
next-session\,:\,overnight leg, data / $\hat\Pi_g-1$ (SMH) & \multicolumn{2}{c}{0.88 / 0.90} \\
next-session\,:\,overnight leg, data / $\hat\Pi_g-1$ (SPY) & \multicolumn{2}{c}{0.37 / 0.38} \\
$\ell$ reproducing $\hat\Pi_g$ at $\hat\psi_c$ (linear route) & 0.375 & 0.293 \\
\quad implied $K$ multiple at $\hat\delta = 0.73$ & 2.09 & 0.39 \\
\quad worldwide $K$ required vs actual (\KRW tn) & 18 vs 8.5 & 17 vs 42.4 \\
measured $\hat\rho^{\,g}$ (curvature evidence, not a parameter) & \multicolumn{2}{c}{0.43 [0.41, 0.46]} \\
\midrule
\multicolumn{3}{l}{\textit{Panel D. What the loop explains at $(\ell^{\mathrm{meas}}, \hat\psi_c)$}}\\
realized annualized volatility (\%) & 117.5 & 138.9 \\
counterfactual at $\ell=0$ (\%) & 104.3 & 102.2 \\
\quad excess (pp) & 13.2 & 36.7 \\
holder loss (\KRW bn) & 79 & 1{,}380 \\
\KRW 10{,}000 becomes at the fundamental, vs no loop & 5{,}971 vs 6{,}740 & 4{,}790 vs 6{,}554 \\
\bottomrule
\end{tabular}
\end{table}

\begin{table}[H]
\centering
\caption{\captitle{The measured impact system}
The table presents the output of the estimator of Appendix~\ref{app:impact} on 626{,}571 signed windows of the native cash tape.
Panel~A identifies the curvature of the impact curve of equation~\eqref{eq:acurve} three ways and reports the invariance restriction as an over-identification test; Panel~B measures permanence, which is negative at small size, where flow is information, and crosses zero below the automaton's own order; Panel~C prices the closing auction against the continuous session at matched size and clock, and the calibration reads $\hat\rho = 0.433$ off its $\varphi = 1.00$ row, the last row being an independent route to the same object that does not use the curve.
Intervals are 500 date-block bootstrap draws with re-binning inside each draw.
The reconciliation ladder is in Table~\ref{tab:calibration_full}, Panel~A, flow-weighted over the post-launch days.}
\label{tab:impact}
\medskip
\singlespacing
\small
\setlength{\tabcolsep}{4pt}
\begin{tabular}{lr}
\toprule
\multicolumn{2}{l}{\emph{A. The impact curve}\quad $R_w/\sigma_\Delta = Y_i\,s_w\,\pi_w^{\delta}\,(\Delta/390)^{\zeta}$} \\
$\hat\delta$ (within-$\Delta$, across-$\pi$) & 0.732 [0.676, 0.762] \\
$\hat\zeta$ (clock) & 0.035 [0.009, 0.085] \\
$\hat\delta$ implied by $\zeta$ (across-$\Delta$) & 0.535 [0.509, 0.585] \\
$\hat\delta$, $\Delta$ dummies (robustness) & 0.739 \\
$\hat Y$, Samsung & 2.36 [1.99, 2.93] \\
$\hat Y$, SK Hynix & 2.64 [2.17, 3.34] \\
$\hat Y$, tier-A median (untreated reference) & 1.97 \\
bins / dropped ($m_k\leq0$) / $R^2$ & 782 / 11 / 0.859 \\
\multicolumn{2}{l}{\quad invariance test $\zeta=0$: \textbf{zeta = 0 REJECTED -- extrapolate only within the late block}} \\
\midrule
\multicolumn{2}{l}{\emph{B. Permanence} ($\tau=1-P/I$ at the next open, grid late20)} \\
\quad $\hat\tau$ at $Q=$ KRW 10bn & -1.84 \\
\quad $\hat\tau$ at $Q=$ KRW 100bn & -0.324 \\
\quad $\hat\tau$ at $Q=$ KRW 580bn & 0.260 \\
\quad $\tau_{adv}$ = Hs / effective spread (Samsung / SK Hynix) & 0.373 / 0.594 \\
\midrule
\multicolumn{2}{l}{\emph{C. The auction on the curve}\quad $\rho^g(W)=e^{b}(W/V_\Delta)^{c}$, gross, matched 10-min clock} \\
\quad $b$ (auction intercept) / $c$ (auction slope) & -0.766 / -0.092 \\
\quad $\rho^g$ at the median $W^{auc}$ & 0.45 [0.42, 0.48] \\
\quad $\rho^g$ at $\varphi F$, $\varphi=0.25$ & 0.492 \\
\quad $\rho^g$ at $\varphi F$, $\varphi=0.44$ & 0.467 \\
\quad $\rho^g$ at $\varphi F$, $\varphi=0.76$ & 0.444 \\
\quad $\rho^g$ at $\varphi F$, $\varphi=1.00$ & 0.433 \\
\quad $\rho^{inv}$ (transient variance, equal size) & 0.482 \\
\midrule
\multicolumn{2}{p{0.78\textwidth}}{\footnotesize Implements \texttt{notes/impact\_estimator\_design\_2026-07-31.md}; the specification, the three curvature readings and the venue comparison are described in Appendix~\ref{app:impact}. Windows with signed coverage below 90\%, missing marks, or fewer than $\Delta/2$ traded bars are dropped; $\pi$ and every return leg are winsorised at 1\%/99\% within issue$\times$grid, and bins hold at least 200 observations ($\leq20$ per liquidity$\times\Delta$ cell). Confidence intervals are 500 date-block bootstrap draws (seed 20260725), re-binning inside every draw. The 96 VI-extended issue-days leave the venue comparison: a 12-minute interrupted auction is not a 10-minute auction. Empirical signed support reaches $\pi=1.33$ against $\pi=1.07$ implied by $F$ = KRW 580bn, an extrapolation factor of 0.808$\times$.} \\
\bottomrule
\end{tabular}
\end{table}


\begin{table}[H]
\centering
\caption{\captitle{The overshoot factor under each shock instrument}
The table applies the ratio identity of Proposition~\ref{prop:chaining}(iii) to the treated post-launch cell of Table~\ref{tab:echo} under each pre-open U.S.\ instrument. $\hat\Pi_g$ is the one-cycle reversal over its overnight leg, exactly free of $\theta$, and $\hat\theta$ is read separately off the decay of successive legs (Table~\ref{tab:theta}).
The final column re-identifies the calibrated capacity under each instrument. $\dagger$ marks the benchmark: QQQ is the middle estimate, and its implied $\hat\theta = 0.88$ lies strictly inside the unit interval, unlike the broad-market instrument's $1.15$.
Intervals are from 800 date-block bootstrap draws.}
\label{tab:pi0_instruments}
\medskip
\singlespacing
\small
\begin{tabular}{lrrrrr}
\toprule
Instrument & one-cycle & overnight & $\hat\Pi_g$ [95\% CI] & $\hat\theta$ & implied $\hat\psi_c$ \\
 & reversal & leg & (ratio identity) & & (SK~Hynix) \\
\midrule
QQQ$^\dagger$ & -1.78 & -1.07 & 1.66 [0.84, 3.27] & 0.88 & 0.608 \\
SMH & -0.84 & -0.45 & 1.85 [0.93, 6.07] & 0.70 & 0.492 \\
SPY & -3.98 & -2.94 & 1.35 [0.76, 1.98] & 1.15 & 1.072 \\
\bottomrule
\end{tabular}
\end{table}

\end{document}